\documentclass[manuscript]{acmart}

\usepackage{booktabs} 

\usepackage{graphicx}
\usepackage{epsfig}
\usepackage{amssymb}
\usepackage{amsmath}
\usepackage{amsfonts}
\usepackage{array}
\usepackage{booktabs}
\usepackage{algorithm}
\usepackage{algorithmic}
\usepackage{multirow}
\usepackage{multicol}
\usepackage{subfigure}
\usepackage{color}
\usepackage{tcolorbox}
\usepackage{enumerate}
\usepackage{diagbox}

\acmJournal{TKDD}
\acmVolume{9}
\acmNumber{4}
\acmArticle{39}
\acmYear{2010}
\acmMonth{3}
\acmArticleSeq{11}

\setcopyright{usgovmixed}

\acmDOI{0000001.0000001}

\received{xxxx}
\received{xxxx}
\received[accepted]{xxxx}

\begin{document}
\title{Recurrent Meta-Structure for Robust Similarity Measure in Heterogeneous Information Networks}

\author{Yu Zhou}
\affiliation{%
  \institution{School of Computer Science and Technology, Xidian University}
  \city{Xi'an}
  \state{Shaanxi}
  \country{China}}
\email{peterjone85@hotmail.com}

\author{Jianbin Huang}
\authornote{This is the corresponding author}
\affiliation{%
  \institution{School of Computer Science and Technology, Xidian University}
  \city{Xi'an}
  \state{Shaanxi}
  \country{China}
}
\email{jbhuang@xidian.edu.cn}

\author{Heli Sun}
\affiliation{%
 \institution{Department of Computer Science and Technology, Xi'an Jiaotong University}
 \city{Xi'an}
 \state{Shaanxi}
 \country{China}}
\email{hlsun@mail.xjtu.edu.cn}

\author{Yizhou Sun}
\affiliation{%
  \institution{Deptartment of Computer Science, University of California at Los Angeles}
  \city{Los Angeles}
  \state{California}
  \country{USA}
}
\email{yzsun@cs.ucla.edu}

\begin{abstract}
  Similarity measure as a fundamental task in heterogeneous information network analysis has been applied to many areas, e.g., product recommendation, clustering and Web search.
  Most of the existing metrics depend on the meta-path or meta-structure specified by users in advance. These metrics are thus sensitive to the pre-specified meta-path or meta-structure.
  In this paper, a novel similarity measure in heterogeneous information networks, called Recurrent Meta-Structure-based Similarity (RMSS), is proposed.
  The recurrent meta-structure as a schematic structure in heterogeneous information networks provides a unified framework to integrate all of the meta-paths and meta-structures.
  Therefore, RMSS is robust to the meta-paths and meta-structures. We devise an approach to automatically constructing the recurrent meta-structure.
  In order to formalize the semantics, the recurrent meta-structure is decomposed into several recurrent meta-paths and recurrent meta-trees,
  and we then define the commuting matrices of the recurrent meta-paths and meta-trees. All of the commuting matrices of the recurrent meta-paths and meta-trees are combined according to different weights.
  Note that the weights can be determined by two kinds of weighting strategies: local weighting strategy and global weighting strategy.
  As a result, RMSS is defined by virtue of the final commuting matrix.
  Experimental evaluations show that the existing metrics are sensitive to different meta-paths or meta-structures and that the proposed RMSS outperforms the existing metrics in terms of ranking and clustering tasks.
\end{abstract}

%
%

\ccsdesc[500]{Information systems~Data mining}
\ccsdesc[300]{Mathematics of computing~Graph algorithms}

%
%

\keywords{Heterogeneous Information Network, Similarity, Schematic Structure, Meta Path, Meta Structure}


\maketitle

\renewcommand{\shortauthors}{Yu Zhou, Jianbin Huang, Heli Sun, Yizhou Sun.}

\section{Introduction}\label{sec:introduction}

As well known, networks can be used to model many real systems such as biological systems and social medium.
As a result, network analysis becomes a hot research topic in the field of data mining.
Many researchers are concerned with information networks with single-typed components, the kind of which is called homogeneous information network.
However, the real information networks usually consist of interconnected and multi-typed components.
This kind of information networks is generally called Heterogeneous Information Networks (HIN).
Mining heterogeneous information networks has attracted many attentions of the researchers.

Measuring the similarity between objects plays fundamental and essential roles in heterogeneous information network mining tasks.
Most of the existing metrics depend on user-specified meta-paths or meta-structures.
For example, PathSim \cite{SHYYW:2011} and Biased Path Constrained Random Walk (BPCRW) \cite{LC:2010a,LC:2010b}
take a meta-path specified by users as input,
and Biased Structure Constrained Subgraph Expansion (BSCSE) \cite{HZCSML:2016} takes a meta-structure specified by users as input.
We investigate these metrics in depth, and discover that they are sensitive to the pre-specified meta-paths or meta-structures in some degree.
The sensitivity requires that the users must know how to select an appropriate meta-path or meta-structure.
Obviously, it is quite difficult for a non-proficient users to make the selection.
For example, a biological information network may contain many different types of objects \cite{CDJWZ:2010,FDSCSB:2016}.
It is hard for a new user to know which meta-paths or meta-structures are appropriate.
In addition, the meta-paths can only capture biased and relatively simple semantics according to literature \cite{HZCSML:2016}.
Therefore, the authors proposed the meta-structure in order to capture more complex semantics.
In fact, the meta-structure can only capture biased semantics as well.
The meta-paths and meta-structures are essentially two kinds of schematic structures.

In this paper, we are concerned with the robust semantic-rich similarity between objects in heterogeneous information networks.
We are inspired by the construction of the subtree pattern proposed in \cite{NPEKK:2011}. In essence, the subtree pattern is a quasi spanning tree of a graph.
The difference between the traditional spanning tree and the subtree pattern lies in that nodes can be re-visited in the process of traversing the graphs.
That means that we can construct a schematic structure by repetitively visiting the object types in the process of traversing the network schema of the HIN.
Obviously, this schematic structure, called  Recurrent Meta Structure (RecurMS), can be constructed automatically.
In addition, it can capture rich semantics because it is composed of many recurrent meta-paths and recurrent meta-trees.

Both the meta-path and meta-structure are essentially two kinds of composite relations because they are composed of object types with different layer labels.
The commuting matrices of the meta-path and meta-structure are employed to extract the semantics encapsulated in them.
In essence, the proposed RecurMS has the same property as the meta-path and meta-structure because all of them have hierarchical structures.
Therefore, the commuting matrix can be employed here to extract the semantics encapsulated in the RecurMS.
The structure of RecurMS has such strong restrictions on the object types that the similarity only between the same objects is nonzero and between the different objects is zero.
That is, the object types are coupled tightly.
To decouple the object types, we decompose the proposed schematic structure into different recurrent meta-paths and recurrent meta-trees,
and then define the commuting matrices of the recurrent meta-paths and meta-trees as similar to the ones of the meta-paths and meta-structures.
As a result, the Recurrent Meta-Structure-based Similarity (RMSS) is defined as the weighted summation of all these commuting matrices.
The proposed RMSS is robust to different schematic structures, i.e., meta-paths or meta-structures,
because its structure integrates all the possible meta-paths and meta-structures.
To evaluate the importance of different recurrent meta-paths and meta-structures, two kinds of weighting strategies, local weighting strategy and global weighting strategy, are proposed.
The weighting strategies consider the sparsity and strength of different recurrent meta-paths and meta-trees in the HIN.
The experimental evaluations on three real datasets reveals that the existing metrics are sensitive to different meta-paths or meta-structures,
and that the proposed RMSS outperforms the existing metrics in terms of ranking and clustering tasks.

The main contributions are summarized as follows.
1) We propose the recurrent meta-structure which combines all the meta-paths and meta-structures. The RecurMS can be constructed automatically.
In order to decouple the object types, the RecurMS is decomposed into several recurrent meta-paths and meta-trees;
2) We define the commuting matrices of the recurrent meta-paths and meta-trees, and propose two kinds of weighting strategies to determine the weights of different recurrent meta-paths and meta-structures.
The proposed robust RMSS is defined by the weighted summation of all these commuting matrices.
3) The experimental evaluations reveal the proposed RMSS outperforms the baselines in terms of ranking and clustering tasks and that the existing metrics are sensitive to different meta-paths and meta-structures.

The rest of the paper is organized as follows. Section \ref{sec:relatedwork} introduces related works.
Section \ref{sec:preliminaries} provides some preliminaries on HINs.
Section \ref{sec:DeepMSDetectDecomp} provides an approach to decomposing the recurrent meta-structure into several recurrent meta-paths and recurrent meta-trees.
Section \ref{sec:deepmetastructure} introduces the definition of RMSS.
The experimental evaluations are introduced in section \ref{sec:expriment}. The conclusion is introduced in section \ref{sec:conclusion}.

\section{Related Work}\label{sec:relatedwork}

The similarity measure plays fundamental roles in the field of network analysis, and can be applied to many areas, e.g., clustering, recommendation, Web search etc.
At the beginning, only the feature based similarity measures were proposed, e.g., Cosine similarity, Jaccard coefficient, Euclidean distance and Minkowski distance \cite{HKP:2012}.
However, the feature based similarity measures ignored the link information in networks.
Afterwards, researchers realized the importance of the links in measuring the similarities between vertices, and proposed the link-based similarity measures \cite{JW:2002,JW:2003,SLZSY:2017}.
Article \cite{JW:2002} proposed a general similarity measure $SimRank$ combining the link information. The $SimRank$ argued that two similar objects must relate to similar objects.
Article \cite{ASCOSPLUSPLUS:2015} discovered that the SimRank in homogeneous networks and its families failed to capture similar node pairs in certain conditions.
Therefore, the authors proposed new similarity measures ASCOS and ASCOS++ to address the above problem.
Article \cite{JW:2003} evaluated the similarities of objects by a random walk model with restart.
In article \cite{LinkPred:2017}, the authors summarized the off-the-shelf works on the link prediction including many state-of-the-art similarity measures in homogeneous information networks.
Article \cite{WSSHSWZ:2016} proposed a socialized word embedding algorithm integrating user's personal characteristics and user's social relationship on social media.
Literature \cite{RelationalRetrievalPCRW:2010} proposed a novel learnable proximity measure which is defined by a weighted combination of simple "path experts" following a particular sequence of labeled edges.

This paper is concerned with the robust and semantic-rich similarity measure in heterogeneous information networks.
To the best of our knowledge, Sun et al. \cite{SHZYCW:2009} proposed the bi-type information network, and integrated clustering and ranking for analyzing it.
In the article \cite{SYH:2009}, She extended the bi-type information network to the heterogeneous information network with star network schema and studied ranking-based clustering on it.
The literatures \cite{SH:2012,SLZSY:2017} gives a comprehensive summarization of research topics on HINs including similarity measure, clustering, classification,
link prediction, ranking, recommendation, information fusion and other applications. Measuring the similarities between objects is a fundamental problem in HINs.
The similarity measures in HINs must organically integrate the rich semantics as well as the structural information.
This is the prominent difference between the similarity measures in HINs and the ones in the homogeneous information networks.
Below, we respectively introduce the similarity and relevance measures in HINs.

\textbf{(Similarity Measure in HINs)} Sun \cite{SHYYW:2011} employed the commuting matrix of a meta-path to define the meta-path-based similarity PathSim in HINs.
Literature \cite{PathSimExtHIN:2014} revisited the definition of PathSim and overcame its drawback, i.e., omiting some supportive information.
Lao and Cohen \cite{LC:2010a,LC:2010b} proposed a Path Constrained Random Walk (PCRW) model to evaluate the entity similarity in labeled directed graphs.
This model can be applied to measuring the similarity between objects in HINs.
Meng eta la. \cite{RelMeasureLSHIN:2014} proposed a novel similarity measure AvgSim which provided a unified framework to measure the similarity of same or different-typed object pairs.
Usman et al. \cite{UsmanOseledets:2015} employed the tensor techniques to measure the similarity between objects in HINs.
Wang et al. \cite{InfluenceSimHIN:2012} merged two different topics, influence maximization and similarity measure, together to reinforce each other for better and more meaningful results.
Yu et al. \cite{UserGuidedEntitySimSel:2012} employed a meta-path-based ranking model ensemble to represent semantic meanings for similarity queries, and exploited user-guidance to understand users query.
Xiong et al. \cite{XZY:2015} studied the problem of obtaining the $\text{top-}k$ similar object pairs based on user-specified join paths.
Usman et al. \cite{UsmanOseledets:2015} employed the tensor techniques to measure the similarity between objects in HINs.
Literature \cite{TopkSimSearchXstarSchema:2015} proposed a structural-based similarity measure NetSim to efficiently compute similarity between centers in HINs with x-star network schema.
Wang et al. \cite{DistantMetaPathSimTextHIN:2017} proposed a distant meta-path similarity, which can capture semantics between two distant objects, to provide more meaningful entity proximity.
Zhou et al. proposed a semantic-rich stratified-meta-structure-based similarity measure SMSS by integrating all of the commuting matrices of the meta-paths and meta-structures in HINs.
The stratified meta-structure can be constructed automatically, and therefore SMSS does not depend on any user-specified meta-paths or meta-structures.
Zhang et al. \cite{HeteRank:2018} proposed a general similarity measure HeteRank, which integrates the multi-relationships between objects for finding underlying similar objects.

\textbf{(Relevance Measure in HINs)} Shi et al. \cite{SKHYW:2014} extended the similarity measure in HINs to the relevance measure which can be used to evaluate the relatedness of two object with different types.
For an user-specified meta-path, His method $HeteSim$ is based on the pairwise random walk from its two endpoints to its center.
Gupta et al. \cite{NewRelMeasureHIN:2015} proposed a new meta-path-based relevance measure, which is semi-metric and incorporates the path semantics by following the user-specified meta-path, in HINs.
Bu et al. \cite{MetaPathSelUserPreferRelSearch:2014} proposed a two phase process to find the top-k relevance search in HINs.
The first phase aimed to obtain the initial relevance score based on the pair-wise path-constrained random walk, and the second phase took user preference into consideration to combine all the relevance matrices.
Xiong et al. \cite{Drug-TargetInteractHBN:2014} proposed an optimization algorithm LSH-HeteSim to capture the drug-target interaction in heterogeneous biological networks.
Literature \cite{LearnRelApplication:2011} proposed a novel approach to modeling user interest from heterogeneous data sources with distinct but unknown importance,
which seeks a scalable relevance model of user interest. Zhu et al. \cite{RelevanceSearchSignedHIN:2015} proposed a relevance search measure SignSim based on signed meta-path factorization in Signed HINs.

\section{Preliminaries}\label{sec:preliminaries}
In this section, we introduce the definition of HINs and some important concepts,
e.g., network schema, meta-paths and meta-structures.
the network schema of a HIN is essentially its template guiding the generation of the HIN.
The meta-paths and meta-structures are two kinds of schematic structures. They can capture semantics encapsulated in the HINs.

\subsection{The HIN Model}\label{subsec:HINModel}

\begin{definition}\label{def_hin}
\textbf{(Heterogeneous Information Network)}
An information network \cite{SWLYW:2014} is a directed graph $G=(V,E,\mathcal{A},\mathcal{R})$ where $V$ is a set of objects and $E$ is a set of links. $\mathcal{A}$ and $\mathcal{R}$ respectively
denote the set of object types and link types. $G$ is called a heterogeneous information network (HIN) if $|\mathcal{A}|>1$ or $|\mathcal{R}|>1$.
Otherwise, it is called a homogeneous information network.
\end{definition}

\begin{figure}[htb]
  \centering
  \includegraphics[width=0.6\textwidth]{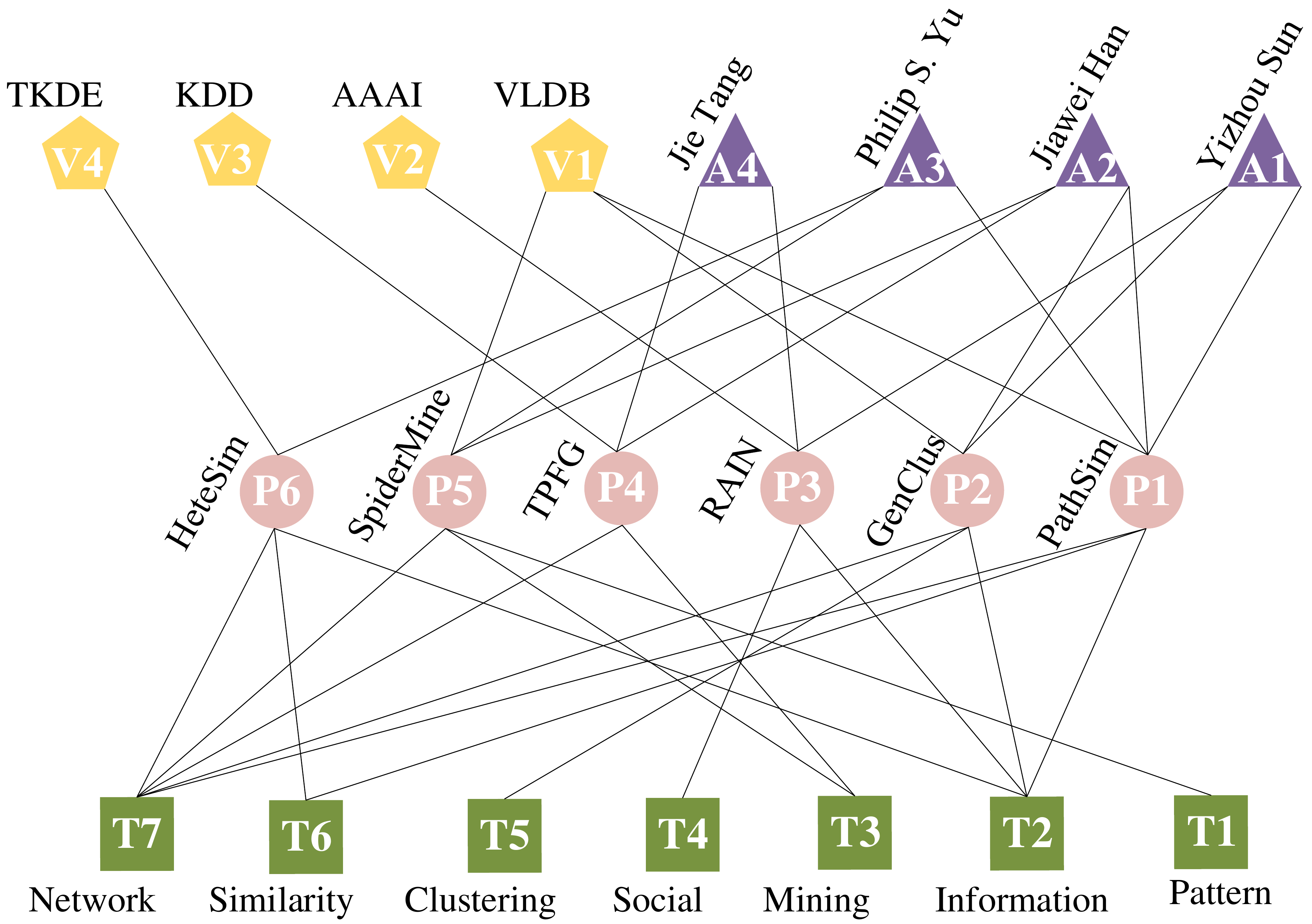}
  \caption{An Illustrative Bibliographic Information Network with actual papers, authors, terms and venues. The triangles, circles, squares, and pentagons respectively stand for authors, papers, terms and venues.}
  \label{example_hins}
\end{figure}

Heterogeneous information networks, which is defined in the definition \ref{def_hin}, consist of multi-typed objects and their interconnected relations.
For any object $v\in V$, it belongs to an object type $\phi(v)\in\mathcal{A}$. For any link $e\in E$, it belongs to a link type $\psi(e)\in\mathcal{R}$.
In essence, $\psi(e)$ represents a relation from its source object type to its target object type.
If two links belong to the same link type, they share the same starting object type as well as the ending object type.

Fig. \ref{example_hins} shows an illustrative bibliographic information network with four actual object types, i.e. \textit{Author} ($A$), \textit{Paper} ($P$), \textit{Venue} ($V$) and \textit{Term} ($T$).
The type \textit{Author} contains four instances: Yizhou Sun, Jiawei Han, Philip S. Yu, and Jie Tang.
The type \textit{Venue} contains four instances: VLDB, AAAI, KDD, TKDE. The type \textit{Paper} contains
six instances: PathSim \cite{SHYYW:2011}, GenClus \cite{SAH:2012}, RAIN \cite{RAIN:2015}, TPFG \cite{TPFG:2010}, SpiderMine \cite{SpiderMine:2011} and HeteSim \cite{SKHYW:2014}.
The type \textit{Term} constains six instances: Pattern, Information, Mining, Social, Clustering, Similarity and Network.
Each paper published at a venue must have its authors and its related terms. Hence, they contain three types of links: $P\leftrightarrow A$, $P\leftrightarrow V$ and $P\leftrightarrow T$.

\begin{figure}[htb]
  \centering
  \includegraphics[width=0.6\textwidth]{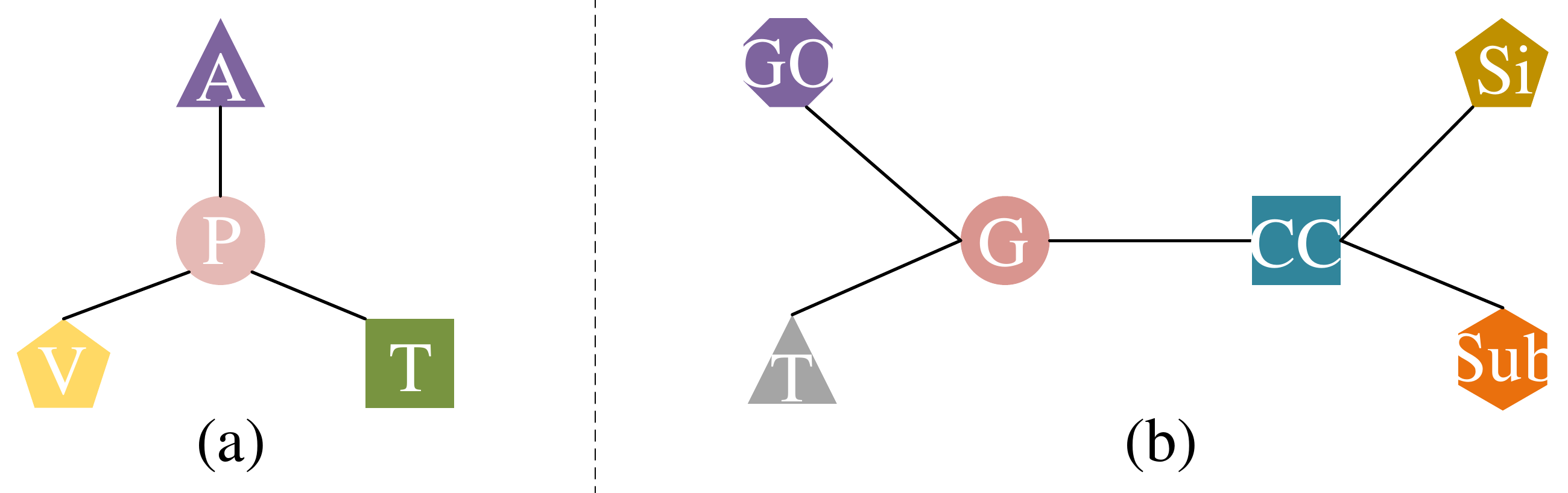}
  \caption{(a) Bibliographic network schema. (b) Biological network schema.}
  \label{example_network_schema}
\end{figure}

\begin{definition}\label{def_networkSchema}
\textbf{(Network Schema)}
The Network schema $\Theta_G=(\mathcal{A},\mathcal{R})$ \cite{SWLYW:2014} of $G$ is a directed graph consisting of the object types in $\mathcal{A}$
and the link types in $\mathcal{R}$.
\end{definition}

The network schema, which is defined in the definition \ref{def_networkSchema}, provides a meta-level description for the HIN.
The link types in $\mathcal{R}$ are essentially the relations from source object types to target object types.
Fig. \ref{example_network_schema}(a) shows the network schema for the HIN in Fig. \ref{example_hins}. Many biological networks can be modeled as HINs as well.
In this paper, we use a biological information network with six object types \textit{Gene} ($G$), \textit{Tissue} ($T$), \textit{GeneOntology} ($GO$),
\textit{ChemicalCompound} ($CC$), \textit{Substructure} ($Sub$) and \textit{SideEffect} ($Si$) as an example. It contains five link types
$GO\leftrightarrow G$, $T\leftrightarrow G$, $G\leftrightarrow CC$, $CC \leftrightarrow Si$, $CC\leftrightarrow Sub$.
Its network schema is shown in \ref{example_network_schema}(b).

\subsection{Meta Paths and Meta Structures}\label{subsec:metapath_metastructure}

There are rich semantics in HIN $G$. These semantics can be captured by meta-paths, meta-structures
or even more complicated schematic structures in $\Theta_G$.

\begin{definition}\label{def_metapath}
\textbf{(Meta-Path)}
The meta-path \cite{SHYYW:2011} is an alternate sequence of object types and link types.
It can be denoted by $\mathcal{P}=T_1\xrightarrow{R_1}T_2\xrightarrow{R_2}\cdots\xrightarrow{R_{l-2}}T_{l-1}\xrightarrow{R_{l-1}}T_l$,
where $T_i\in\mathcal{A}, i=1,\cdots,l$ and $R_j\in\mathcal{R}, j=1,\cdots,l-1$.
\end{definition}

In definition \ref{def_metapath}, $R_j$ is a link type starting from  $T_j$ to $T_{j+1}, j=1,\cdots,l-1$.
In essence, $R_j$ is a relation from $T_j$ to $T_{j+1}$.
The meta-path is essentially a composite relation $R_1\circ R_2\circ\cdots\circ R_{l-1}$.
That is, the meta-path can capture rich semantics contained in the HINs.
Throughout the paper, the meta-path $\mathcal{P}$ is compactly denoted as $(T_1,T_2,\cdots,T_{l-1},T_l)$ unless stated otherwise.

According to article \cite{SHYYW:2011}, there are some useful concepts related to $\mathcal{P}$.
the length of $\mathcal{P}$ is equal to the number of link types, i.e. $l-1$.
A path $P=(o_1,o_2,\cdots,o_l)$ in the HIN $G$ is an instance of $\mathcal{P}$ if $\phi(o_i)=T_i$ and $\psi(o_j,o_{j+1})=R_j$, where
$i=1,2,\cdots,l$ and $j=1,2,\cdots,l-1$. In general, $P$ is called a path instance following $\mathcal{P}$.
A meta-path $\mathcal{P}'=T_l\xrightarrow{R_{l-1}^{-1}}T_{l-1}\xrightarrow{R_{l-2}^{-1}}\cdots\xrightarrow{R_2^{-1}}T_2\xrightarrow{R_1^{-1}}T_1$
is called the reverse meta-path of $\mathcal{P}$, where $R^{-1}_i$ denotes the reverse relation of $R_i$ from $A_{i+1}$ to $A_i$.
The reverse meta-path of $\mathcal{P}$ is denoted by $\mathcal{P}^{-1}$. A meta-path $\mathcal{P}$ is symmetric if $\mathcal{P}=\mathcal{P}^{-1}$.
For the meta-path $\mathcal{P}$, let $W_{T_iT_{i+1}}$ denote the relation matrix of the relation $R_i$, where $i=1,2,\cdots,l-1$.
Its entry $W_{T_iT_{i+1}}(s,t)=1$ if there is an edge from the $s-\text{th}$ object in $T_i$ to the $t-\text{th}$ object in $T_{i+1}$, otherwise it is equal to 0.
The commuting matrix $\mathcal{M}_{\mathcal{P}}$ of the meta-path $\mathcal{P}$ is defined in the definition \ref{def_commutingMatrixOfMetaPath}.
The commuting matrix of the $\mathcal{P}^{-1}$ is equal to $\mathcal{M}^T_{\mathcal{P}}$.

\begin{definition}\label{def_commutingMatrixOfMetaPath}
\textbf{(Commuting Matrix of the Meta-Path)}
The commuting matrix $\mathcal{M}_{\mathcal{P}}$ of the meta-path $\mathcal{P}=(T_1,T_2,\cdots,T_{l-1},T_l)$ is defined as
\begin{displaymath}
\mathcal{M}_{\mathcal{P}}=W_{T_1T_2}\times W_{T_2T_3}\times\cdots\times W_{T_{l-1}T_l},
\end{displaymath}
where $W_{T_iT_{i+1}}$ denotes the relation matrix from $T_i$ to $T_{i+1}$.
\end{definition}

\begin{figure}[htb]
  \centering
  \includegraphics[width=0.6\textwidth]{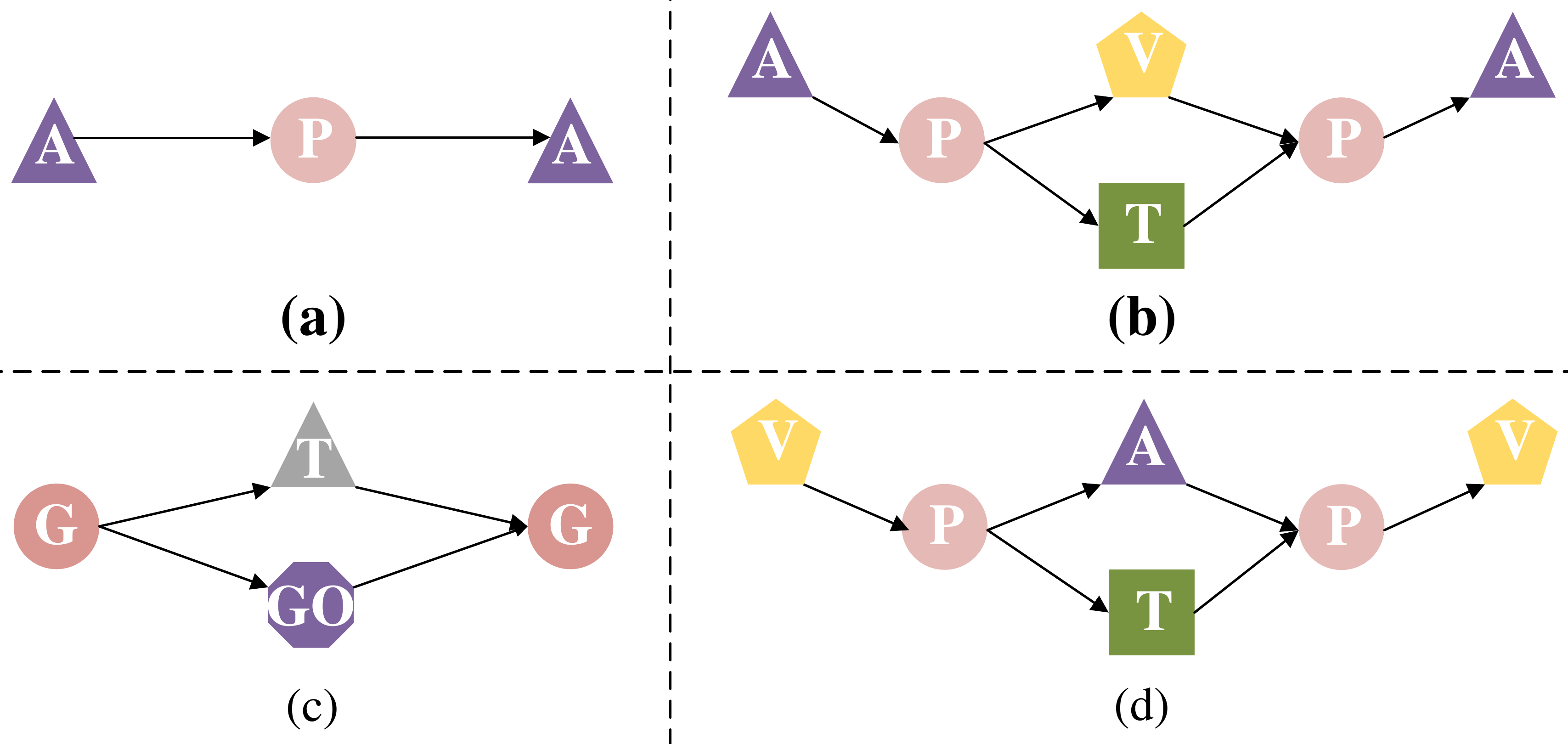}
  \caption{Some Meta Paths and meta-structures on the network schema shown in Fig. \ref{example_network_schema}(a,b).}
  \label{example_meta_path}
\end{figure}

Different meta-paths carry different semantics.
$(A,P,A)$ shown in Fig. \ref{example_meta_path}(a) expresses the information ``Two authors cooperate on a paper''.
However, literature \cite{HZCSML:2016} pointed out meta-paths can only capture relatively simple and biased semantics.
For example, $(A,P,V,P,A)$ expresses the information ``Two authors write a paper published on the same venue'',
but neglects the one ``Two authors write a paper containing the same term''. To overcome this issue, \cite{HZCSML:2016} proposed the meta-structure.

\begin{definition}\label{def_metaStructure}
\textbf{(Meta-Structure)}
The meta-structure \cite{HZCSML:2016} $\mathcal{S}=(\mathcal{V}_{\mathcal{S}},\mathcal{E}_{\mathcal{S}},T_s,T_t)$ is a directed acyclic graph
with a single source object type $T_s$ and a single target object type $T_t$. $\mathcal{V}_{\mathcal{S}}$ is a set of object types, and $\mathcal{E}_{\mathcal{S}}$ is a set of link types.
\end{definition}

The meta-structure, which is defined in the definition \ref{def_metaStructure}, can capture complex semantics.
Fig. \ref{example_meta_path}(b,c,d) shows three kinds of meta-structures.
For ease of presentation, these meta-structures are denoted as $(A,P,(V,T),P,A)$, $(G,(GO,T),G)$ and $(V,P,(A,T),P,V)$ respectively.
The meta-structure shown in Fig. \ref{example_meta_path}(b) expresses the information ``Two authors write their papers both containing the same terms and in the sam venue'',
but ignores the information ``Two authors cooperate on a paper.'' That is, the meta-structure can only capture biased semantics as well.

Given a meta-structure $\mathcal{S}$ with height $h_0$, its object types are sorted in the topological order.
For $h=0,1,\cdots,h_0-1$, let $L_h$ denote the set of object types on the layer $h$, and
$CP_{L_h}$ denote the Cartesian product of the set of objects belonging to different types in $L_h$.
The relation matrix $W_{L_hL_{h+1}}$ from $CP_{L_h}$ to $CP_{L_{h+1}}$ is defined as:
the entry $(s,t)$ of $W_{L_hL_{h+1}}$ is equal to 1 if the $s$-th element $CP_{L_h}(s)$ of $CP_{L_h}$ is adjacent to the $t$-th one $CP_{L_{h+1}}(t)$ of $CP_{L_{h+1}}$ in $G$,
otherwise it is equal to 0. $CP_{L_h}(s)$ and $CP_{L_{h+1}}(t)$ are adjacent if and only if for any $u\in CP_{L_h}(s)$ and $v\in CP_{L_{h+1}}(t)$,
$u$ and $v$ are adjacent in $G$, and $\phi(u)$ and $\phi(v)$ are adjacent in $\Theta_G$.
The commuting matrix $\mathcal{M}_{\mathcal{S}}$ of the meta-structure $\mathcal{S}$ is defined in the definition \ref{def_commutingMatrixOfMetaStructure}.
Each entry in $\mathcal{M}_{\mathcal{S}}$ represents the number of instances following $\mathcal{S}$.
The commuting matrix of its reverse is equal to $\mathcal{M}_{\mathcal{S}}^T$.

\begin{definition}\label{def_commutingMatrixOfMetaStructure}
\textbf{(Commuting Matrix of the Meta-Structure)}
The commuting matrix of the meta-structure $\mathcal{S}$ is defined as
\begin{displaymath}
\mathcal{M}_{\mathcal{S}}=\prod_{i=0}^{h_0-1}W_{L_iL_{i+1}},
\end{displaymath}
where $W_{L_iL_{i+1}}$ denotes the relation matrix from $C_{L_i}$ to $C_{L_{i+1}}$.
\end{definition}

Both meta-paths and meta-structures need to be specified by users.
In the bibliographical information networks, it is comparatively easy for users to specify meta-paths or meta-structures.
However, specifying meta-paths or meta-structures becomes very difficult in the biological information networks,
because in reality it contains many object types (Gene, Gene Ontology, Tissue, Chemical Compound,
Chemical Ontology, Side Effect, Substructure, Pathway, Disease and Gene Family) and many relations. In Fig. \ref{example_network_schema}(b), we give a biological network schema only containing
six object types and five link types.

In this paper, we aim to define a robust semantic-rich similarity measure in HINs
Formally, the problem takes a HIN, a source object as input, and then outputs a vector whose entries denote the similarity between the source object to the target object.

\section{Recurrent Meta Structure Construction and Decomposition}\label{sec:DeepMSDetectDecomp}

In this section, we introduce the architecture of the recurrent meta-structure and an approach to decomposing the recurrent meta-structure into several recurrent meta-paths and recurrent meta-trees.

\subsection{Recurrent Meta Structure Construction}\label{subsec:metastructure}

Before proceeding, we introduce an important concept, an augmented spanning tree of the network schema $\Theta_G$, see the definition \ref{def_AST}.
It is used in the processing of constructing and decomposing the recurrent meta-structure.

\begin{definition}\label{def_AST}
\textbf{(Augmented Spanning Tree)}
An augmented spanning tree $AST_{\Theta_G}=(\mathcal{A}_{AST},\mathcal{R}_{AST})$ of $\Theta_G$ is a tree rooted at the source object type and containing
all the link types in $\Theta_G$. $\mathcal{A}_{AST}$ denotes the set of object types in $AST_{\Theta_G}$, and $\mathcal{R}_{AST}$ denotes the set of link types in $AST_{\Theta_G}$.
Note that $\mathcal{A}_{AST}$ contains the object types in $\Theta_G$ and some of their duplicates, and $\mathcal{R}_{AST}$ contains the links types consisting of two object types in $\mathcal{A}_{AST}$.
\end{definition}

Now, we introduce the construction rule of the augmented spanning tree of $\Theta_G$.
If the network schema is a tree, its augmented spanning tree is equal to the network schema itself.
If the network schema is not a tree, its augmented spanning tree is constructed based on its spanning tree as follows.
The spanning tree of the network schema can be constructed using Breadth-First Search (BFS) starting from the source object type.
We then traverse the spanning tree from top to bottom and from left to right.
For the current object type in the process of traversing, if an edge adjacent to it in the network schema is not contained in the current spanning tree,
we duplicate the object type adjacent to it and add an edge from it to the copied object type in the current spanning tree.

\begin{figure}[htb]
  \centering
  \includegraphics[width=0.5\textwidth]{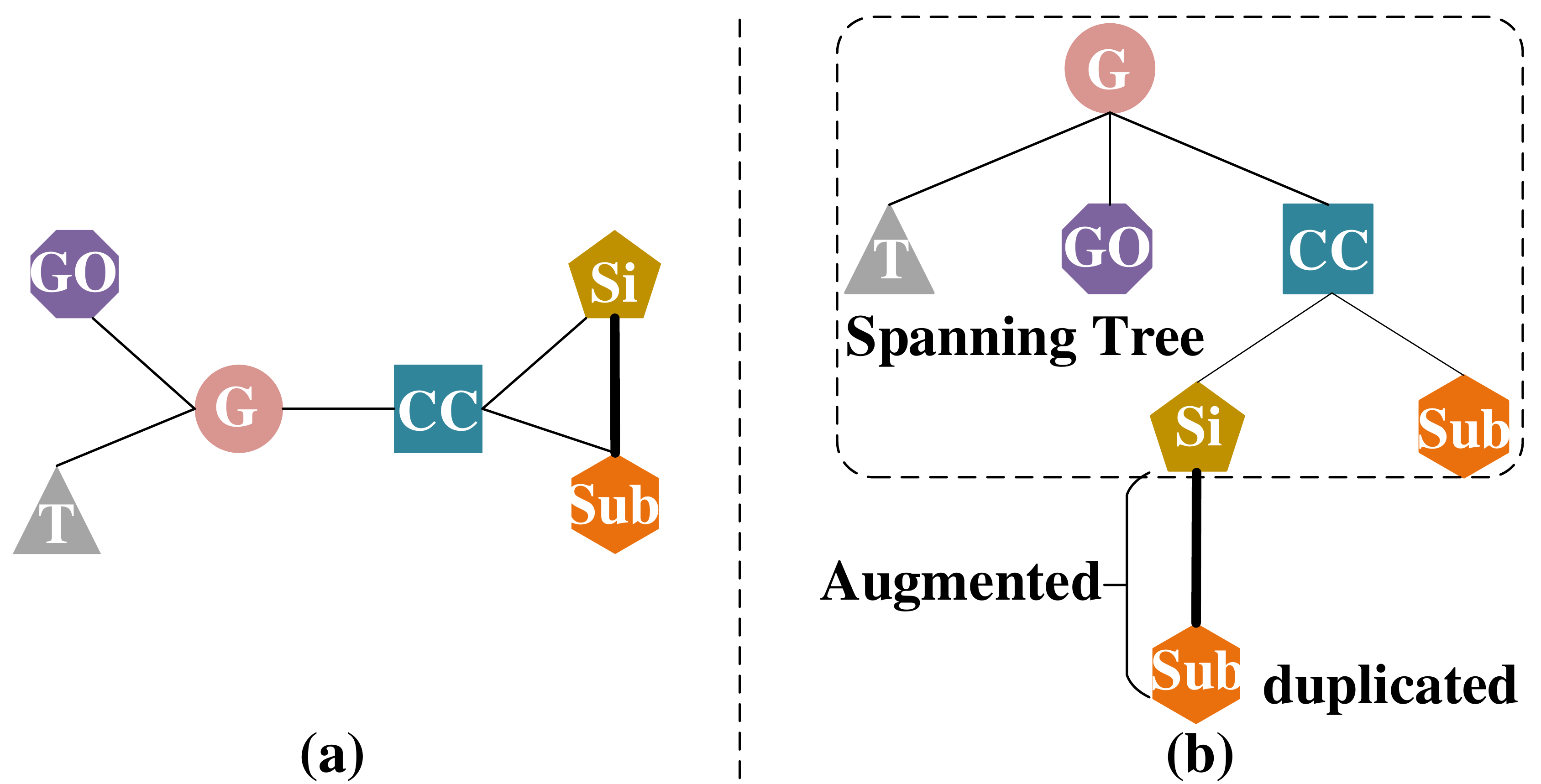}
  \caption{Constructing the augmented spanning tree when the network schema is not a tree.}
  \label{augmented_spanning_tree}
\end{figure}

We exemplify the construction of the augmented spanning tree when the network schema is not a tree. Suppose an edge $\left(Si,Sub\right)$ is added to the network schema shown in Fig. \ref{example_network_schema}(b).
As a result, we get a new network schema shown in Fig. \ref{augmented_spanning_tree}(a). Next, we show how to construct the augmented spanning tree for this network schema, see Fig. \ref{augmented_spanning_tree}(b).
Its spanning tree is enclosed by the dashed line frame. When we reach the node $Si$ in the process of traversing, the edge $(Si,Sub)$ incidental to $Si$ is not contained in the spanning tree.
So, we make a copy of the node $Sub$ and add an edge from $Si$ to the copied $Sub$.

\begin{lemma}\label{AST_NetworkSchema_edge}
Given a HIN $G$, its network schema is denoted by $\Theta_G=(\mathcal{A},\mathcal{R})$. The augmented spanning tree of $\Theta_G$ is denoted by $AST_{\Theta_G}=(\mathcal{A}_{AST},\mathcal{R}_{AST})$.
If one object type and its duplicate are not distinguished explicitly in $\mathcal{A}_{AST}$, we have $\mathcal{A}=\mathcal{A}_{AST}$ and $\mathcal{R}=\mathcal{R}_{AST}$.
\end{lemma}
\begin{proof}
According to the construction rule of $AST_{\Theta_G}$, obviously $\mathcal{A}=\mathcal{A}_{AST}$ and $\mathcal{R}=\mathcal{R}_{AST}$
because one object type and its duplicate are not distinguished explicitly in $\mathcal{A}_{AST}$.
\end{proof}

According to lemma \ref{AST_NetworkSchema_edge}, the augmented spanning tree reformulates the network schema if the object types and their duplicates are thought of as the same elements.
That is, a link type $R_1$ in $AST_{\Theta_G}$ is equal to one $R_2$ in $\Theta_G$
if and only if they share the same endpoints or one endpoint of $R_1$ is a copy of one endpoint of $R_2$.
Below, we introduce the definition of the recurrent meta-structure (RecurMS, see the definition \ref{def_recurMS}), and describe the construction rule of the recurrent meta-structure
based on the augmented spanning tree of the network schema.

\begin{definition}\label{def_recurMS}
\textbf{(Recurrent Meta Structure)}
A recurrent meta-structure is essentially a hierarchical graph consisting of object types with different layer labels.
Formally, it is denoted as $\mathcal{D}_G=(L_{0:\infty},\mathcal{R}_{\mathcal{D}_G})$, where $L_i, i=0,\cdots,\infty$ denotes the set of object types on the $i-\text{th}$ layer and
$\mathcal{R_{\mathcal{D}_G}}$ denotes the set of link types in RecurMS.
\end{definition}

RecurMS has two prominent advantages: (1) being automatically constructed by repetitively visiting object types in the process of traversing network schema;
(2) combining all the meta-paths and meta-structures.
Given a HIN $G$, we first extract its network schema $\Theta_G$, and then select a source object type and a target object type.
In this paper, we only consider the scenario that the source object type is the same as the target one.
The construction rule of the RecurMS $\mathcal{D}_G$ of $G$ is described as follows.
The source object type is placed on the 0-th layer.
The object types on the layer $l=1,2,\cdots,+\infty$ are composed of the neighbors of the object types on the layer $l-1$ on the network schema $\Theta_G$.
The adjacent object types are linked by an arrow pointing from the $(l-1)$-th layer down to the $l$-th layer.
Repeating the above process, we obtain the RecurMS $\mathcal{D}_G$.
It is noteworthy that an object type may appear in adjacent layers of the RecurMS if there exist circles (or self-loops) in the network schema.
At this time, one of them can be viewed as a copy of another one. 

\begin{figure}[htb]
  \centering
  \includegraphics[width=0.6\textwidth]{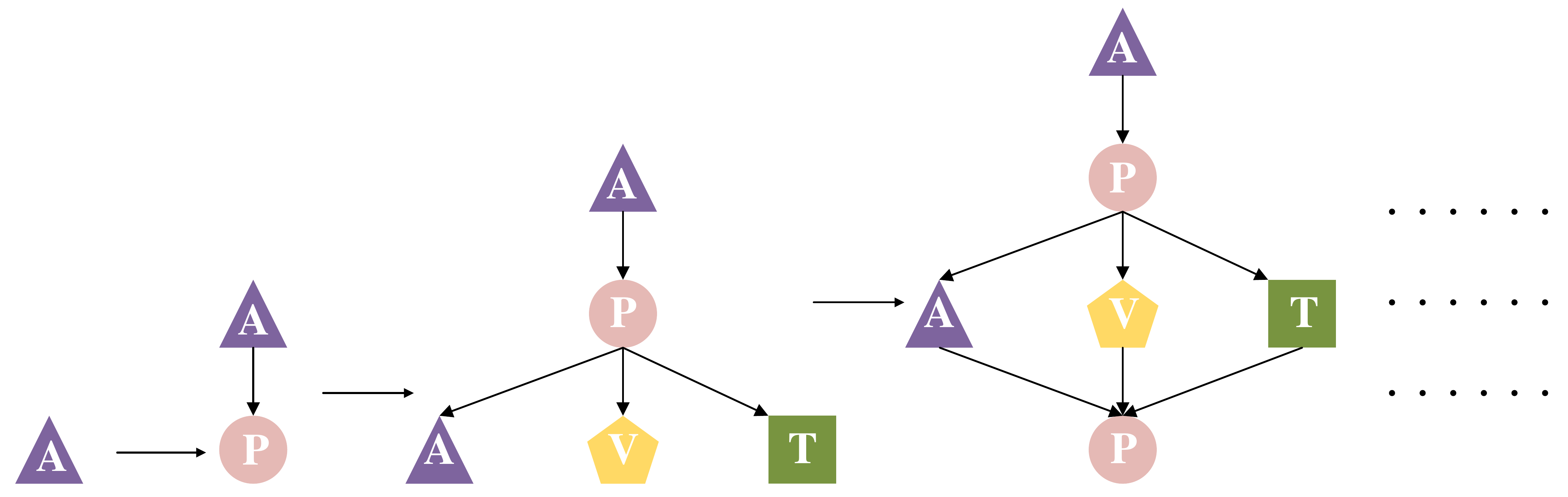}
  \caption{Constructing the RecurMS of the illustrative bibliographic information network.}
  \label{construction_DeepMS}
\end{figure}


Fig. \ref{dblp_decomposition}(a) shows the RecurMS of the network schema shown in Fig. \ref{example_network_schema}(a).
As shown in Fig. \ref{construction_DeepMS}, it can be constructed as follows.
$A$ is both the source and target object type. Firstly, $A$ is placed on the 0-th layer, see Fig. \ref{construction_DeepMS}(a).
$P$ is placed on the 1-st layer, because $P$ is the only neighbor of $A$ in the network schema shown in Fig. \ref{example_network_schema}(a), see Fig. \ref{construction_DeepMS}(b).
$A$, $V$ and $T$ are placed on the 3-rd layer, because they are the neighbors of $P$, see Fig. \ref{construction_DeepMS}(c).
Similarly, $P$ is again placed on the 4-th layer, because it is the neighbor of $A$, $V$ and $T$, see Fig. \ref{construction_DeepMS}(d). At this time, $P$ is visited again.
Repeating the above procedure, we obtain the RecurMS shown in Fig. \ref{dblp_decomposition}(a).
Fig. \ref{slap_decomposition}(a) shows the RecurMS of the network schema shown in Fig. \ref{example_network_schema}(b). Gene is both the source and target object type.
It is constructed as similarly as the one of the bibliographic network schema.

\begin{figure}[htb]
  \centering
  \includegraphics[width=0.95\textwidth]{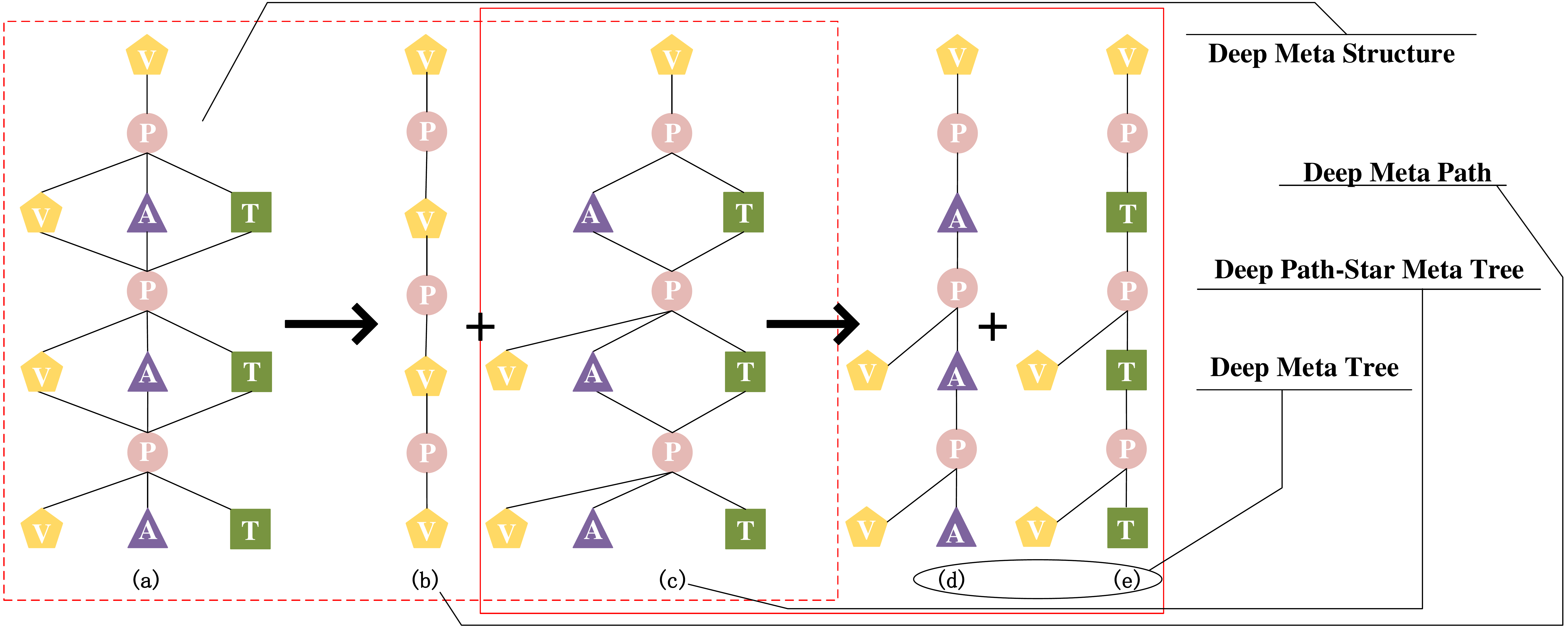}
  \caption{Decomposing the RecurMS of the bibliographic network schema.}
  \label{dblp_decomposition}
\end{figure}

\begin{figure}[htb]
  \centering
  \includegraphics[width=0.95\textwidth]{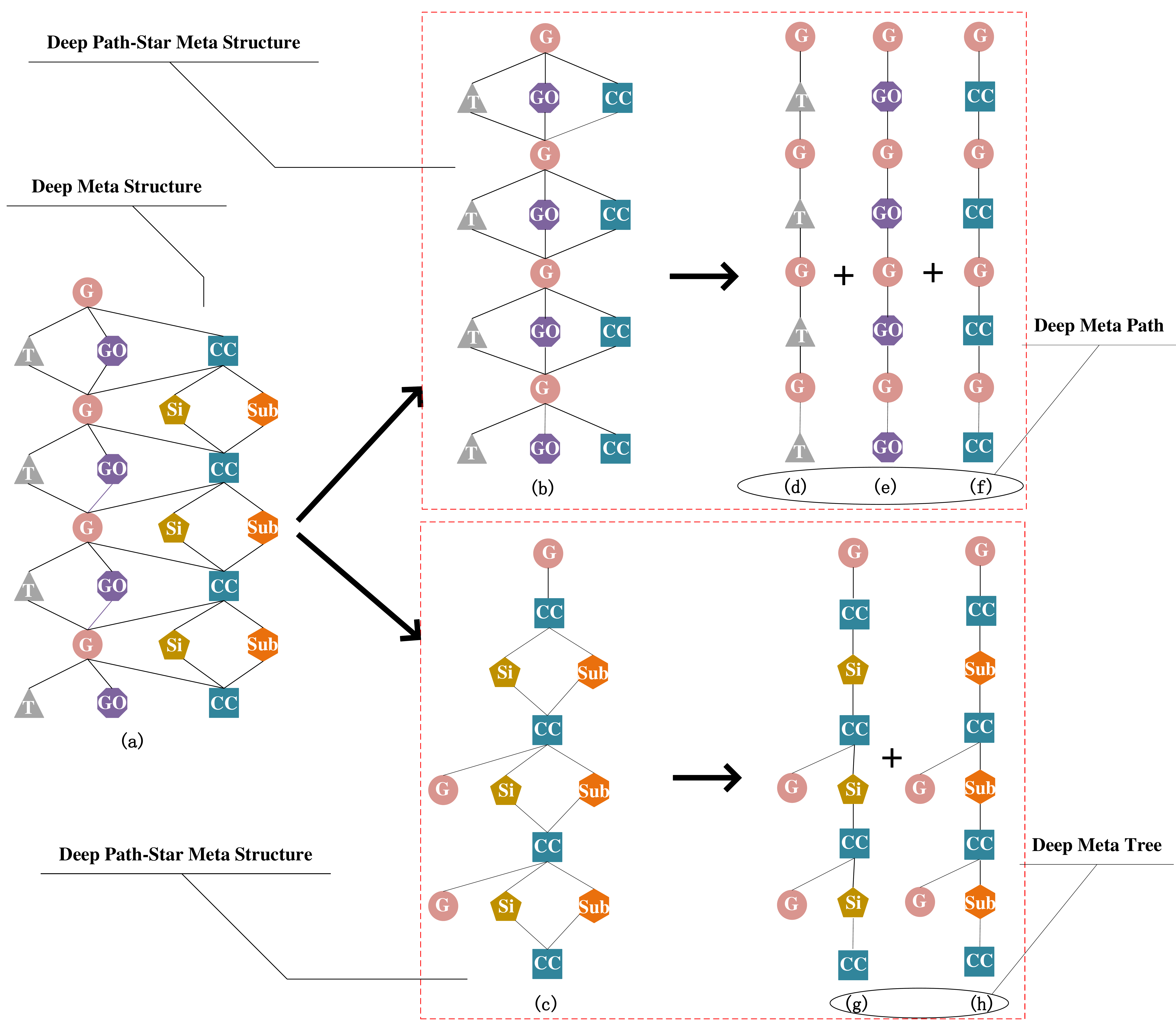}
  \caption{Decomposing the RecurMS of the biological network schema.}
  \label{slap_decomposition}
\end{figure}

According to definition \ref{def_recurMS}, the recurrent meta-structure consists of the object types with different layer labels and their relations in the network schema. Each layer is a set of object types.
Below, we give some properties of $\mathcal{D}_G$ in lemma \ref{RecurMS_Properties}.
According to these properties, $\mathcal{D}_G$ contains rich semantics.

\begin{lemma}\label{RecurMS_Properties}
Assume $h_1$ denotes the height of the augmented spanning tree $AST_{\Theta_G}$. Without loss of generality, $h_1\geq 2$.
$L_i, i=0,\cdots,\infty$ denotes the set of object types on the $i-\text{th}$ layer of $\mathcal{D}_G$.
$\mathcal{D}_G=(L_{0:\infty},\mathcal{R}_{\mathcal{D}_G})$ has the following properties.
\begin{enumerate}
  \item $L_0=\left\{T_s\right\}$, where $T_s$ is the source object type;
  \item $L_i\subseteq L_{i+2}, i=0,1,\cdots,h_1-2$ and $L_{j-1}=L_{j+1}, j=h_1,h_1+1,\cdots,\infty$;
  \item $\mathcal{D}_G$ contains all the meta-paths and the meta-structures.
\end{enumerate}
\end{lemma}
\begin{proof}
(1) According to the construction rule of $\mathcal{D}_G$, obviously $L_0=\left\{T_s\right\}$.

(2) When $i=0,1,\cdots,\infty$, the object types in $L_i$ must be added to $L_{i+2}$ according to the construction rule of $\mathcal{D}_G$.
In addition, there are some new object types in $L_{i+2}$, e.g. some children of the object types in $L_{i+1}$ in $AST_{\Theta_G}$.
Therefore, $L_i\subseteq L_{i+2}, i=0,\cdots,h_1-2$.
When $j=h_1,h_1+1,\cdots,\infty$, obviously $L_{j-1}=L_{j+1}$.
At this time, it is impossible for $L_{j+1}$ to contain some new object types because its layer label is larger than $h_1$.
Thus, $L_{j+1}\subseteq L_{j-1}$.

(3) Any meta-path $\mathcal{P}$ can be compactly denoted by $\left(T_s,T_{1,i_1},\cdots,T_{l-1,i_{l-1}},T_s\right)$ without loss of generality.
According to the construction rule of $\mathcal{D}_G$, we have $T_s\in L_0$, $T_{k,i_k}\in L_k, k=1,2,\cdots,l-1$ and $T_s\in L_{l}$.
Therefore, $\mathcal{P}$ must be in $\mathcal{D}_G$. For meta-structure, we can take same measures to prove.
\end{proof}

\subsection{Recurrent Meta Structure Decomposition}\label{subsec:decompositionOfDeepMS}

\begin{figure}[htb]
  \centering
  \includegraphics[width=0.8\textwidth]{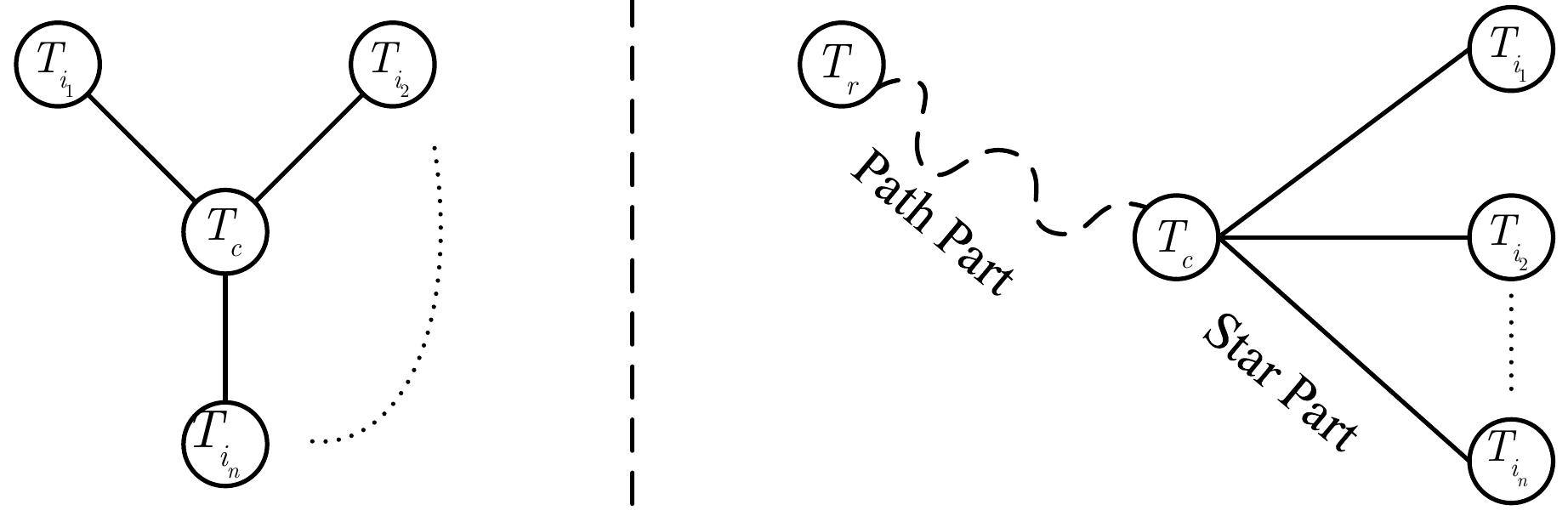}
  \caption{A star whose center is $T_c$ on the left hand, a path-star tree on the right hand.}
  \label{star_path_star_tree}
\end{figure}

This section provides some important concepts including star, path-star tree, recurrent path-star meta-structure, recurrent meta-path and recurrent meta-tree.
The star, which is defined in the definition \ref{def_star}, is a special tree consisting of a center and its neighbors.
The star $\left(T_c,\left(T_{i_1},\cdots,T_{i_n}\right)\right)$ is illustrated with Fig. \ref{star_path_star_tree}(a). Its center is $T_c$ and the neighbors of $T_c$ is $T_{i_1},T_{i_2},\cdots,T_{i_n}$.
The path-star tree, which is defined in the definition \ref{def_pathStarTree}, consists of a path and a star.
The path-star tree $\left(T_r,\cdots,T_p,\left(T_{i_1},\cdots,T_{i_n}\right)\right)$ is illustrated with Fig. \ref{star_path_star_tree}(b).
Its path part is from $T_r$ to $T_c$, its star part consists of the center $T_c$ and its neighbors $T_{i_1},\cdots,T_{i_n}$.
Throughout the paper, an infinite sequence $\left(a,x_0,\cdots,x_l,b,c,b,c,\cdots\right)$ is compactly denoted as
\begin{displaymath}
\left(a,x_0,\cdots,x_l,\overbrace{b,c}^{\infty}\right).
\end{displaymath}
The recurrent path-star meta-structure is defined in the definition \ref{def_recurPathStarMetaStruc},
and the recurrent meta-path and meta-tree are defined in the definitions \ref{def_recurrentMetaPath} and \ref{def_recurrentMetaTree} respectively.

\begin{definition}\label{def_star}
\textbf{(Star)}
A star, compactly denoted as $\left(T_c,\left(T_{i_1},\cdots,T_{i_n}\right)\right)$, is a tree consisting of a center $T_c$ and its neighbors $T_{i_1},\cdots,T_{i_n}$.
\end{definition}

\begin{definition}\label{def_pathStarTree}
\textbf{(Path-Star Tree)}
A path-star tree, compactly denoted as $\left(T_r,\cdots,T_p,\left(T_{i_1},\cdots,T_{i_n}\right)\right)$ is a rooted one consisting of a path and a star.
In specific, the path, compactly denoted as $\left(T_r,\cdots,T_p\right)$, is from the pivotal vertex $T_p$ to the root $T_r$,
and the star $T_p,\left(T_{i_1},\cdots,T_{i_n}\right)$ is composed of the pivotal vertex $T_p$ and its children $T_{i_1},\cdots,T_{i_n}$.
\end{definition}

\begin{definition}\label{def_recurPathStarMetaStruc}
\textbf{(Recurrent Path-Star Meta-Structure)}
A recurrent path-star meta-structure, compactly denoted as
\begin{displaymath}
\left(T_r,\cdots,\overbrace{T_p,\left(T_{i_1},\cdots,T_{i_n}\right)}^{\infty}\right),
\end{displaymath}
is a hierarchical structure consisting of a path-star tree and its duplicates.
It can be constructed by repetitively duplicating the star part of the path-star tree.
Note that each pivotal vertex except the first one is also connected to the root along the path $\left(T_r,\cdots,T_p\right)$.
\end{definition}

\begin{definition}\label{def_recurrentMetaPath}
\textbf{(Recurrent Meta-Path)}
The recurrent path-star meta-structure is called a recurrent meta-path if the path-star tree is a single edge.
It can be compactly denoted as
\begin{displaymath}
\left(\overbrace{T_r,T_{i_j}}^{\infty}\right)
\end{displaymath}
\end{definition}

\begin{definition}\label{def_recurrentMetaTree}
\textbf{(Recurrent Meta-Tree)}
A recurrent meta-tree is a hierarchical structure consisting of a path from the pivotal vertex to the root and one of children of the pivotal vertex.
It can be compactly denoted as
\begin{displaymath}
\left(T_r,\cdots,\overbrace{T_p,T_{i_j}}^{\infty}\right).
\end{displaymath}
Note that each pivotal vertex except the first one is also connected to the root along the path $\left(T_r,\cdots,T_p\right)$
\end{definition}

The object types with different layer labels are tightly coupled in the RecurMS. To decouple them, we should decompose the RecurMS.
After obtaining the augmented spanning tree $AST_{\Theta_G}$, we traverse its internal nodes from top to bottom and from left to right.
Each current object type is treated as a pivot like a bridge connecting two different components: (1) the path form the root (i.e. the source object type) to the pivot; (2) the star consisting of
the pivot as the center and its children. We obtain a path-star tree according to definition \ref{def_pathStarTree}.
Then, we augment all these path-star trees by repetitively duplicating the star part consisting of the pivotal object types and their children.
For each duplicated pivotal object type, it is connected to the target object type by the path part of the path-star tree.
Finally, we obtain several recurrent path-star meta-structures of the RecurMS. In essence, the RecurMS can be viewed as the combination of these substructures.
If the path-star tree is a single edge, the recurrent path-star meta-structure generated by it is specially called the recurrent meta-path.

Now, we formally describe the procedure of decomposing the RecurMS into several recurrent path-star meta-structures.
As stated previously, the RecurMS can be denoted as $\mathcal{D}_G=(L_{0:\infty},\mathcal{R}_{\mathcal{D}_G})$.
Without loss of generality, let $L_i=\left\{T_{i,1},T_{i,2},\cdots,T_{i,n_i}\right\}$. According to lemma \ref{RecurMS_Properties},
$L_0=\{T_s\}$, i.e. $n_0=1$, $T_{0,1}=T_s$. Assume $\mathcal{A}^{in}_{\Theta_G}$ denotes the set of internal nodes of the augmented spanning tree $AST_{\Theta_G}$,
whose elements are listed in the order from the top to the bottom and from the left to the right. Obviously, the source object type $T_s$ is firstly selected as the pivot.
As a result, we obtain a star consisting of the source object type $T_s$ and its children $T_{1,1},\cdots,T_{1,n_1}$.
At this time, we augment this star by repetitively duplicating $T_s$ and its children $T_{1,1},\cdots,T_{1,n_1}$.
As a result, the recurrent path-star meta-structure with $T_s$ as the pivot can be compactly denoted as
\begin{equation}\label{RecurPathStarMetaStrucSource}
\left(\overbrace{T_s,\left(T_{1,1},\cdots,T_{1,n_1}\right)}^{\infty}\right).
\end{equation}
For the pivot $T_{j,k}\in\mathcal{A}^{in}_{\Theta_G}$ and $T_{j,k}\neq T_s$, where $k\leq n_j$,
we should firstly calculate the path from $T_{j,k}$ to the root $T_s$, denoted by
$\mathcal{P}_{T_{j,k},T_s}=\left(T_{j,k}, T_{j-1,i_{j-1}}, \cdots, T_{1,i_1}, T_s\right)$.
As a result, we obtain the recurrent path-star meta-structure compactly denoted by
\begin{equation}\label{RecurPathStarMetaStrucNonsource}
\left(T_s,T_{1,i_1},\cdots,\overbrace{T_{j,k},\left(T_{j+1,1},\cdots,T_{j+1,n_{j+1}}\right)}^{\infty}\right)
\end{equation}
Note that all the pivots except the first one in Formula \ref{RecurPathStarMetaStrucNonsource}
are also linked to the path $\mathcal{P}_{T_{j,k},T_s}$.

Here, we respectively take the bibliographic network schema and the biological network schema, shown in Fig. \ref{example_network_schema}(a,b),
as examples to present how to generate the path-star trees. For the bibliographic network schema, $V$ is selected as the source object type.
Its augmented spanning tree rooted at $V$ is equal to the network schema itself because the bibliographic network schema is a tree.
For the biological network schema, $G$ is selected as the source object type.
Its augmented spanning tree rooted at $G$ is equal to the network schema itself because the biological network schema is a tree.
After obtaining their augmented spanning trees, we traverse its internal nodes from top to bottom and from left to right.
For the bibliographic network schema, its internal nodes are $V$ and $P$. When $V$ is treated as the pivot, its path from the root ($V$ itself) to $V$ is empty, and the star consists of
$V$ as its center and $P$.
When $P$ is treated as the pivot, the path from the root $V$ to $P$ is the edge $(V,P)$, and the star consists of $P$ and its children $A$ and $T$.
Their path-star trees are shown in Fig. \ref{PathChildrenTree}(a,b).
For the biological network schema, its internal nodes are $G$ and $CC$.
Their path-star trees are shown in Fig. \ref{PathChildrenTree}(c,d). They can be constructed as similarly as the bibliographic network schema.

\begin{figure}[htb]
  \centering
  \includegraphics[width=0.75\textwidth]{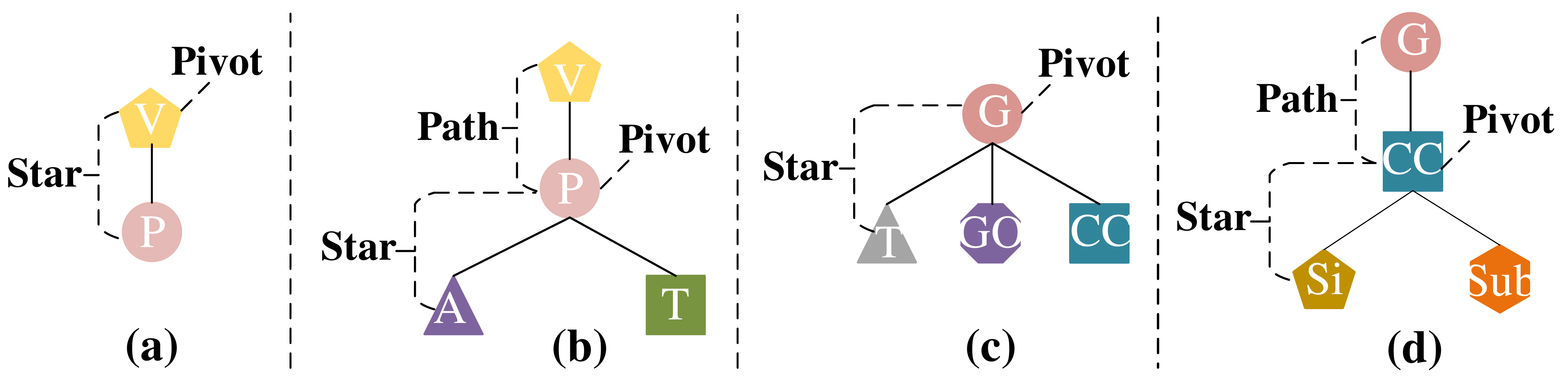}
  \caption{Path-star trees of the bibliographic and biological network schema.}
  \label{PathChildrenTree}
\end{figure}

For the bibliographic network schema, its recurrent path-star meta-structure can be congstructed as follows. The object types $V$ and $P$ are respectively treated as the pivots.
If $V$ is the pivot, it has only one child $P$. Its path-star tree is a single edge, see Fig. \ref{PathChildrenTree}(a).
We repetitively duplicate its star part, and finally obtain a recurrent meta-path shown in Fig. \ref{dblp_decomposition}(b).
It is noteworthy that the path part of the path-star tree is null at this time because the pivot $V$ is the source object type.
If $P$ is the pivot, it has two children $A$ and $T$. Its path-star tree is a tree, see Fig. \ref{PathChildrenTree}(b).
We repetitively duplicate its star part, and the pivot $P$ is linked to the target object type $V$ by the path part of the path-star tree.
Finally, we obtain a recurrent path-star meta-structure, see Fig. \ref{dblp_decomposition}(c).
Obviously, the RecurMS shown in Fig. \ref{dblp_decomposition}(a) can be decomposed into the recurrent meta-path
(see Fig. \ref{dblp_decomposition}(b)) and the recurrent path-star meta-structure (see Fig. \ref{dblp_decomposition}(c)).

For the biological network schema, its recurrent path-star meta-structure can be constructed as follows. The object types $G$ and $CC$ are respectively treated as the pivot.
If $G$ is the pivot, it has three children $T$, $GO$ and $CC$. Its path-star tree is shown in Fig. \ref{PathChildrenTree}(c).
We repetitively duplicate its star part, and finally obtain a recurrent path-star meta-structure shown in Fig. \ref{slap_decomposition}(b).
At this time, the path part of the path-star tree is null because the pivot $G$ is the source object type.
If $CC$ is the pivot, it has two children $Si$ and $Sub$. Its path-star tree is shown in Fig. \ref{PathChildrenTree}(d).
We repetitively duplicate it star part, and the pivot $CC$ is linked to each target object type $G$ by the path part of the path-star tree.
Finally, we obtain a recurrent path-star meta-structure shown in Fig. \ref{slap_decomposition}(c).
Obviously, the RecurMS shown in Fig. \ref{slap_decomposition}(a) can be decomposed into two recurrent path-star meta-structures, respectively shown in Fig. \ref{slap_decomposition}(b,c).

After obtaining the recurrent path-star meta-structures, we employ the commuting matrices of meta-paths or the meta-structures to extract semantics in them.
For recurrent meta-paths, it is comparatively easy to do this. For recurrent path-star meta-structures (not a path), the size of the commuting matrices may be very large
because the Cartesian product may yield a very large set.
At this time, we further decompose the recurrent path-star meta-structures into several simpler substructures respectively called recurrent meta-trees or recurrent meta-path.
The decomposition rule is to respectively consider each child of the pivotal object type.

For the recurrent path-star meta-structure shown in Formula \ref{RecurPathStarMetaStrucSource}, it can be decomposed into several recurrent meta-paths as follows.
\begin{equation}\label{DecompRecurPathStarMetaStrucSource}
\left(\overbrace{T_s,T_{1,1}}^{\infty}\right), \cdots,
\left(\overbrace{T_s,T_{1,n_1}}^{\infty}\right)
\end{equation}
For the recurrent path-star meta-structure shown in Formula \ref{RecurPathStarMetaStrucNonsource}, it can be decomposed into several recurrent meta-trees as follows.
\begin{equation}\label{DecompRecurPathStarMetaStrucNonsource}
\left(T_s,T_{1,i_1},\cdots,\overbrace{T_{j,k},T_{j+1,1}}^{\infty}\right), \cdots,
\left(T_s,T_{1,i_1},\cdots,\overbrace{T_{j,k},T_{j+1,n_{j+1}}}^{\infty}\right)
\end{equation}
Note that all the pivots except the first one in Formula \ref{DecompRecurPathStarMetaStrucNonsource} are also linked to the path $\mathcal{P}_{T_{j,k},T_s}$.

For example, the recurrent path-star meta-structure shown in Fig. \ref{dblp_decomposition}(c) can be decomposed into two recurrent meta-trees, see Fig. \ref{dblp_decomposition}(d,e).
The Fig. \ref{dblp_decomposition}(d) only consider the object type $A$ and the Fig \ref{dblp_decomposition}(e) only consider the object type $T$.
Similarly, the recurrent path-star meta-structure shown in Fig. \ref{slap_decomposition}(b) is decomposed into three recurrent meta-paths, see Fig. \ref{slap_decomposition}(d,e,f).
The recurrent path-star meta-structure shown in Fig. \ref{slap_decomposition}(c) is decomposed into two recurrent meta-trees, see Fig. \ref{slap_decomposition}(g,h).

The deep meta-paths and deep meta-trees shown in Formulas \ref{DecompRecurPathStarMetaStrucSource} and \ref{DecompRecurPathStarMetaStrucNonsource} is an infinite sequence of object types.
In essence, both deep meta-paths and deep meta-trees consists of a finite number of ingredients.
In specific, deep meta-paths consist of the source object type $T_s$ and one of its children $T_c$, and deep meta-trees consist of the path from the pivot $T_p$ up to $T_s$ and one of the children $T_c$ of $T_p$.
Algorithm \ref{DecompRMS} presents the pseudo-code of decomposing the recurrent meta-structure into deep meta-paths or deep meta-trees.
In algorithm \ref{DecompRMS}, recurrent meta-paths such as $\left(T_s,T_c,T_s,T_c,\cdots\right)$ is succinctly denoted as $\left(T_s,T_c\right)$,
and recurrent meta-tree such as $\left(T_s,T_{1,i_1},\cdots,T_{j,i_j},T_p,T_c,T_p,T_c,\cdots\right)$ is succinctly
denoted as $\left(T_s,T_{1,i_1},\cdots,T_{j,i_j},T_p,T_c\right)$. Line 2 employ BFS to construct a spanning tree of $\Theta_G$ rooted as $T_s$.
Lines 3-7 yields the augmented spanning tree $AST_{\Theta_G}$ of $\Theta_G$.
Lines 8-15 traverse the nodes of $AST_{\Theta_G}$ from top to bottom and from left to right, and yields deep meta-paths and deep meta-trees.
The time complexity of algorithm \ref{DecompRMS} is $O(|\mathcal{R}_{AST}|)$.

\begin{algorithm}\footnotesize
\caption{DecompRMS: Decomposing Recurrent Meta Structure }
\begin{algorithmic}[1]
\label{DecompRMS}
\renewcommand{\algorithmicrequire}{\textsc{Input:}}
\renewcommand{\algorithmicensure}{\textsc{Output:}}
\REQUIRE{Network Schema $\Theta_G$, Source Object Type $T_s$.}
\ENSURE{A list of deep meta-paths or deep meta-trees $DSL$}
\STATE $DSL\leftarrow\emptyset$, $AST_{\Theta_G}\leftarrow BFS(\Theta_G,T_s)$;
\FOR{$u\in AST_{\Theta_G}$}
    \FOR{$(u,v)\in\Theta_G$ and $(u,v)\notin AST_{\Theta_G}$}
        \STATE $v\_copy\leftarrow v$;
        \STATE $AST_{\Theta_G}\leftarrow AST_{\Theta_G}\cup\left\{(u,v\_copy)\right\}$ to $AST_{\Theta_G}$;
    \ENDFOR
\ENDFOR
\FOR{$T_p\in AST_{\Theta_G}$}
    \IF{$T_p=T_s$}
        \FOR{ each child $T_c$ of $T_s$}
            \STATE $DSL\leftarrow DSL\cup\left\{(T_s,T_c)\right\}$ to $DSL$;
        \ENDFOR
    \ELSE
        \STATE Construct the path $(T_s,T_{1,i_1},\cdots,T_{j,i_j},T_p)$ from $T_p$ up to $T_s$;
        \FOR{each child $T_c$ of $T_p$}
            \STATE $DSL\leftarrow DSL\cup\left\{(T_s,T_{1,i_1},\cdots,T_{j,i_j},T_p,T_c)\right\}$;
        \ENDFOR
    \ENDIF
\ENDFOR
\RETURN $DSL$;
\end{algorithmic}
\end{algorithm}

\section{Recurrent Meta Structure Based Similarity}\label{sec:deepmetastructure}

This section defines the proposed semantic-rich similarity measure RMSS and presents the pseudo-code of the algorithm for computing the similarity matrix.
RMSS does not depend on any pre-specified schematic structures, and therefore is robust to the schematic structures.
Throughout the paper, $\overline{X}$, which is defined in the definition \ref{def_normalizedVersion}, represents the normalized version of a matrix $X$.

\begin{definition}\label{def_normalizedVersion}
\textbf{(Normalized Matrix)}
The normalization of a matrix $X$ is defined as
\begin{displaymath}
\overline{X}=U_X^{-1}\cdot X,
\end{displaymath}
where $U_X$ is a diagonal matrix whose nonzero entries are equal to the row sum of $X$.
\end{definition}

\subsection{Similarity Measure}\label{subsec:similaritymeasure}

In this section, we first define commuting matrices of recurrent meta-paths and recurrent meta-trees, and then propose two kinds of strategies to determine the weights of these schematic structures.

For the recurrent meta-path $RMP_s$ shown in Formula \ref{DecompRecurPathStarMetaStrucSource}, e.g. Fig. \ref{dblp_decomposition}(b) and Fig. \ref{slap_decomposition}(d,e,f),
they can be collectively denoted as
\begin{equation}\label{DMP_source}
RMP_s=\left(\overbrace{T_s,T_c}^{\infty}\right).
\end{equation}
The substructure $(T_c,T_s,T_c)$ recurs $i=0,1,\cdots$ times in $RMP_s$.
In essence, $RMP_s$ can be decomposed into an infinite number of meta-paths such as
\begin{displaymath}
\begin{array}{c}
RMP^0_s=\left(T_s,T_c,T_s\right), \\
RMP^1_s=\left(T_s,T_c,T_s,T_c,T_s\right), \\
\cdots, \\
RMP^t_s=\left(T_s,\overbrace{T_c,T_s,T_c}^{t},T_s\right), \\
\cdots. \\
\end{array}
\end{displaymath}
The substructure $\left(T_c,T_s,T_c\right)$ recurs $t$ times in the meta path $RMP^t_s$, $t=0,1,\cdots$.
Assume $W_{T_sT_c}$ denotes the relation matrix from $T_s$ to $T_c$ and $\overline{W_{T_sT_c}}$ is its normalized version.
The commuting matrix $\mathcal{M}_{RMP_s}$ of $RMP_s$ is defined as the summation of the commuting matrices of $RMP^t_s, t=0,1,\cdots$, see Formula \ref{CommutingMatrix_source}.
\begin{equation}\label{CommutingMatrix_source}
\begin{array}{rl}
\mathcal{M}_{RMP_s} & =\displaystyle\sum^{\infty}_{t=0}\mathcal{M}_{RMP^t_s} \\
 & =W_{T_sT_c}\times\left[\displaystyle\sum^{\infty}_{i=0}W^T_{T_sT_c}\cdot W_{T_sTc}\right]\times W^T_{T_sTc}
\end{array}
\end{equation}
In order to ensure that the matrix series converges, all of the commuting matrices in \ref{CommutingMatrix_source} are normalized according to Formula \ref{def_normalizedVersion}
and a decaying parameter $\lambda$ is used here. The Perron-Frobenius theorem is used here \cite{Roger:MatrixAnalysis}.
The normalized version of $\mathcal{M}_{RMP_s}$ is defined in Formula \ref{normalized_CommutingMatrix_source}.
\begin{equation}\label{normalized_CommutingMatrix_source}
\begin{array}{rl}
\overline{\mathcal{M}_{RMP_s}} & =\overline{W_{T_sT_c}}\times
\left[\displaystyle\sum_{i=0}^{\infty}\left(\lambda\cdot\overline{W^T_{T_sT_c}\cdot W_{T_sT_c}}\right)^i\right]\times
\overline{W^T_{T_sT_c}} \\
 & =\overline{W_{T_sT_c}}\times
\left(\mathbb{I}-\lambda\cdot\overline{W^T_{T_sT_c}\cdot W_{T_sT_c}}\right)^{-1}\times
\overline{W^T_{T_sT_c}}. \\
\end{array}
\end{equation}
where $\lambda\in(0,1)$ is called decaying parameter and $\mathbb{I}$ is the identity matrix with the same size as
$\overline{W^T_{T_sT_c}\cdot W_{T_sT_c}}$.
Note that $\overline{\mathcal{M}_{RMP_s}}$ may be diagonal. At this time, $RMP_s$ should be removed from all the
recurrent meta-paths because it can not provide any useful information for the similarities between source objects.

For the recurrent meta-tree $RMT_p$ shown in Formula \ref{DecompRecurPathStarMetaStrucNonsource}, e.g. Fig. \ref{dblp_decomposition}(d,g) and Fig. \ref{slap_decomposition}(g,h),
they can be collectively denoted as
\begin{equation}\label{DMP_notSource}
RMT_p=\left(T_s,T_{1,i_1},\cdots,T_{j,i_j},\overbrace{T_p,T_c}^{\infty}\right).
\end{equation}
Note that for each $T_p$ in $RMT_p$, its right side is also linked to the path $\left(T_p,T_{j,i_j},\cdots,T_{1,i_1},T_s\right)$,
see Fig. \ref{dblp_decomposition}(d,e) and Fig. \ref{slap_decomposition}(g,h).
In essence, $RMT_p$ can be decomposed into an infinite number of meta-paths such as
\begin{displaymath}
\begin{array}{c}
RMT^0_p=\left(T_s,T_{1,i_1},\cdots,T_{j,i_j},T_p,T_{j,i_j},\cdots,T_{1,i_1},T_s\right), \\
RMT^1_p=\left(T_s,T_{1,i_1},\cdots,T_{j,i_j},T_p,T_c,T_p,T_{j,i_j},\cdots,T_{1,i_1},T_s\right), \\
\cdots, \\
RMT^t_p=\left(T_s,T_{1,i_1},\cdots,T_{j,i_j},\overbrace{T_p,T_c,T_p}^{t},T_{j,i_j},\cdots,T_{1,i_1},T_s\right), \\
\cdots.
\end{array}
\end{displaymath}
The substructure $(T_p,T_c,T_p)$ recurs $t$ times in the meta-path $RMT^t_p, t=0,1,\cdots$.
Therefore, the commuting matrix of $RMT_p$ is defined as the summation of the commuting matrices of $RMT^t_p, t=0,1,\cdots$, see Formula \ref{CommutingMatrix_notSource}.
\begin{equation}\label{CommutingMatrix_notSource}
\begin{array}{rl}
\mathcal{M}_{RMT_p} &
=F_l\times\left[\displaystyle\sum_{i=0}^{\infty}\left(\lambda\cdot W_{T_pT_c}\cdot W_{T_pT_c}^T\right)^i
\right]\times F_r \\
& =F_l\times\left(\mathbf{I}-\lambda W_{T_pT_c}\cdot W^T_{T_pT_c}\right)^{-1}\times F_r,
\end{array}
\end{equation}
where
\begin{displaymath}
F_l=W_{T_sT_{1,i_1}}\cdot\left(\prod^{n-1}_{k=1}W_{T_{k,i_k}T_{k+1,i_{k+1}}}\right)\cdot W_{T_{j,i_j}T_p},
\end{displaymath}
and
\begin{displaymath}
F_r=W_{T_{j,i_j}T_p}^T\cdot\left(\prod^{n-1}_{k=1}W_{T_{n-k,i_{n-k}}T_{n-k+1,i_{n-k+1}}^T}\right)\cdot W_{T_sT_{1,i_1}}^T.
\end{displaymath}
In order to ensure that the matrix series converges, all of the commuting matrices in Formula \ref{CommutingMatrix_notSource} are normalized according to Formula \ref{def_normalizedVersion}
and the decaying parameter $\lambda$ is used as well. The Perron-Frobenius theorem is used here \cite{Roger:MatrixAnalysis}.
The normalized version of $\mathcal{M}_{RMT_p}$ is defined in Formula \ref{normalized_CommutingMatrix_notSource}.
\begin{equation}\label{normalized_CommutingMatrix_notSource}
\begin{array}{rl}
\overline{\mathcal{M}_{RMT_p}} &
=\overline{F_l}\times\left[\displaystyle\sum_{i=0}^{\infty}\left(\lambda\cdot\overline{W_{T_pT_c}\cdot W_{T_pT_c}^T}\right)^i
\right]\times\overline{F_r} \\
& =\overline{F_l}\times\left(\mathbf{I}-\lambda\overline{W_{T_pT_c}\cdot W^T_{T_pT_c}}\right)^{-1}\times\overline{F_r},
\end{array}
\end{equation}
where
\begin{displaymath}
\overline{F_l}=\overline{W_{T_sT_{1,i_1}}}\cdot\left(\prod^{n-1}_{k=1}\overline{W_{T_{k,i_k}T_{k+1,i_{k+1}}}}\right)\cdot\overline{W_{T_{j,i_j}T_p}},
\end{displaymath}
and
\begin{displaymath}
\overline{F_r}=\overline{W_{T_{j,i_j}T_p}^T}\cdot\left(\prod^{n-1}_{k=1}\overline{W_{T_{n-k,i_{n-k}}T_{n-k+1,i_{n-k+1}}}^T}\right)\cdot\overline{W_{T_sT_{1,i_1}}^T}.
\end{displaymath}

Both the recurrent meta-paths and recurrent meta-trees only consider the structure of the network schema $\Theta_G$, but ignore the structure of the HIN $G$.
In fact, they play different roles in the HIN due to the sparsity and strength of their instances, i.e. the sparsity and strength of the entries of their commuting matrices.
Therefore, we should combine the commuting matrices of different recurrent meta-paths according to different weights.
Below, we introduce two kinds of strategies, global weighting strategy and local weighting strategy, to determine these weights.

The global weighting strategy is to determine the weight of a recurrent meta-path or recurrent meta-tree by the strength of its commuting matrix, i.e. the sum of all the entries of the commuting matrix.
The local weighting strategy is to determine the weight of a recurrent meta-path or recurrent meta-tree by the sparsity of its instances.
Take the recurrent meta-tree $RMT_p$ shown in the Formula \ref{DMP_notSource} as an example.
We traverse the objects belonging to $T_p$ for $N$ times, and then randomly sample an object from their neighbors.
The drawn object must belong to $T_{c_1},\cdots T_{c_{m-1}}$ or $T_{c_m}$.
Let $Num_{T_{c_i}}$ denote the number of the drawn objects belonging to $T_{c_i}$. The frequency from $T_p$ to $T_{c_i}$ is equal to $\omega_{T_pT_{c_i}}=\frac{Num_{T_{c_i}}}{N}$.
As a result, the weight of the recurrent meta-tree $RMT_p$ is equal to
\begin{equation}\label{weights}
\omega_{RMT_p}=\omega_{T_sT_1}\left(\prod^{n-1}_{i=1}\omega_{T_iT_{i+1}}\right)\omega_{T_nT_p}\omega_{T_pT_c}.
\end{equation}

The proposed similarity measure RMSS is defined as,
\begin{equation}\label{DMSS_def}
RMSS(o_s,o_t)=\frac{\mathcal{U}(o_s,o_t)}{\mathcal{U}(o_s,o_s)},
\end{equation}
where
\begin{displaymath}
\mathcal{U}=\sum_{T_p\in\mathcal{I}}\omega_{RMX_p}\overline{\mathcal{M}_{RMX_p}}.
\end{displaymath}
Note that $RMX_p=RMP_s$ when $T_p=T_s$ and $RMX_p=RMT_p$ when $T_p\neq T_s$.
Obviously, the proposed RMSS is asymmetric. In reality, a lot of similarities are asymmetric (i.e. directed).
Take the bibliographic information network as example. Yizhou Sun is similar to Jiawei Han because she is one of Professor Han's students.
However, Professor Han is relatively less similar to Yizhou Sun because Han has so many students that his similarities are allocated to all the students.

\textbf{NOTE.} According to Formula \ref{DMSS_def}, we only need to compute $\omega_{RMX_p}$ and $\overline{\mathcal{M}_{RMX_p}}$ for each recurrent meta-path or recurrent meta-tree.
This means we can employ the distributed computing techniques to speed up the computation. All the recurrent meta-paths and recurrent meta-trees can be separately assigned to different computation nodes.
For each recurrent meta-path or recurrent meta-tree, we can employ Graphics Processing Unit (\textit{GPU}) to speed up the matrix operations.

\subsection{Algorithm Description}\label{subsec:algorithmdescription}

\begin{algorithm}\footnotesize
\caption{Computing RMSS}
\begin{algorithmic}[1]
\label{DMSSALG}
\renewcommand{\algorithmicrequire}{\textsc{Input:}}
\renewcommand{\algorithmicensure}{\textsc{Output:}}
\REQUIRE{HIN $G$, Network Schema $\Theta_G$, Source Object Type $T_s$, Decaying Parameter $\lambda\in(0,1)$, weight\_type.}
\ENSURE{Similarity Matrix RMSS}
\STATE \textbf{Offline Computing:}
\STATE $DSL\leftarrow\text{DecompRMS}(\Theta_G,T_s)$;
\FOR{$dsl\in DSL$}
    \STATE Compute $\mathcal{M}_{dsl}$ via Formulas \ref{normalized_CommutingMatrix_source} or \ref{normalized_CommutingMatrix_notSource};
\ENDFOR
\STATE $dmp\_weight\_dict\leftarrow\left\{(dsl,\emptyset): dsl\in DSL\right\}$;
\IF{weight\_type=``global"}
    \STATE Compute $dmp\_weight\_dict$ via the global weighting strategy;
\ELSE
    \STATE Compute $dmp\_weight\_dict$ via the local weighting strategy;
\ENDIF
\STATE \textbf{Online Computing:}
\STATE Compute $RMSS(o_s,o_t)$ for any two objects $o_s$ and $o_t$ via Formula \ref{DMSS_def}.
\RETURN RMSS.
\end{algorithmic}
\end{algorithm}

In this section, we present the pseudo-code of the algorithm for computing RMSS, see algorithm \ref{DMSSALG}.
The algorithm includes two parts: offline part and online part. The offline part (from line 2 to line 11 in algorithm \ref{DMSSALG}) takes responsibility for
1) periodically computing the commuting matrices of the recurrent meta-paths and recurrent meta-trees;
2) computing the weights of the recurrent meta-paths and recurrent meta-trees. The online part (line 13 in algorithm \ref{DMSSALG}) takes responsibility for computing the similarity matrix RMSS.

Algorithm \ref{DMSSALG} spends most of the time on the offline part. This part spends most of the time on lines 3-5 involving a number of matrix operations, e.g. matrix multiplication, matrix inverse.
For a $s\times s$ matrix, the time complexity of calculating its inverse is $O(s^3)$. For a $r\times s$ matrix and a $s\times t$ matrix, the time complexity of multiplying them is $O(r\times s\times t)$.
For the recurrent meta-path shown in Formula \ref{DMP_source}, it takes
\begin{displaymath}
O\left((|T_s|+|T_c|)\cdot|T_c|^2\right)
\end{displaymath}
to compute its commuting matrix.
For the recurrent meta-tree shown in Formula \ref{DMP_notSource},
it takes
\begin{displaymath}
O\left(|T_s|\cdot|T_1|\cdot|T_2|+\sum_{i=1}^{n-2}|T_i|\cdot|T_{i+1}|\cdot|T_{i+2}|+|T_{n-1}|\cdot|T_n|\cdot|T_p|+|T_p|^3\right)
\end{displaymath}
to compute its commuting matrix.
As a result, the time complexity of algorithm \ref{DMSSALG} is the maximum of the above two terms.

Below, we take the HIN shown in Fig. \ref{example_hins} as an example to present the procedure of computing RMSS.
As shown in Fig. \ref{dblp_decomposition}, its RecurMS
can be decomposed into a recurrent meta-path (see \ref{dblp_decomposition}(b)) and a recurrent path-star meta-structure (see \ref{dblp_decomposition}(c)).
The recurrent path-star meta-structure can be further decomposed two recurrent meta-trees, see Fig. \ref{dblp_decomposition}(d,e).
The commuting matrix of the recurrent meta-path shown in Fig. \ref{dblp_decomposition}(b) is a diagonal one because each paper can only be published in a single venue.
Therefore, The commuting matrix should be ignored because it can not provide any useful information for the similarity measure.
The commuting matrix of the recurrent meta-path shown in Fig. \ref{dblp_decomposition}(d) is
\begin{displaymath}
\mathcal{M}_A=
\begin{bmatrix}
1.26397 & 0.17306 & 0.02962 & 0.53333 \\
0.17306 & 1.26397 & 0.02962 & 0.53333 \\
0.05925 & 0.05925 & 1.24814 & 0.63333 \\
0.15555 & 0.15555 & 0.08888 & 1.60000 \\
\end{bmatrix}
.
\end{displaymath}
Similarly, the commuting matrix of the recurrent meta-path shown in Fig. \ref{dblp_decomposition}(e) is
\begin{displaymath}
\mathcal{M}_T=
\begin{bmatrix}
1.28688 & 0.04076 & 0.21267 & 0.45967 \\
0.02912 & 1.22955 & 0.16789 & 0.57342 \\
0.09666 & 0.10684 & 1.24700 & 0.54948 \\
0.07520 & 0.14452 & 0.20369 & 1.57658 \\
\end{bmatrix}
.
\end{displaymath}
For the global weighting strategy, the corresponding similarity matrix is
\begin{displaymath}
RMSS\_global=
\begin{bmatrix}
1.00000 & 0.08382 & 0.09498 & 0.38928 \\
0.08108 & 1.00000 & 0.07921 & 0.44385 \\
0.06249 & 0.06657 & 1.00000 & 0.47404 \\
0.07264 & 0.09446 & 0.09210 & 1.00000 \\
\end{bmatrix}
.
\end{displaymath}
For the local weighting strategy, the corrsponding similarity matrix is
\begin{displaymath}
RMSS\_local=
\begin{bmatrix}
1.00000 & 0.07575 & 0.10586 & 0.38431 \\
0.07236 & 1.00000 & 0.08792 & 0.44727 \\
0.06480 & 0.06950 & 1.00000 & 0.46891 \\
0.06883 & 0.09403 & 0.09776 & 1.00000 \\
\end{bmatrix}
.
\end{displaymath}
The rows/columns of RMSS respectively represent `AAAI', `KDD', `TKDE' and `VLDB'.

\section{Experimental Evaluations}\label{sec:expriment}

In this section, we compare RMSS with the state-of-the-art metrics on three real datasets.
As similarly as the literatures \cite{SHYYW:2011,HZCSML:2016}, we also employ the ranking quality and the clustering quality to evaluate the goodness of metrics.
The configuration of my PC is Intel(R) Core(TM) i5-4570 CPU @ 3.20GHz and RAM 12GB.

\subsection{Evaluation Metrics}\label{subsec:comparison_metrics}

For the ranking task, we choose a popular comparison metric which is also used in papers \cite{SHYYW:2011,HZCSML:2016},
called Normalized Discounted Cumulative Gain ($nDCG$, the bigger its value, the better the ranking) \cite{nDCG:2013}, to evaluate the quality of ranking.
$nDCG$ is defined in Formula \ref{def_ndcg}.
\begin{equation}\label{def_ndcg}
nDCG=\frac{DCG}{iDCG},
\end{equation}
where
\begin{displaymath}
DCG=\sum^n_{j=1}\frac{2^{r(j)-1}}{\log(1+j)}.
\end{displaymath}
and $iDCG$ is the ideal $DCG$. Note that $iDCG$ is calculated according to the ideal ranking result.

For the clustering task, we also choose a popular comparison metric which is used in papers \cite{SHYYW:2011,HZCSML:2016},
called Normalized Mutual Information ($NMI$, the bigger its value, the better the clustering) \cite{SHYYW:2011}, to evaluate the quality of clustering.
$NMI$ is defined in Formula \ref{def_nmi}.
\begin{equation}\label{def_nmi}
NMI=\frac{2\times I(\Omega,C)}{H(\Omega)+H(C)},
\end{equation}
where
\begin{displaymath}
I(\Omega,C)=\sum_{k}\sum_{j}\frac{|\omega_k\cap c_j|}{N}\log\frac{N|\omega_k\cap c_j|}{|\omega_k||c_j|},
\end{displaymath}
and
\begin{displaymath}
H(\Omega)=-\sum_{k}\frac{|\omega_k|}{N}\log\frac{|\omega_k|}{N}.
\end{displaymath}
$\Omega=\left\{\omega_1,\cdots,\omega_K\right\}$ is the set of clusters, and $C=\left\{c_1,\cdots,c_J\right\}$ is the set of classes.

\subsection{Datasets}\label{subsec:dataset}

Three real datasets, respectively called \textbf{DBLPc}, \textbf{DBLPr} and \textbf{BioIN}, are used here.
The first two are extracted from \textbf{DBLP}\footnote{http://dblp.uni-trier.de/db/}.
The last is extracted from \textbf{Chem2Bio2RDF} \cite{CDJWZ:2010,FDSCSB:2016}.
They are summarized in table \ref{datasetSummarization}
DBLPc includes 21 venues coming from four areas: database, data mining, information retrieval and machine learning, 25858 papers, 25780 authors and 10436 terms.
DBLPr includes 20 venues, 23906 papers, 24078 authors and 9862 terms.
Their network schema is shown in Fig. \ref{example_network_schema}(a).
BioIN includes 2018 genes, 300 tissues, 4331 gene ontology instances, 224 substructures, 712 side effects and 18097 chemical compounds.
Its network schema is shown in Fig. \ref{example_network_schema}(b).
Note in particular that we only consider the genes assigned to a single cluster here because we use k-means algorithm to cluster the genes.
The RecurMS for DBLPc and DBLPr is shown in Fig. \ref{dblp_decomposition}(a),
and for BioIN is shown in Fig.  \ref{slap_decomposition}(a).

\begin{table}\footnotesize
  \centering
  \caption{Dataset Summarization.}
  \begin{tabular}{c|c|c}
    \hline
     \textbf{Dataset}                & \textbf{Object Type} & \textbf{Num} \\
    \hline
     \multirow{4}{*}{DBLPr} & $Paper$              & 23906 \\ \cline{2-3}
                                     & $Venue$              & 20 \\ \cline{2-3}
                                     & $Term$               & 9862 \\ \cline{2-3}
                                     & $Author$               & 24078 \\
    \hline
    \hline
     \multirow{4}{*}{DBLPc} & $Paper$              & 25858 \\ \cline{2-3}
                                     & $Venue$              & 21 \\ \cline{2-3}
                                     & $Term$               & 10436 \\ \cline{2-3}
                                     & $Author$               & 25780 \\
    \hline
    \hline
     \multirow{6}{*}{BioIN} & $GeneOntology$       & 4331 \\ \cline{2-3}
                                     & $Tissue$             & 300 \\ \cline{2-3}
                                     & $Gene$               & 2018 \\ \cline{2-3}
                                     & $ChemicalCompound$   & 18097 \\ \cline{2-3}
                                     & $SideEffect$         & 712 \\ \cline{2-3}
                                     & $Substructure$       & 224 \\ \cline{2-3}
    \hline
  \end{tabular}
  \label{datasetSummarization}
\end{table}

\subsection{Baselines}\label{subsec:baseline}

In this paper, RMSS is compared with three state-of-the-art similarity metrics: BSCSE \cite{HZCSML:2016}, BPCRW \cite{LC:2010a,LC:2010b}, PathSim \cite{SHYYW:2011}.
Let $\mathcal{P}$ and $\mathcal{S}$ respectively denote a meta-path and a meta-structure. For a given source-target object pair $(o_s,o_t)$, they are defined as follows.
\begin{displaymath}
BSCSE(g,i|\mathcal{S},o_t)=\frac{\textstyle\sum_{g'\in\sigma(g,i|\mathcal{S},G)}BSCSE(g',i+1|\mathcal{S},o_t)}{|\sigma(g,i|\mathcal{S},G)|^{\alpha}}.
\end{displaymath}
\begin{displaymath}
BPCRW(o,o_t|\mathcal{P})=\frac{\sum_{o'\in N_{\mathcal{P}}(o)}BPCRW(o',o_t|\mathcal{P})}{|N_{\mathcal{P}}(o)|^{\alpha}}.
\end{displaymath}
\begin{displaymath}
PathSim(o_s,o_t|\mathcal{P})=\frac{2\times\mathcal{M}_{\mathcal{P}}(o_s,o_t)}{\mathcal{M}_{\mathcal{P}}(o_s,o_s)+\mathcal{M}_{\mathcal{P}}(o_t,o_t)}.
\end{displaymath}

In these definitions, $\alpha$ is a biased parameter.
For BSCSE, $\sigma(g,i|\mathcal{S},G)$ denotes the ($i+1$)-th layer's instances expanded from $g\in\mathcal{S}[1:i]$ on $G$ \cite{HZCSML:2016}.
For BPCRW, $N_{\mathcal{P}}(o)$ denotes the neighbors of $o$ along meta-path $\mathcal{P}$ \cite{LC:2010b,LC:2010a}.
For PathSim, $\mathcal{M}_{\mathcal{P}}$ denotes the commuting matrix of the meta-path $\mathcal{P}$ \cite{SHYYW:2011}.

BSCSE and BPCRW involve a biased parameter $\alpha$. In this paper, $\alpha$ is respectively set to 0.1, 0.3, 0.5, 0.7 and 0.9.
For the proposed RMSS, its decaying parameter $\lambda$ is set to 0.1, 0.2, 0.3, 0.4, 0.5, 0.6, 0.7, 0.8, 0.9.
Below we respectively evaluate the values of $NMI$ and $nDCG$ yielded by RMSS, PathSim, $BPACRW$ and BSCSE under these parameter settings.

%

\subsection{Sensitivity Analysis}\label{subsec:sensitivity}

As stated previously, the proposed RMSS is robust (insensitive) to schematic structures, because it combines all the possible schematic structures in the form of recurrent meta-paths and recurrent meta-trees.
Before proceeding, we first evaluate the sensitivity of the state-of-the-art metrics (PathSim, BPCRW and BSCSE) in terms of clustering task and ranking task.
Specifically, PathSim and BPCRW take a meta-path as input, and BSCSE take a meta-structure as input.
When they are feeded different schematic structures, i.e. different meta-paths or meta-structures, we will show the fluctuation of the clustering quality ($NMI$) and the ranking quality ($nDCG$).

\subsubsection{Sensitivity in terms of clustering quality}\label{subsec:insensitivity_nmi}

Two meta-paths $(V,P,A,P,V)$ and $(V,P,T,P,V)$ on DBLPc are selected for PathSim and BPCRW.
Then, we compare the $NMI$ values for these two meta-paths under different biased parameters $\alpha$.
On BioIN, four meta-paths $(G,GO,G)$, $(G,T,G)$, $(G,CC,Si,CC,G)$ and $(G,CC,Sub,CC,G)$ are selected for PathSim and BPCRW.
Two meta-structures $(G,(GO,T),G)$ and $(G,CC,(Si,Sub),CC,G)$ are selected for BSCSE.
Then, we compare the $NMI$ values for these schematic structures under different biased parameters $\alpha$.
Note in particular PathSim does not depend on any parameters. Therefore, its lines for different meta-paths are always parallel to x-axis.

\begin{figure}[htb]
  \centering
  \includegraphics[width=0.6\textwidth]{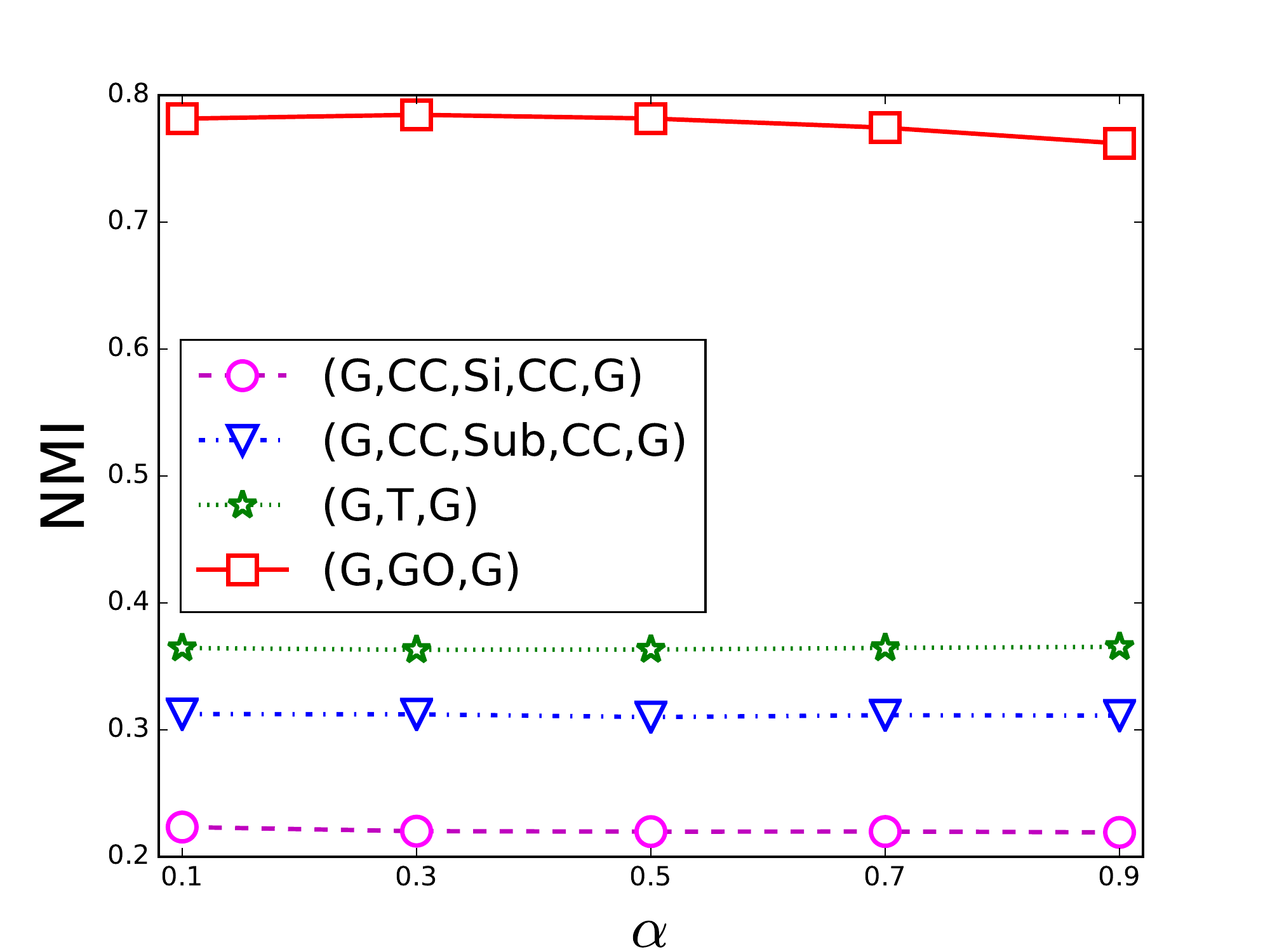}
  \caption{Sensitivity of BPCRW to different schematic structures on BioIN in terms of clustering.}
  \label{insensitivity_slap_bpcrw}
\end{figure}

\begin{figure}[htb]
  \centering
  \includegraphics[width=0.6\textwidth]{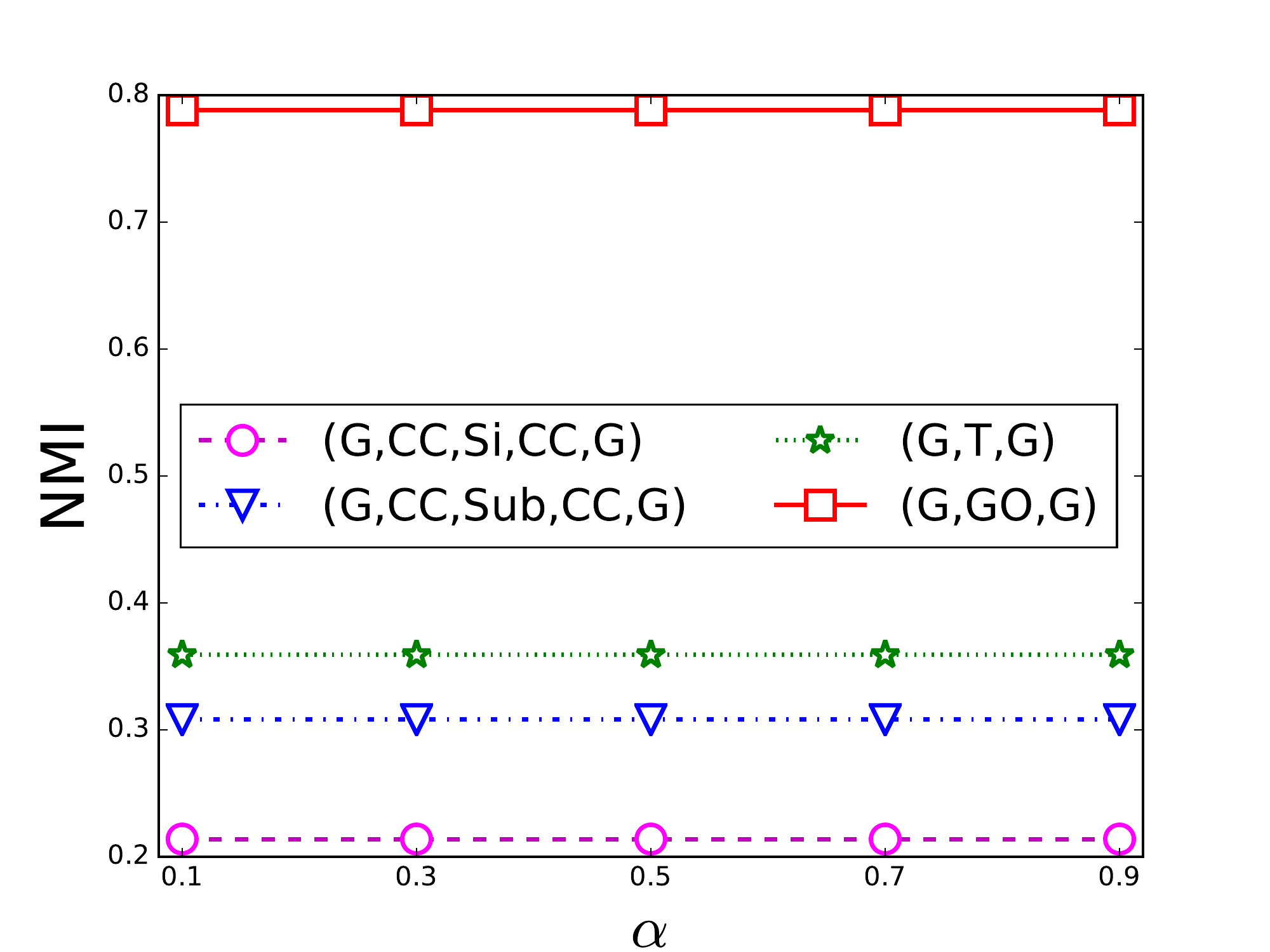}
  \caption{Sensitivity of PathSim to different schematic structures on BioIN in terms of clustering.}
  \label{insensitivity_slap_pathsim}
\end{figure}

\begin{figure}[htb]
  \centering
  \includegraphics[width=0.6\textwidth]{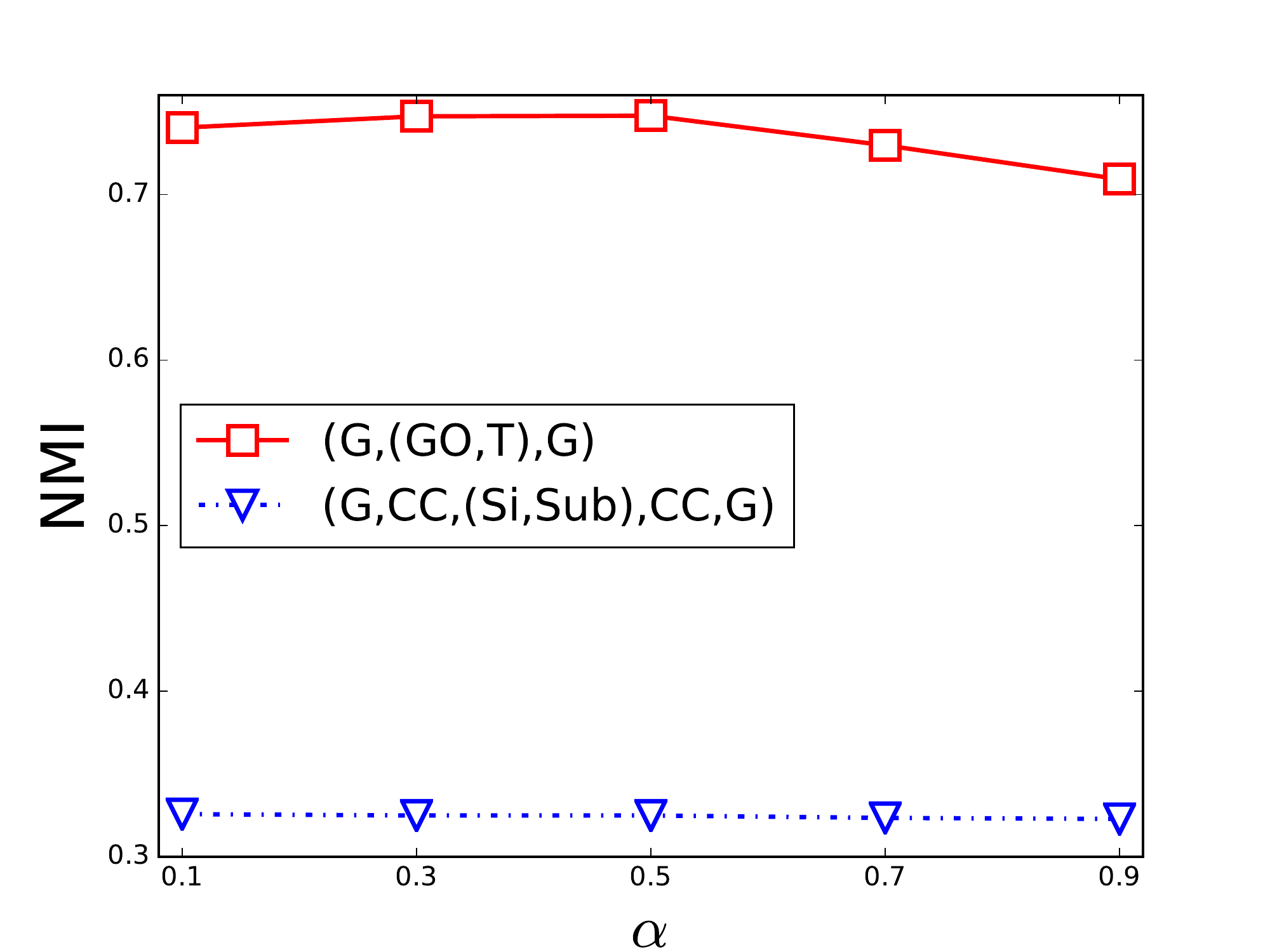}
  \caption{Sensitivity of BSCSE to different schematic structures on BioIN in terms of clustering.}
  \label{insensitivity_slap_bscse}
\end{figure}

Fig. \ref{insensitivity_slap_bpcrw} and Fig. \ref{insensitivity_slap_pathsim} respectively show the $NMI$ values under different $\alpha$ for BPCRW and PathSim respectively with different meta-paths on BioIN.
Fig. \ref{insensitivity_slap_bscse} shows the $NMI$ values under different $\alpha$ for BSCSE with different meta-structures on BioIN.
For PathSim and BPCRW, the $NMI$ values with $(G,GO,G)$ are much larger than those with the other meta-paths.
For BSCSE, the $NMI$ values with $(G,(GO,T),G)$ are much larger than
that with $(G,CC,(Si,Sub),CC,G)$. This reveals PathSim and BPCRW are sensitive to meta-paths, and BSCSE is sensitive to meta-structures.

\begin{figure}[htb]
  \centering
  \includegraphics[width=0.6\textwidth]{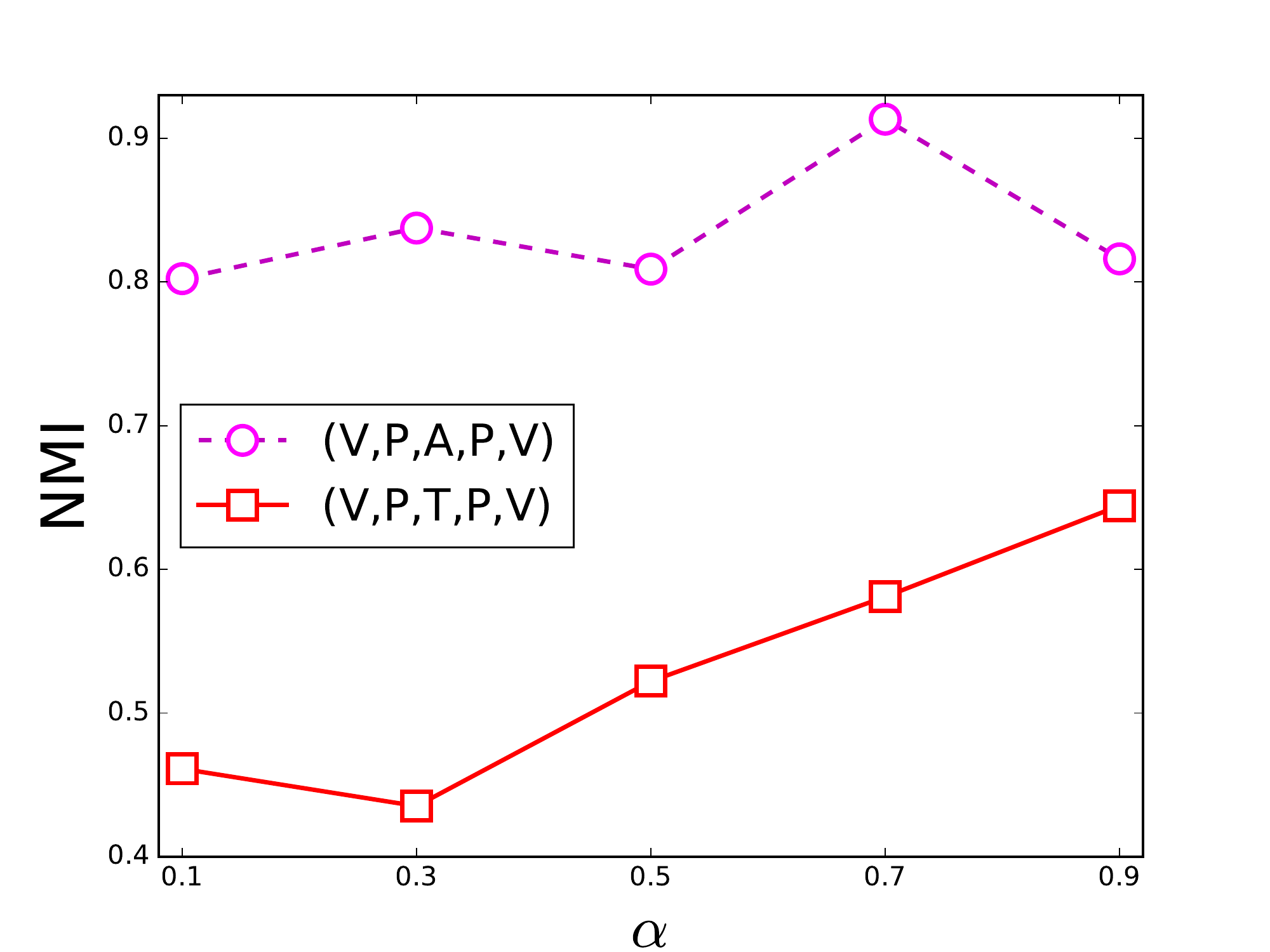}
  \caption{Sensitivity of BPCRW to different meta-paths DBLPc in terms of clustering.}
  \label{insensitivity_dblp_bpcrw}
\end{figure}

\begin{figure}[htb]
  \centering
  \includegraphics[width=0.6\textwidth]{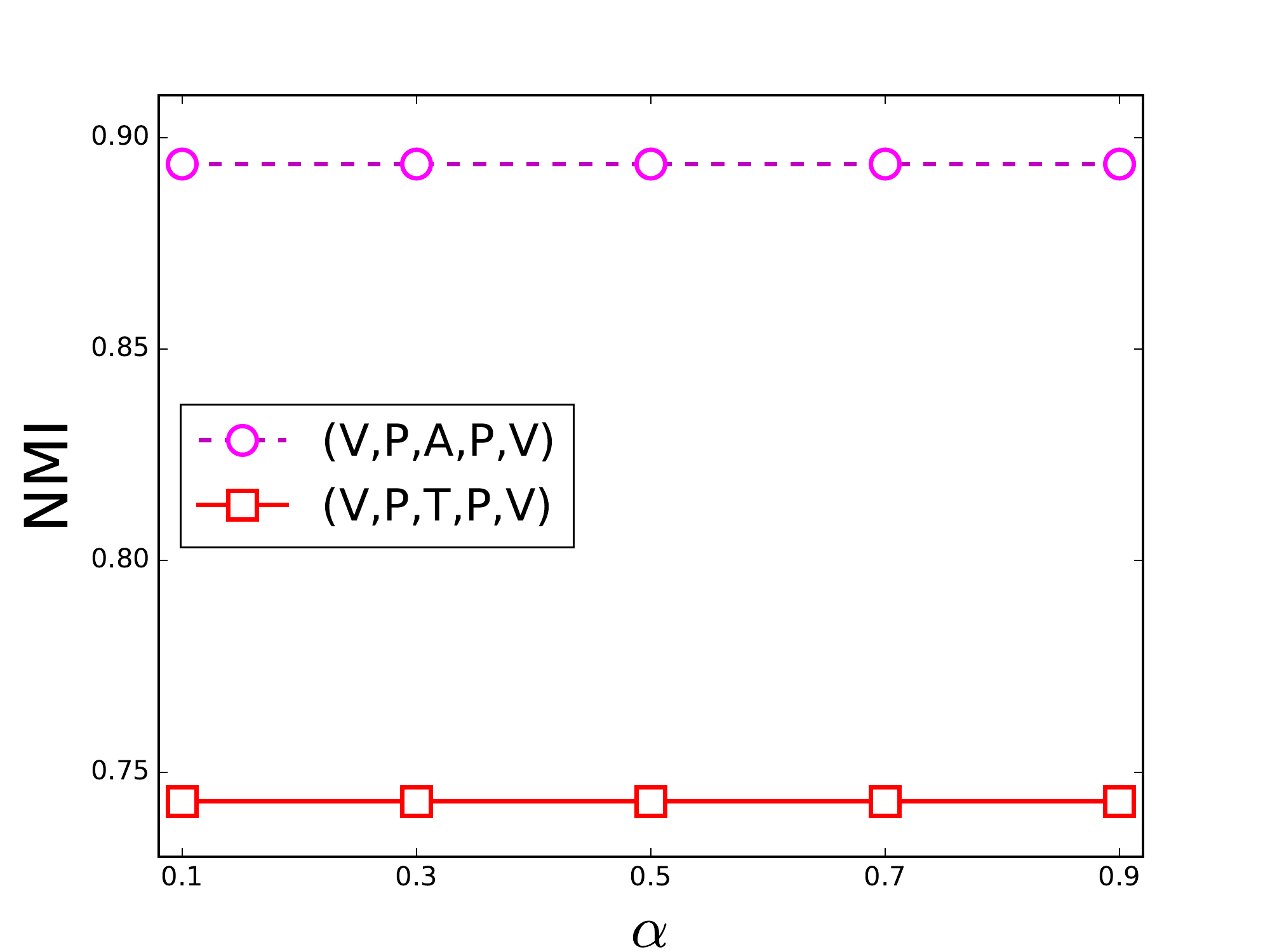}
  \caption{Sensitivity of PathSim to different meta-paths on DBLPc in terms of clustering.}
  \label{insensitivity_dblp_pathsim}
\end{figure}

Fig. \ref{insensitivity_dblp_bpcrw} and \ref{insensitivity_dblp_pathsim} respectively show the $NMI$ values
under different $\alpha$ for PathSim and BPCRW with different meta-paths on DBLPc.
Note in particular that we do not consider the sensitivity of BSCSE to different meta-structures because the meta-structure $(V,P,(A,T),P,V)$ is most frequently used on DBLPc.
According to these figures, we know that BPCRW and PathSim with the meta-path $(V,P,A,P,V)$ achieve significantly better clustering quality that those with the meta-path $(V,P,T,P,V)$.
That is to say, BPCRW and $PahtSim$ are also sensitive to different meta-paths on DBLPr.

To be summarized, all of the baselines are sensitive to different schematic structures on BioIN and DBLPc.
This suggests that it is important for the baselines to select an appropriate schematic structures.
The proposed metric RMSS does not depend on any schematic structures. This is the biggest advantage of RMSS relative to the baselines.

\subsubsection{Sensitivity in terms of ranking quality}\label{subsec:insensitivity_ndcg}

\begin{figure}[htb]
  \centering
  \includegraphics[width=0.6\textwidth]{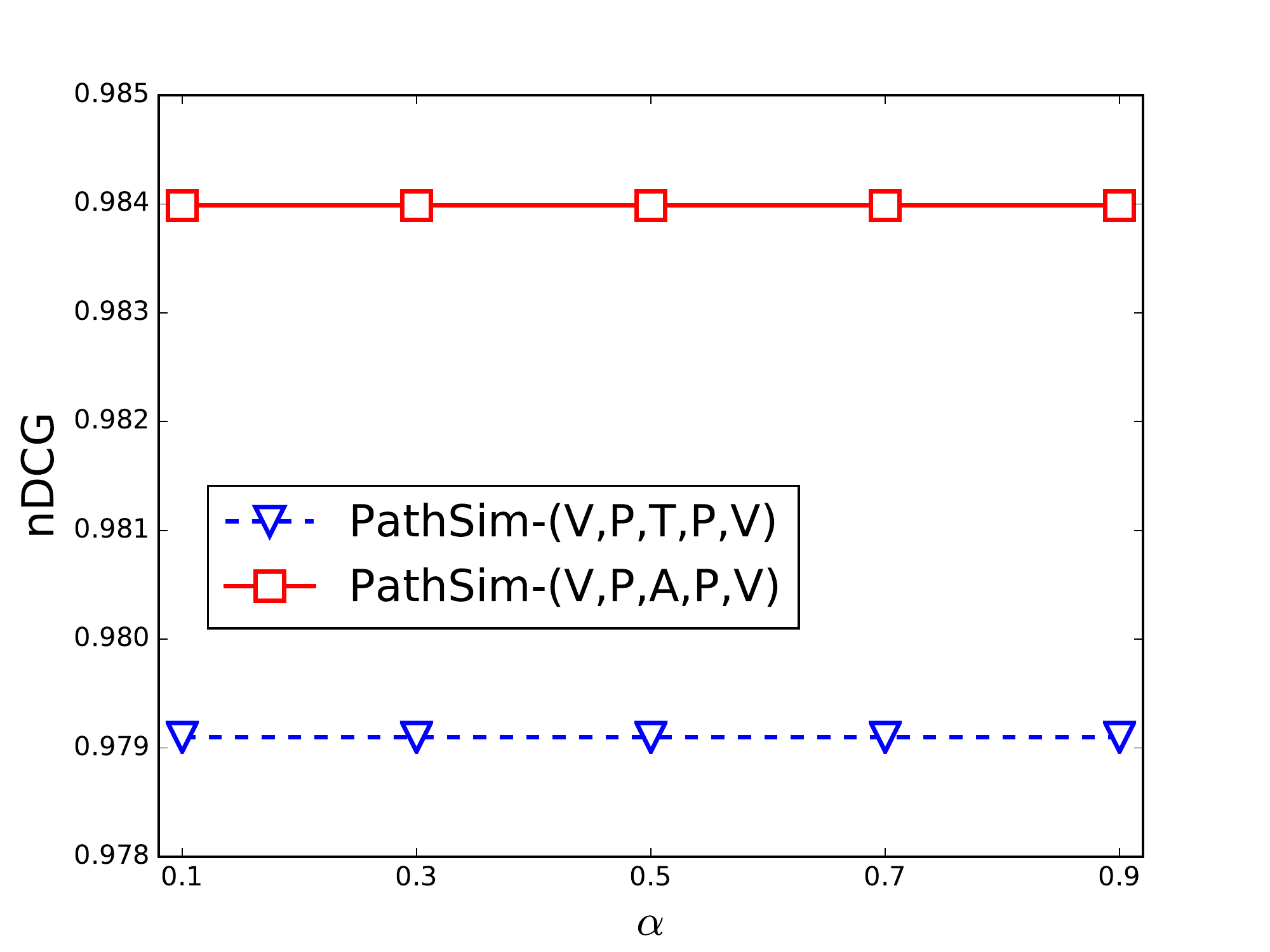}
  \caption{Sensitivity of PathSim to different meta-paths on DBLPr in terms of ranking with CIKM as the source object.}
  \label{insensitivity_pathsim_ndcg_kdd}
\end{figure}

\begin{figure}[htb]
  \centering
  \includegraphics[width=0.6\textwidth]{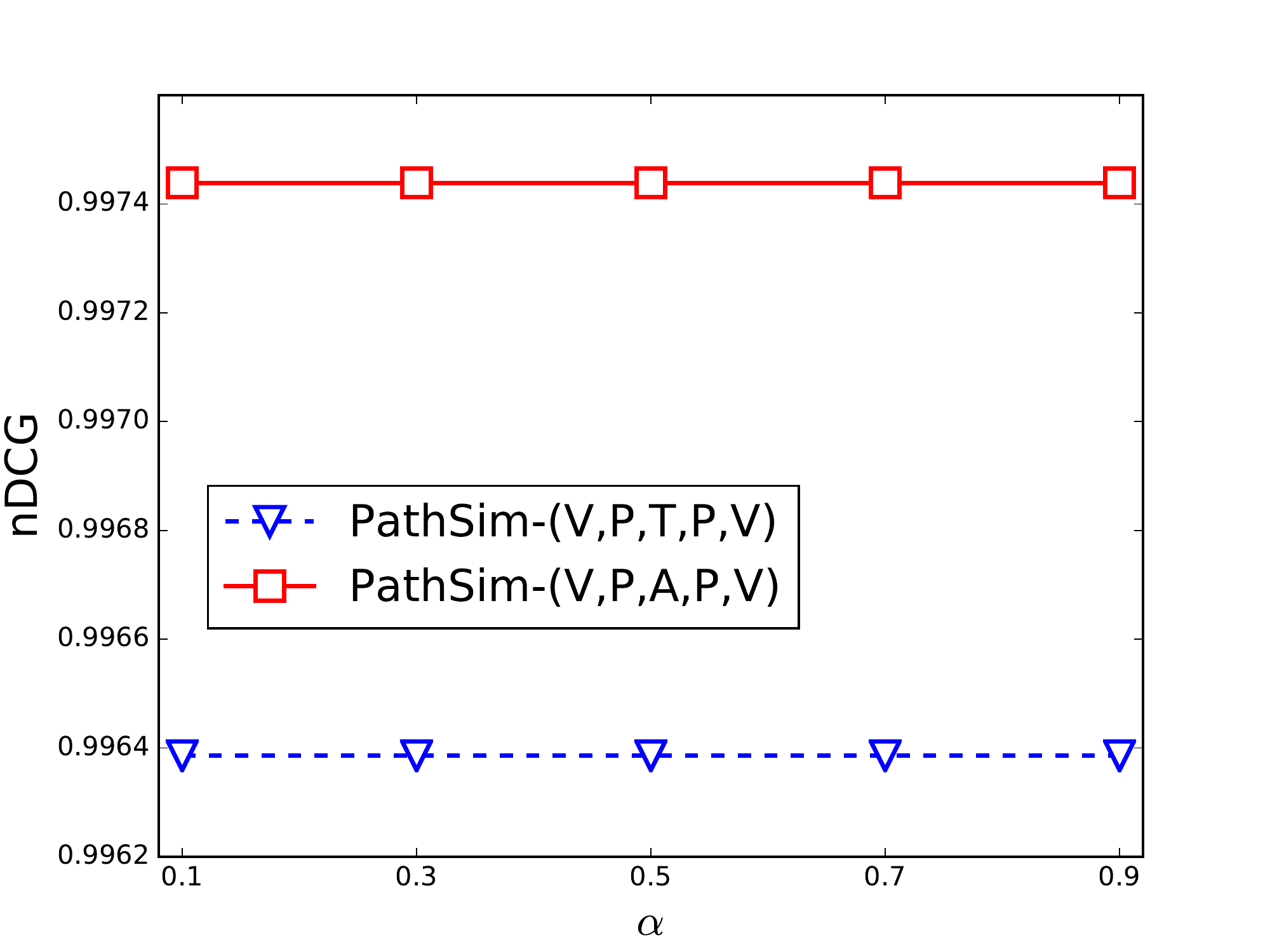}
  \caption{Sensitivity of PathSim to different meta-paths on DBLPr in terms of ranking with SIGMOD as the source object.}
  \label{insensitivity_pathsim_ndcg_sigmod}
\end{figure}

\begin{figure}[htb]
  \centering
  \includegraphics[width=0.6\textwidth]{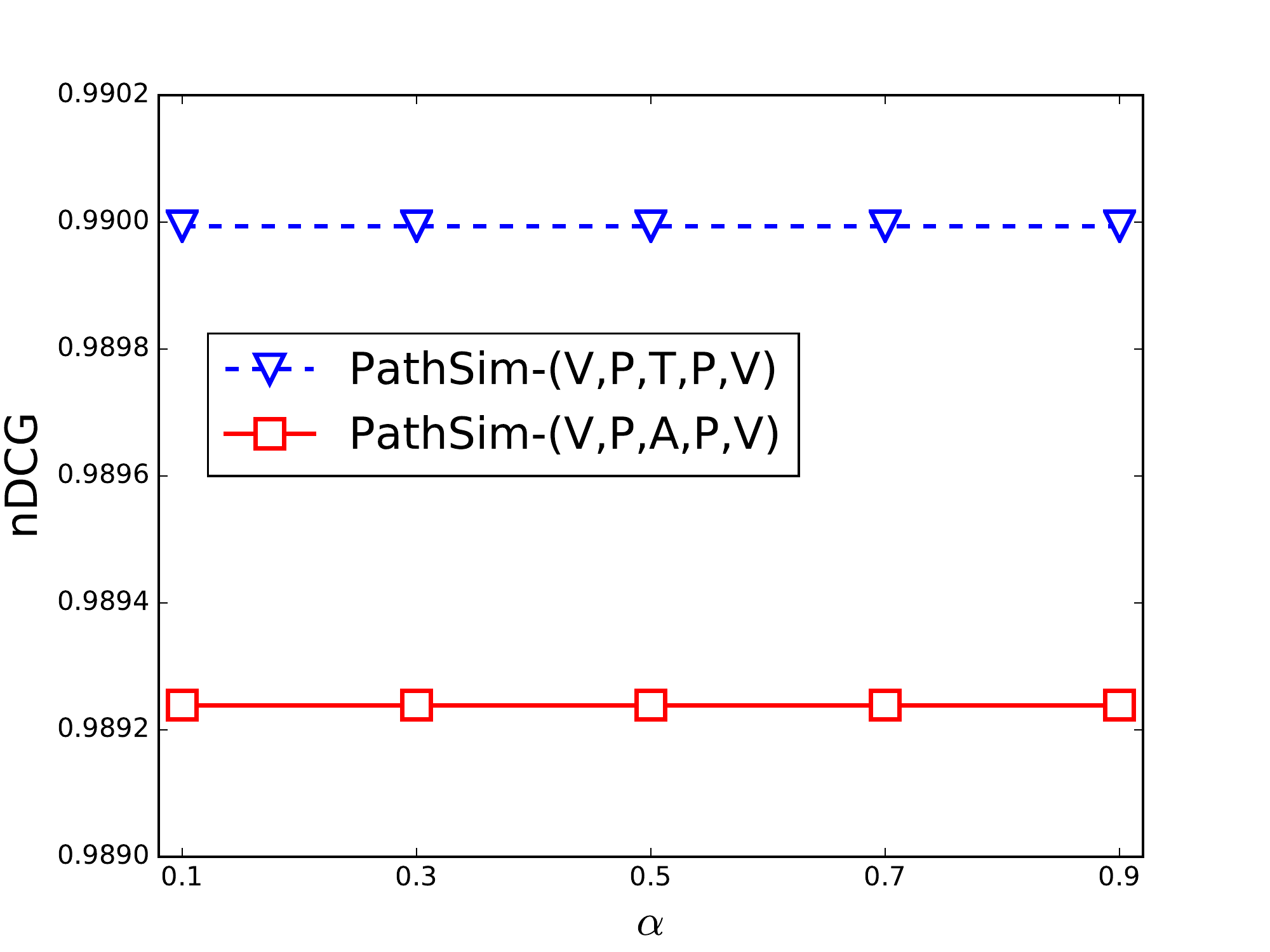}
  \caption{Sensitivity of PathSim to different meta-paths on DBLPr in terms of ranking with TKDE as the source object.}
  \label{insensitivity_pathsim_ndcg_tkde}
\end{figure}

//
Here, we investigate whether $Pathim$ and BPCRW are sensitive to different meta-paths in terms of ranking task as similarly
as the section \ref{subsec:insensitivity_nmi}.
Fig. \ref{insensitivity_pathsim_ndcg_kdd}, Fig. \ref{insensitivity_pathsim_ndcg_sigmod} and Fig. \ref{insensitivity_pathsim_ndcg_tkde}
respectively show
the $nDCG$ values of PathSim with different meta-paths when respectively selecting CIKM, SIGMOD and TKDE  as the source objects.
According to these figures, we know that 1) the $nDCG$ values for $(V,P,A,P,V)$ are a little larger than that for $(V,P,T,P,V)$ respectively with CIKM  and SIGMOD be the source objects;
2) the $nDCG$ values for $(V,P,T,P,V)$ are a little larger than that for $(V,P,A,P,V)$ respectively with TKDE be the source object.
That is to say, different meta-paths for PathSim yield different $nDCG$ values.

\begin{figure}[htb]
  \centering
  \includegraphics[width=0.6\textwidth]{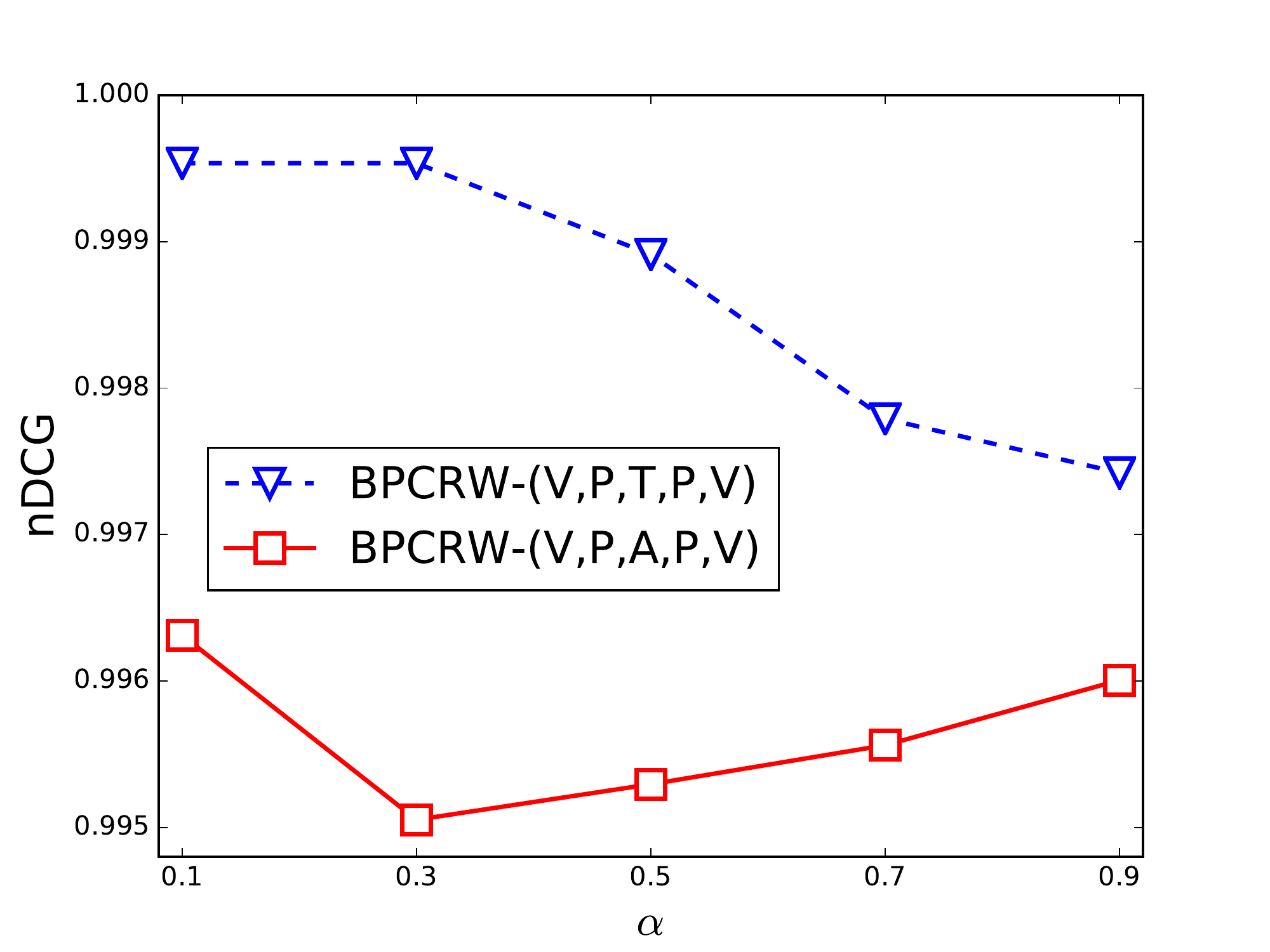}
  \caption{Sensitivity of BPCRW to different meta-paths on DBLPr in terms of ranking with CIKM as the source object.}
  \label{insensitivity_bpcrw_ndcg_kdd}
\end{figure}

\begin{figure}[htb]
  \centering
  \includegraphics[width=0.6\textwidth]{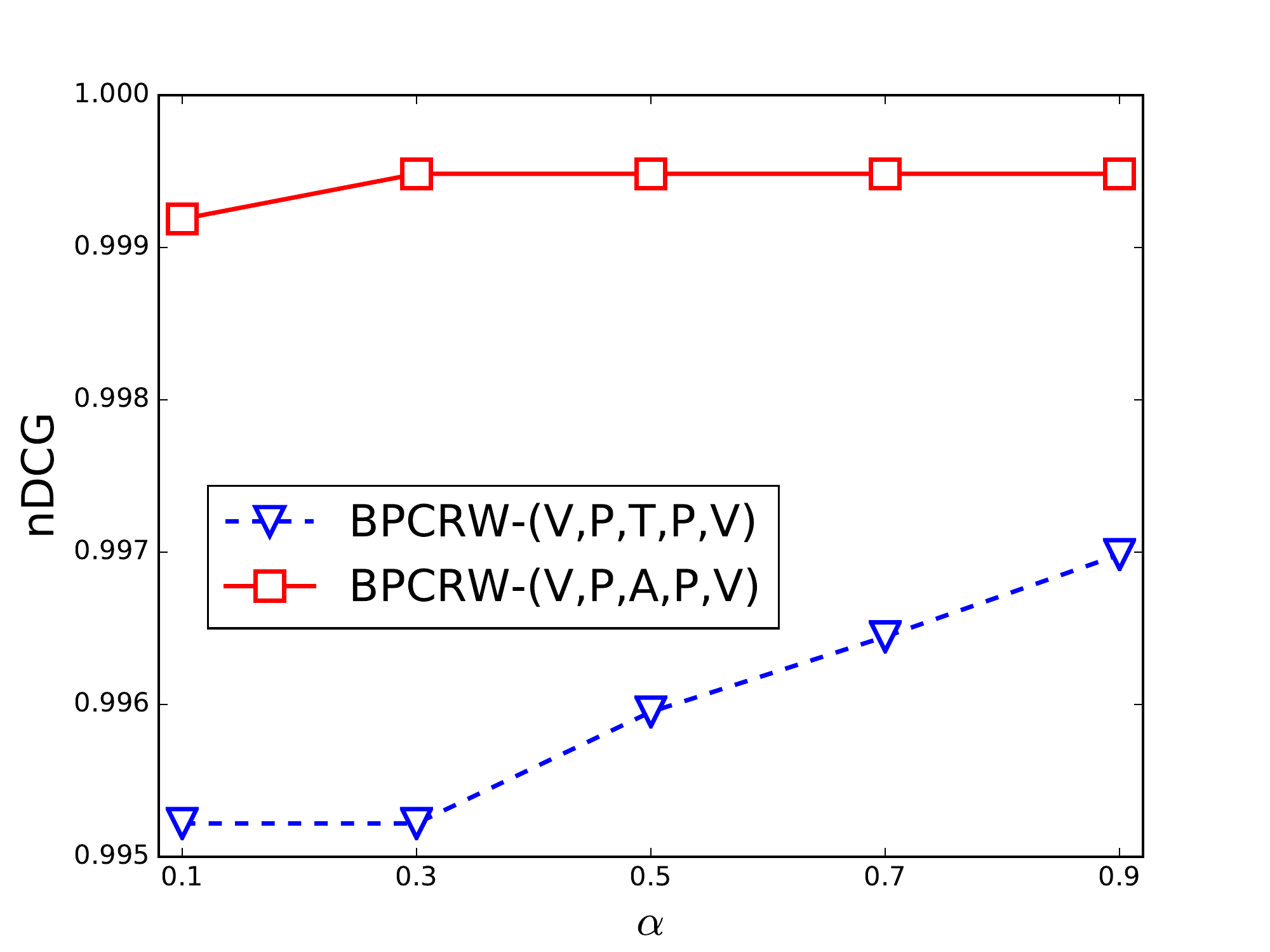}
  \caption{Sensitivity of BPCRW to different meta-paths on DBLPr in terms of ranking with SIGMOD as the source object.}
  \label{insensitivity_bpcrw_ndcg_sigmod}
\end{figure}

\begin{figure}[htb]
  \centering
  \includegraphics[width=0.6\textwidth]{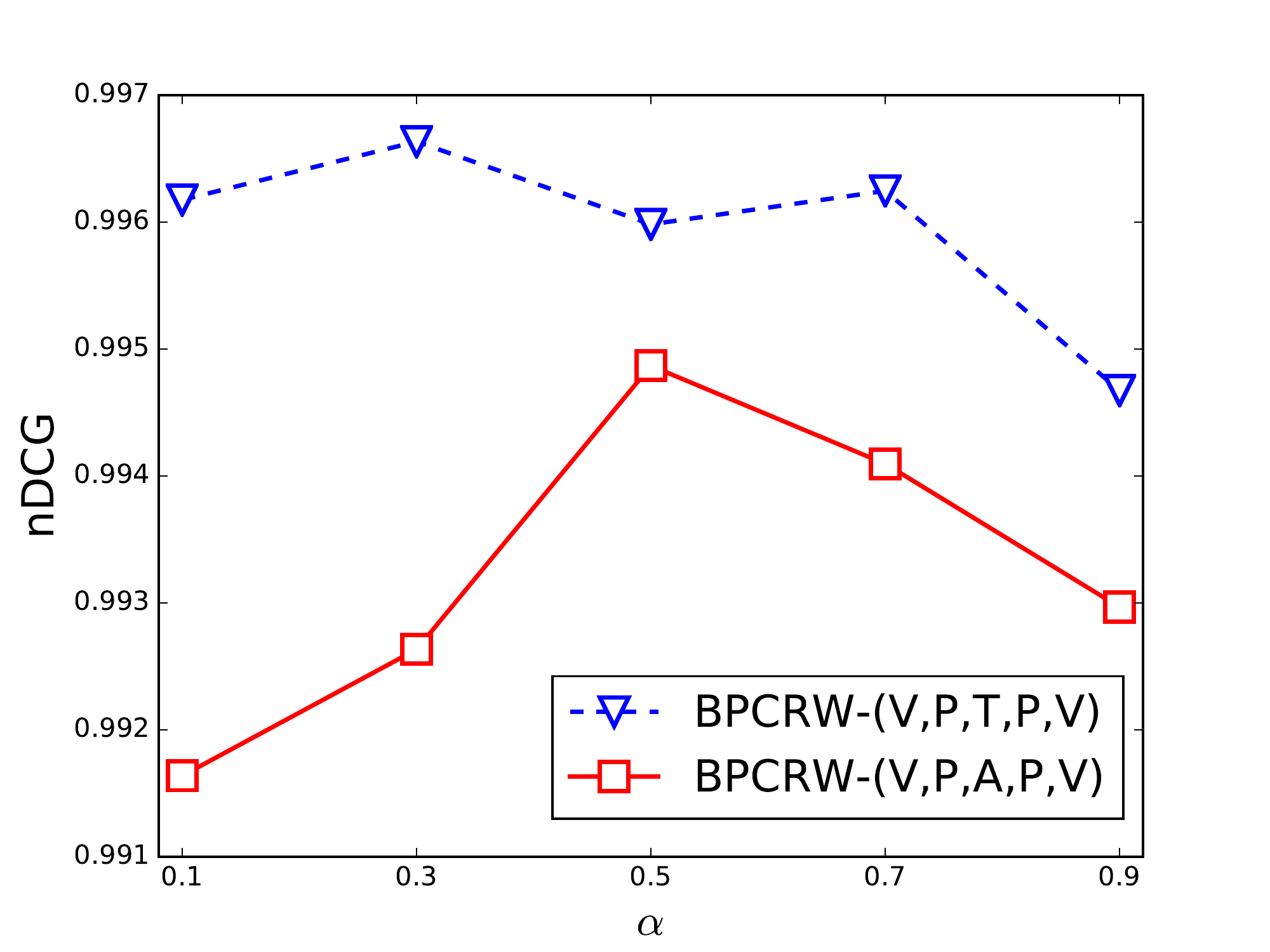}
  \caption{Sensitivity of BPCRW to different meta-paths on DBLPr in terms of ranking with TKDE as the source object.}
  \label{insensitivity_bpcrw_ndcg_tkde}
\end{figure}

Fig. \ref{insensitivity_bpcrw_ndcg_kdd}, Fig. \ref{insensitivity_bpcrw_ndcg_sigmod} and Fig. \ref{insensitivity_bpcrw_ndcg_tkde} respectively
shows the $nDCG$ values of BPCRW under different meta-paths when respectively selecting CIKM, SIGMOD and TKDE as the source objects.
According to these figures, we know that 1) the $nDCG$ values for the meta-path $(V,P,T,P,V)$ are a little larger than
that for the meta-path $(V,P,A,P,V)$ with CIKM and TKDE be the source objects;
2) the $nDCG$ values for the meta-path $(V,P,A,P,V)$ are a little larger than
that for the meta-path $(V,P,T,P,V)$ with SIGMOD be the source objects.
That is to say, different meta-paths for BPCRW yield different $nDCG$ values even though their gap is small.

In conclusion, both $Pathim$ and BPCRW are a little sensitive to different meta-paths in terms of ranking task.

\subsection{Comparison in terms of Clustering Quality}\label{subsec:clustering}

Now, we compare RMSS with local and global weighting strategies against the baselines in terms of clustering quality ($NMI$ \cite{SHYYW:2011}, the bigger, the better) on DBLPc and BioIN.
First, we compute the similarities between two objects respectively using these metrics.
That means we obtain a feature vector for each object. Then, we employ $k$-means algorithm to cluster these feature vectors (i.e. the objects).
For DBLPc, $Venue$ is selected as the source and target object type. Its benchmark is given according to the field of the venues.
For BioIN, $Gene$ is selected to the source and target object type. Its benchmark is extracted from the one used in paper \cite{benchmark}.
$k$ is set to the number of clusters in the benchmark.

\subsubsection{On BioIN}\label{subsubsec:clustering_slap}

//
Here, we compare the $NMI$ values yielded by RMSS with local and global weighting strategies against those yielded by the baselines under different decaying parameters $\lambda$ on BioIN.
For BSCSE and BPCRW, all the $NMI$ values under different biased parameters $\alpha$ are considered here.

\begin{figure}[htb]
  \centering
  \includegraphics[width=0.6\textwidth]{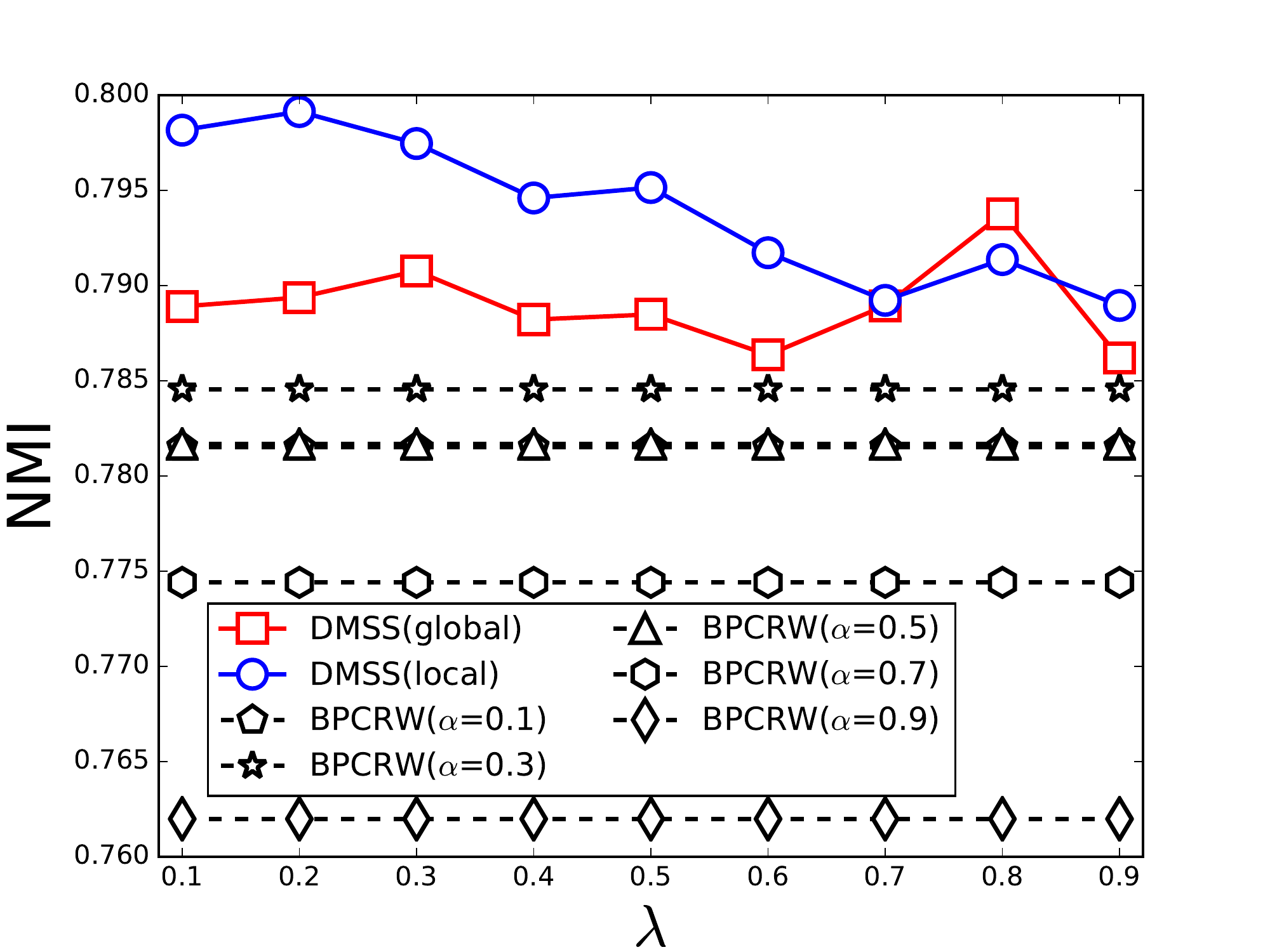}
  \caption{Comparison of $NMI$ for RMSS under different $\lambda$ with $NMI$ for BPCRW under different $\alpha$ on BioIN.}
  \label{slap_cluster_lambda_bpcrw}
\end{figure}

\begin{figure}[htb]
  \centering
  \includegraphics[width=0.6\textwidth]{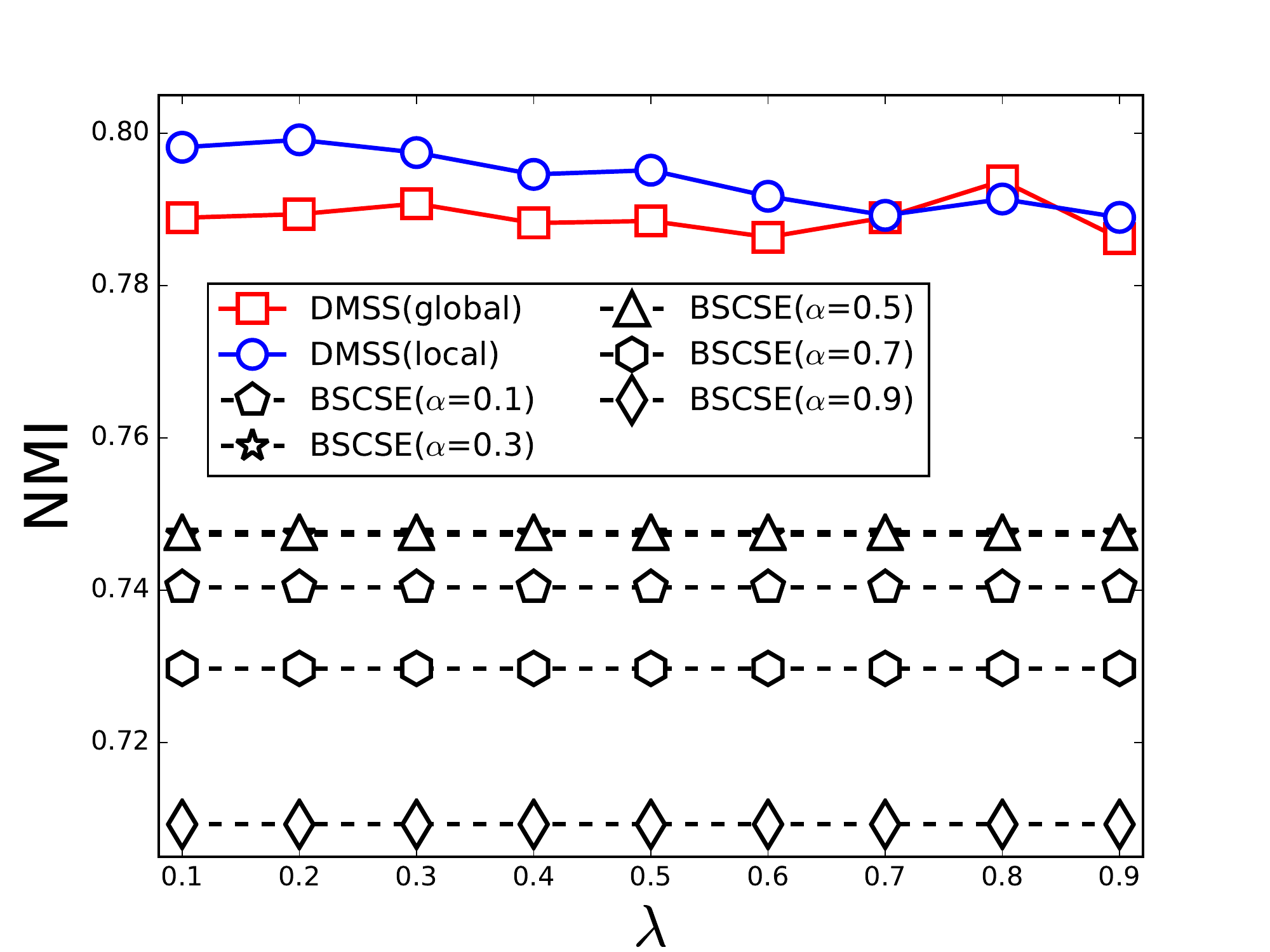}
  \caption{Comparison of $NMI$ for RMSS under different $\lambda$ with optimal $NMI$ for BSCSE under different $\alpha$ on BioIN.}
  \label{slap_cluster_lambda_bscse}
\end{figure}

\begin{figure}[htb]
  \centering
  \includegraphics[width=0.6\textwidth]{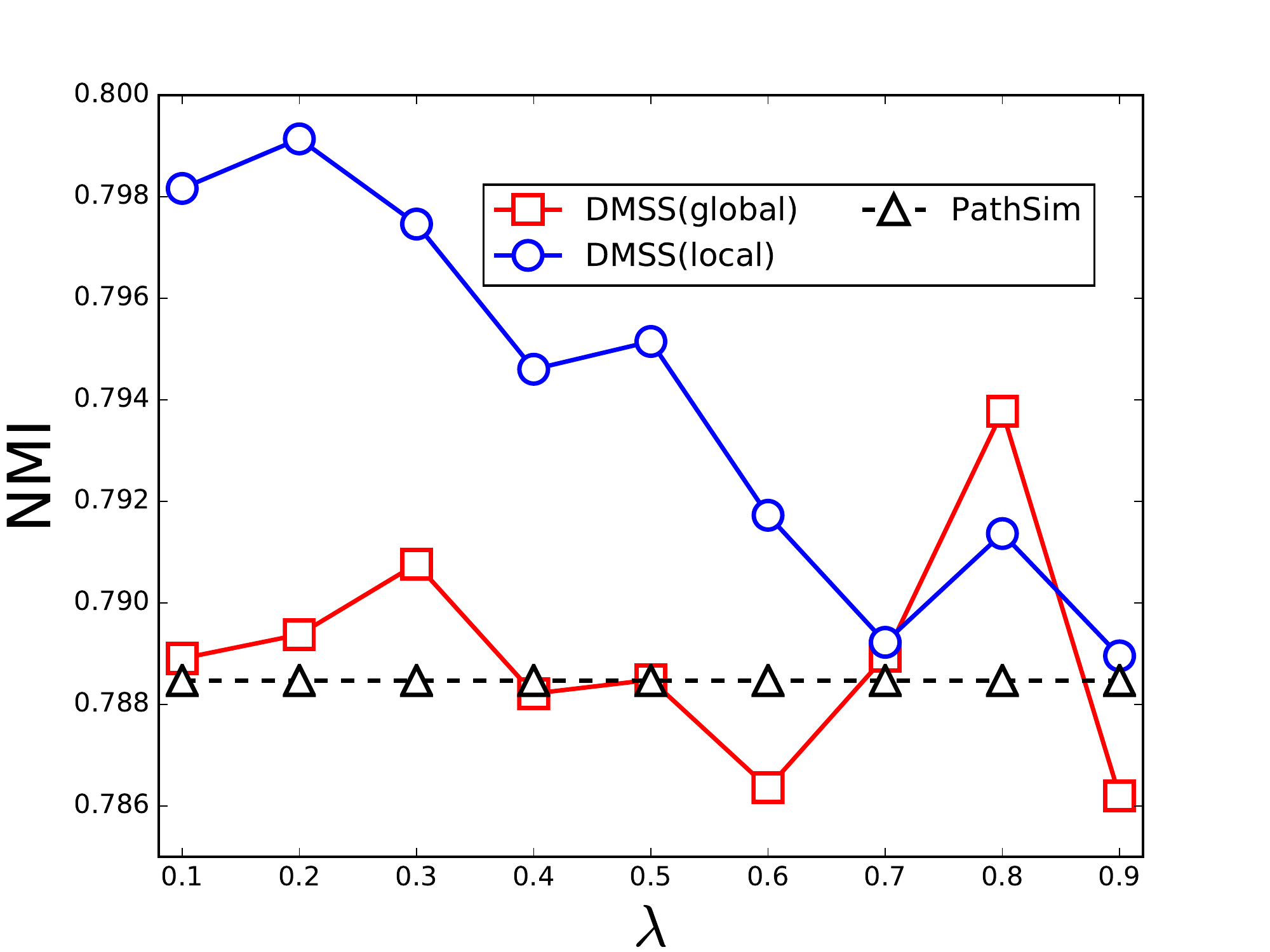}
  \caption{Comparison of $NMI$ for RMSS under different $\lambda$ with optimal $NMI$ for PathSim under different $\alpha$ on BioIN.}
  \label{slap_cluster_lambda_pathsim}
\end{figure}

Fig. \ref{slap_cluster_lambda_bpcrw}, Fig. \ref{slap_cluster_lambda_bscse} and Fig. \ref{slap_cluster_lambda_pathsim} present
the comparisons of $NMI$ respectively yielded by BPCRW, BSCSE and PathSim against that yielded by RMSS with local and global weighting strategies.
According to these figures, 1) the $NMI$ values yielded by RMSS with local weighting strategy are always larger than those yielded by the baselines (with different $\alpha$);
2) the $NMI$ values yielded by RMSS with global weighting strategy are larger than those yielded by BPCRW and BSCSE. For PathSim, its $NMI$ values may be larger than that yielded
by RMSS with global weighting strategy when $\lambda=0.4,0.6,0.9$. However, the difference is small.

In conclusion, RMSS with local weighting strategy outperforms the baselines, and RMSS with global weighting strategy is comparable to the baselines.

\subsubsection{On DBLPc}\label{subsubsec:clustering_dblp}

//
Here, we compare the $NMI$ values yielded by RMSS with local and global weighting strategies against those yielded by the baselines under different decaying parameters $\lambda$ on DBLPc.
For BSCSE and BPCRW, all the $NMI$ values under different biased parameters $\alpha$ are considered here.

\begin{figure}[htb]
  \centering
  \includegraphics[width=0.6\textwidth]{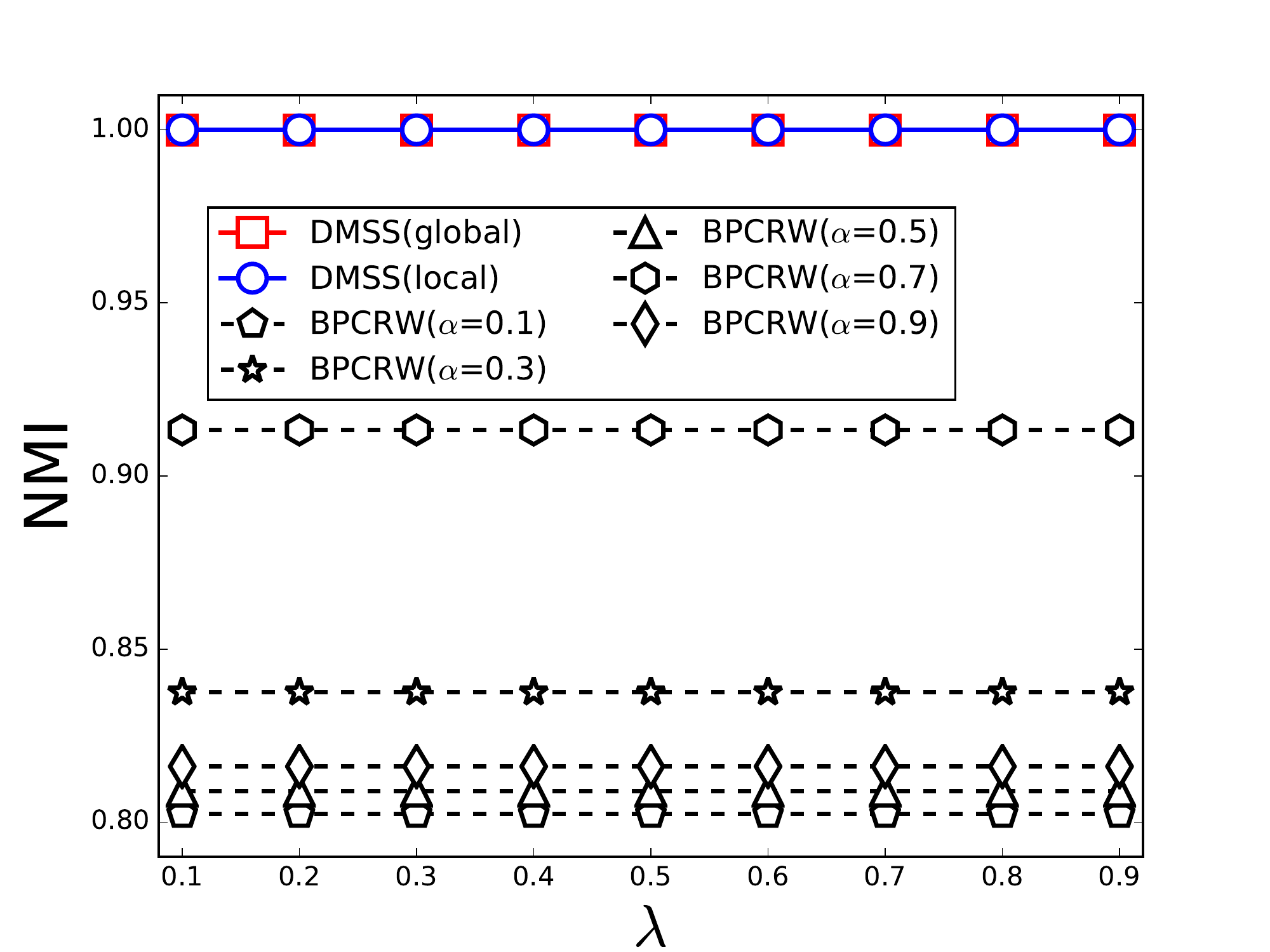}
  \caption{Comparison of $NMI$ for RMSS under different $\lambda$ with $NMI$ for BPCRW under different $\alpha$ on DBLPc.}
  \label{dblp_cluster_lambda_bpcrw}
\end{figure}

\begin{figure}[htb]
  \centering
  \includegraphics[width=0.6\textwidth]{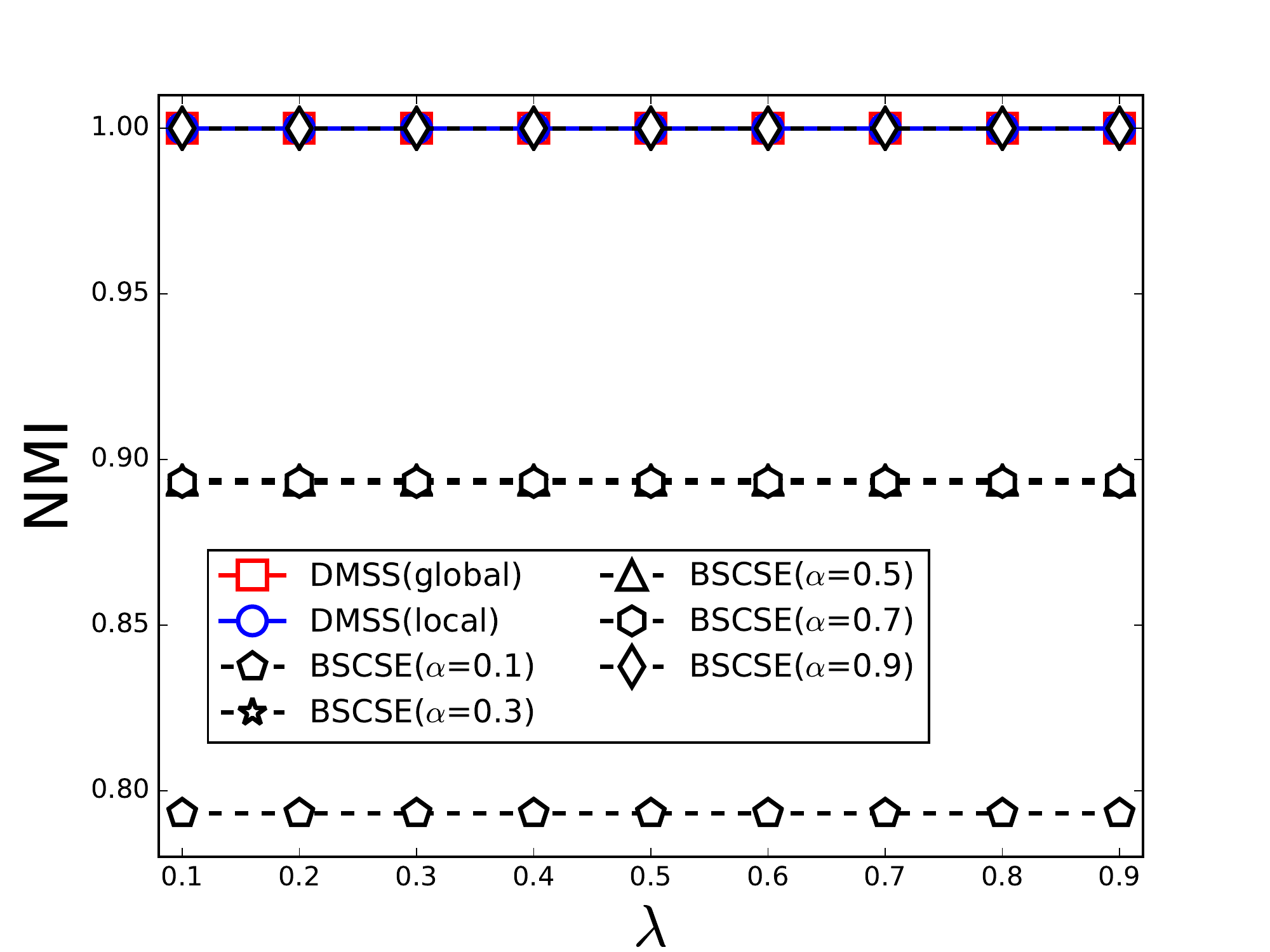}
  \caption{Comparison of $NMI$ for RMSS under different $\lambda$ with optimal $NMI$ for BSCSE under different $\alpha$ on DBLPc.}
  \label{dblp_cluster_lambda_bscse}
\end{figure}

\begin{figure}[htb]
  \centering
  \includegraphics[width=0.6\textwidth]{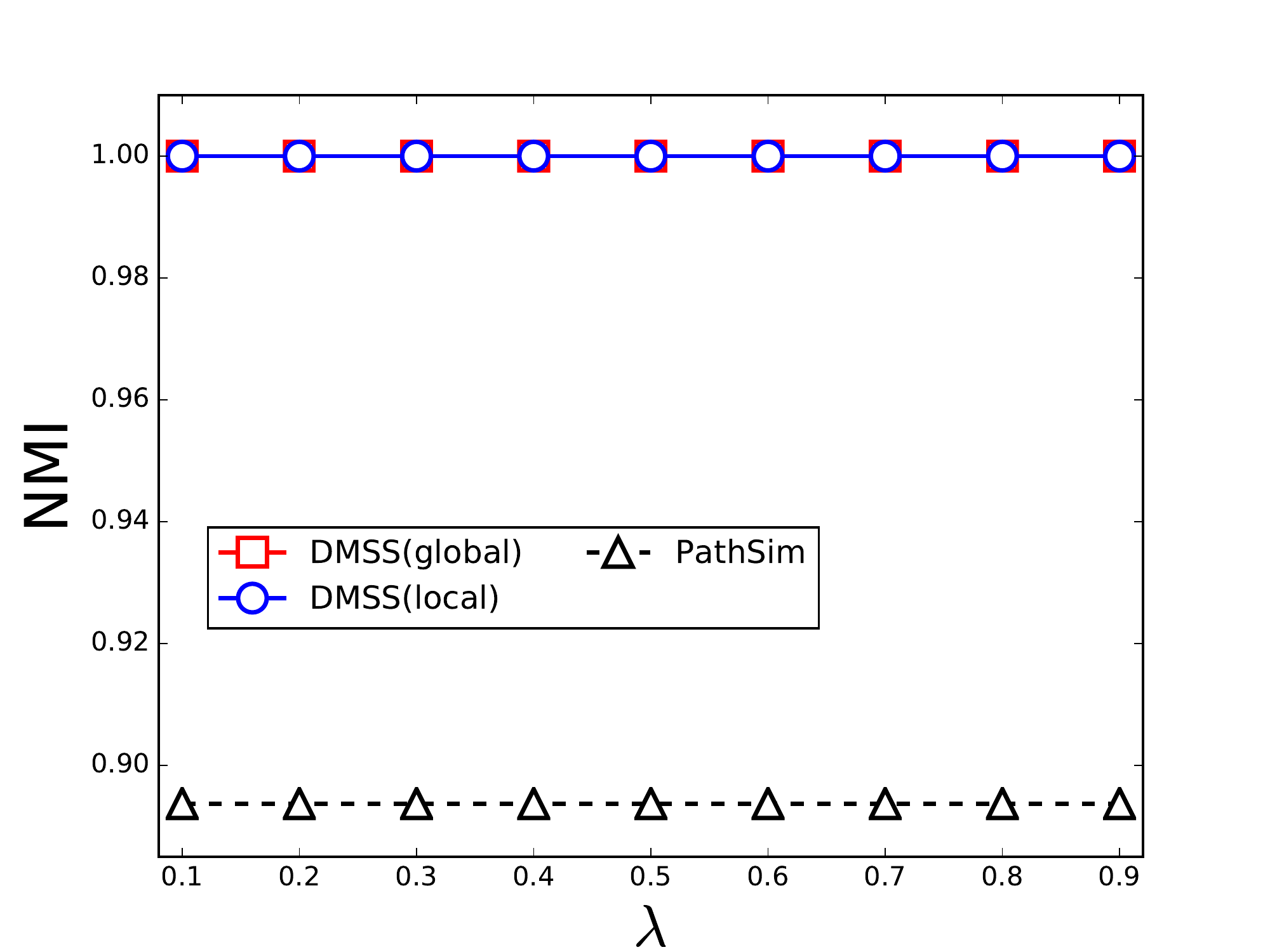}
  \caption{Comparison of $NMI$ for RMSS under different $\lambda$ with optimal $NMI$ for PathSim under different $\alpha$ on DBLPc.}
  \label{dblp_cluster_lambda_pathsim}
\end{figure}

Fig. \ref{dblp_cluster_lambda_bpcrw}, Fig. \ref{dblp_cluster_lambda_bscse} and Fig. \ref{dblp_cluster_lambda_pathsim} present
the comparisons of $NMI$ respectively yielded by BPCRW, BSCSE and PathSim against that yielded by RMSS with local and global weighting strategies.
As shown in Fig. \ref{dblp_cluster_lambda_bpcrw}, the $NMI$ values (no matter what $\lambda$ takes) yielded by RMSS with local and global strategies are always larger than those yielded by BPCRW;
As shown in Fig. \ref{dblp_cluster_lambda_bscse}, the $NMI$ values for BSCSE only under $\alpha=0.9$ is equal 1.0. And the $NMI$ yielded by RMSS with local and global strategies are always equal to 1.0.
This is larger than those yielded by BSCSE with different $\alpha$ except $\alpha=0.9$;
As shown in Fig. \ref{dblp_cluster_lambda_pathsim}, the $NMI$ values (1.0 no matter what $\lambda$ takes) yielded by RMSS with local and global strategies are always larger than thos yielded by PathSim.

In conclusion, RMSS with local and global strategies outperforms the baselines.

\subsection{Ranking Quality of RMSS}\label{subsec:ranking}

Now, we compare RMSS with local and global weighting strategies against the baselines in terms of ranking quality ($nDCG$ \cite{SHYYW:2011}, the higher, the better) on DBLPr.
First, we select three venues `CIKM', `SIGMOD' and `TKDE' as the source objects.
All the venues can be ranked as 0 (unrelated), 1 (slightly related), 2 (fairly related), 3 (highly related)
according to their similarities to the source object.
Then, we employ RMSS and the other baselines to compute the similarities between the source objects and the other venues.
As a result, we obtain $nDCG$ values respectively for the source objects.

\subsubsection{Ranking Quality}\label{subsubsec:rankingquality}

//
Now, we compare the $nDCG$ values yielded by RMSS with local and global weighting strategies and the baselines when selecting three source venues CIKM, SIGMOD and TKDE.
For each baseline, we respectively compare its $nDCG$ with different biased parameters $\alpha$ against that yielded by RMSS with local and global weighting strategy under different decaying parameters $\lambda$.

\begin{figure}[htb]
  \centering
  \includegraphics[width=0.6\textwidth]{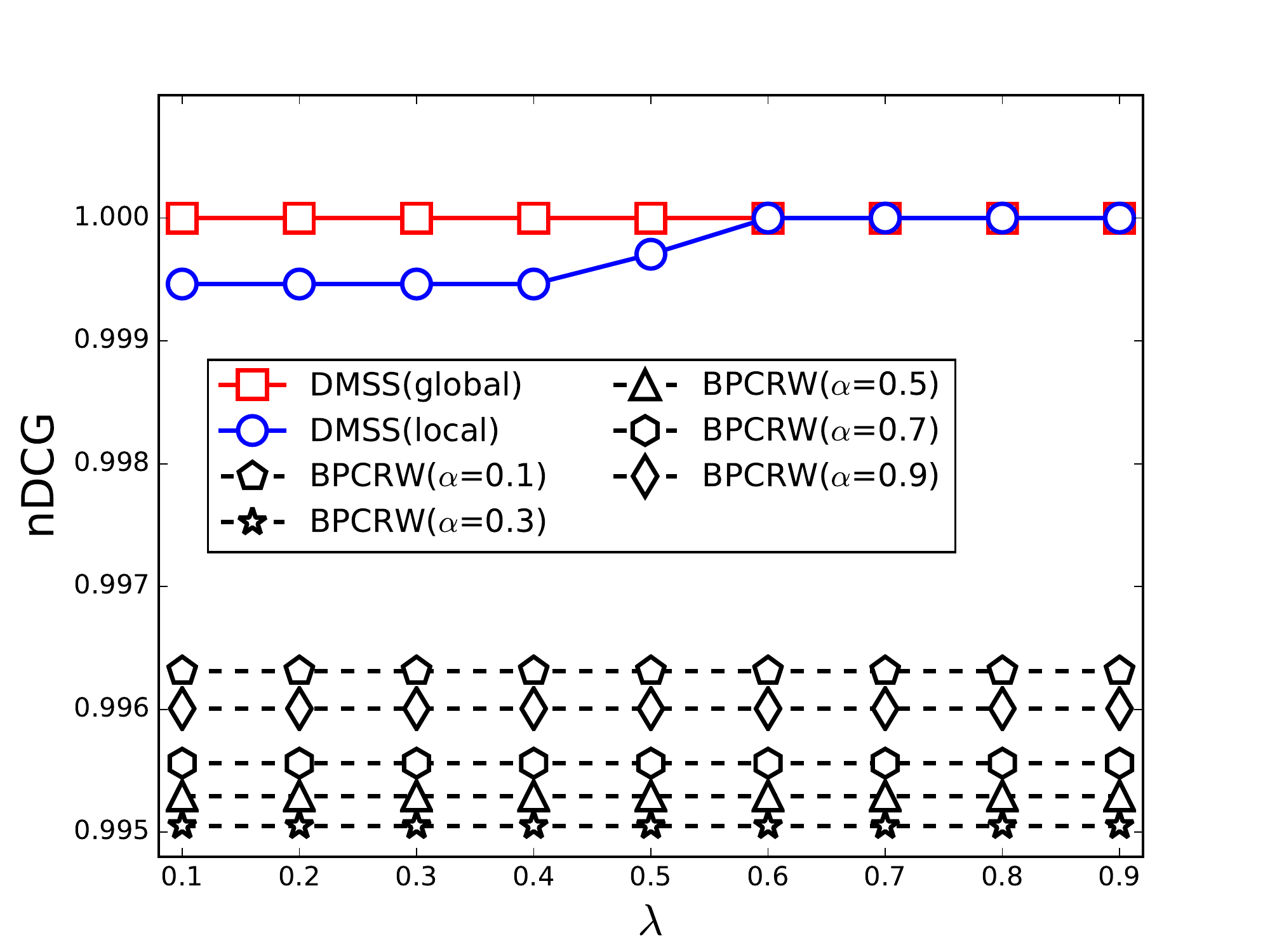}
  \caption{The comparison of BPCRW with different $\alpha$ against RMSS under the different decaying parameters $\lambda$. CIKM is the source object.}
  \label{bpcrw_cikm}
\end{figure}

\begin{figure}[htb]
  \centering
  \includegraphics[width=0.6\textwidth]{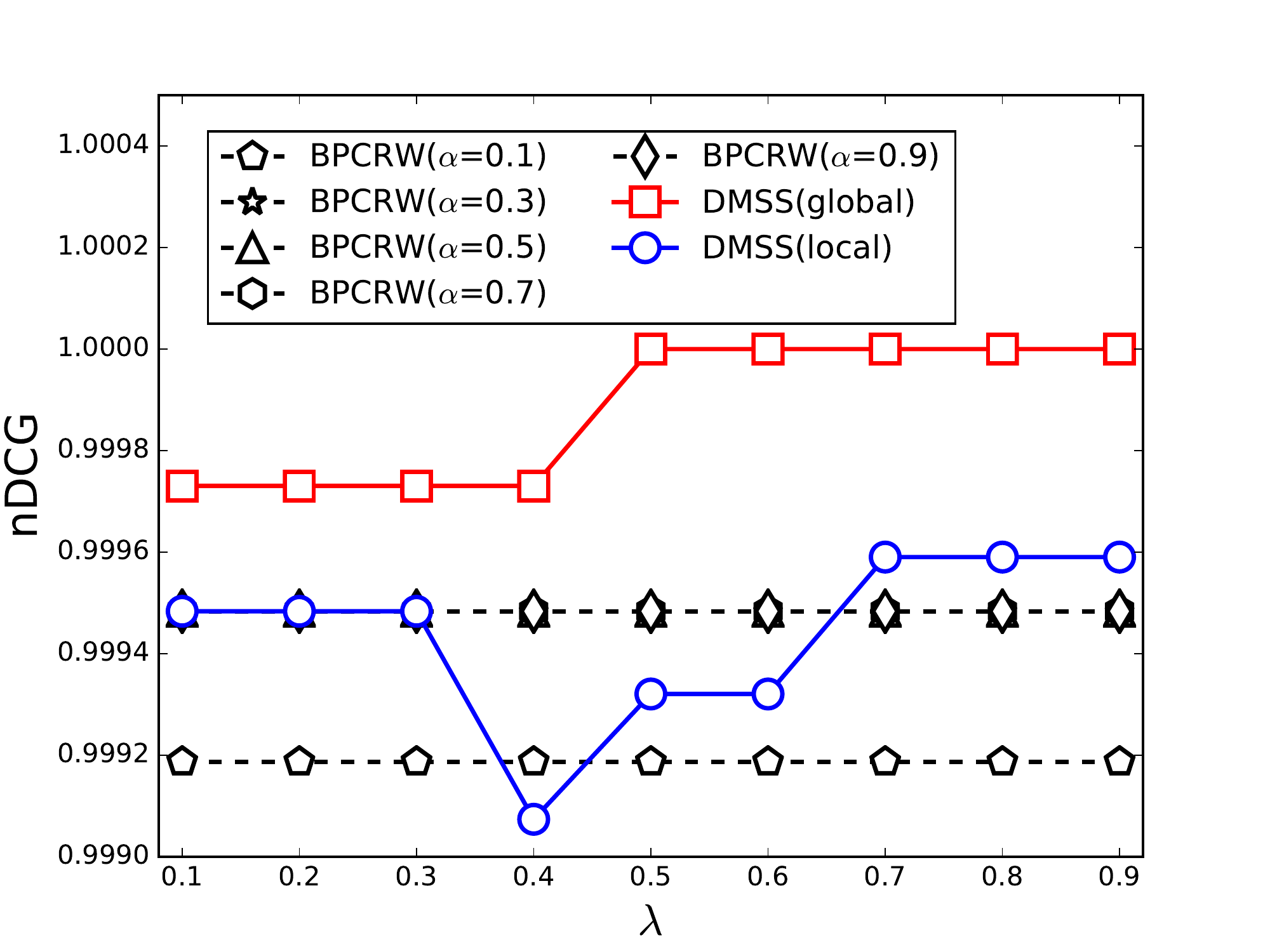}
  \caption{The comparison of BPCRW with different $\alpha$ against RMSS under the different decaying parameters $\lambda$. SIGMOD is the source object.}
  \label{bpcrw_sigmod}
\end{figure}

\begin{figure}[htb]
  \centering
  \includegraphics[width=0.6\textwidth]{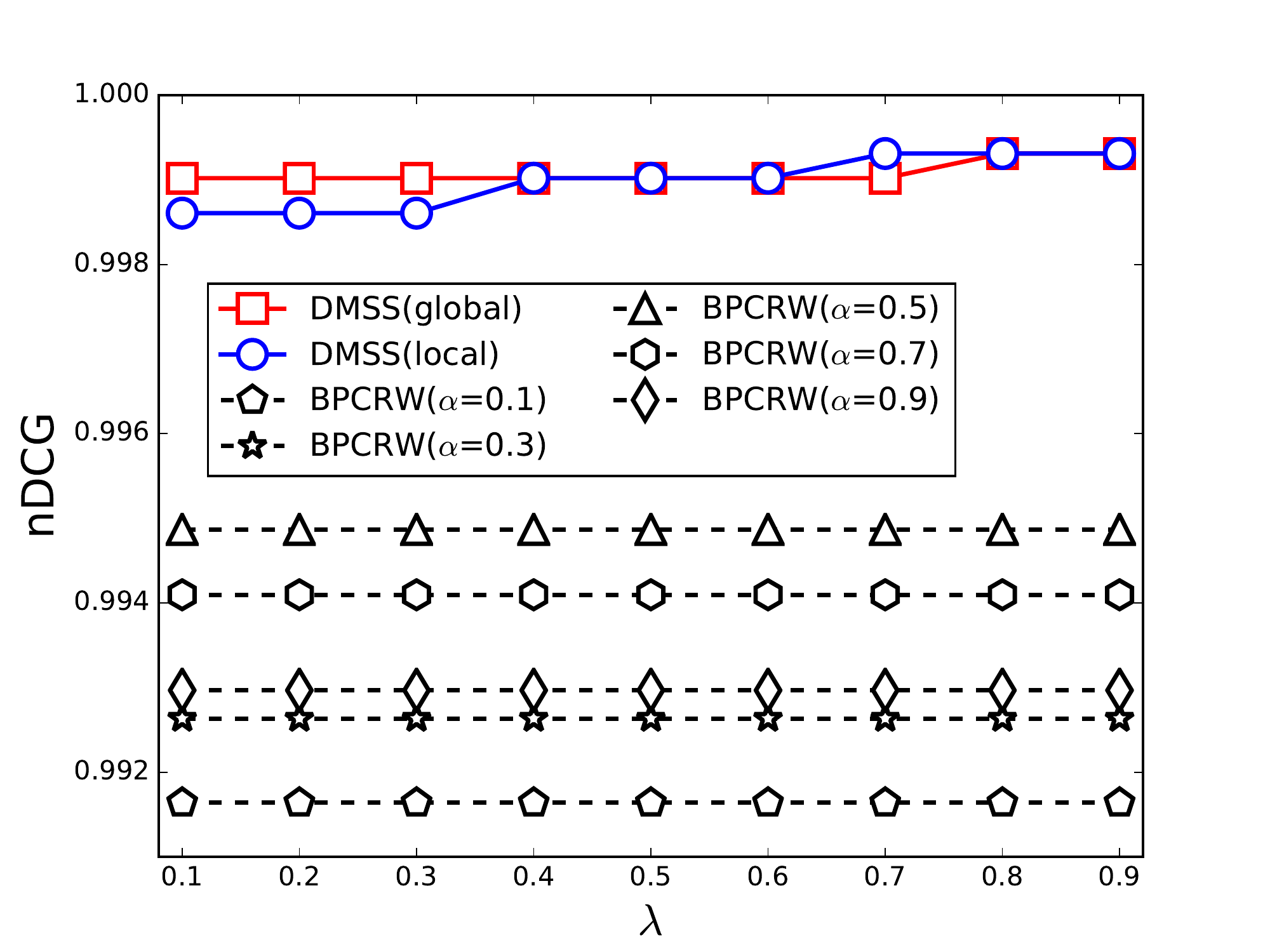}
  \caption{The comparison of BPCRW with different $\alpha$ against RMSS under the different decaying parameters $\lambda$. TKDE is the source object.}
  \label{bpcrw_tkde}
\end{figure}

\begin{figure}[htb]
  \centering
  \includegraphics[width=0.6\textwidth]{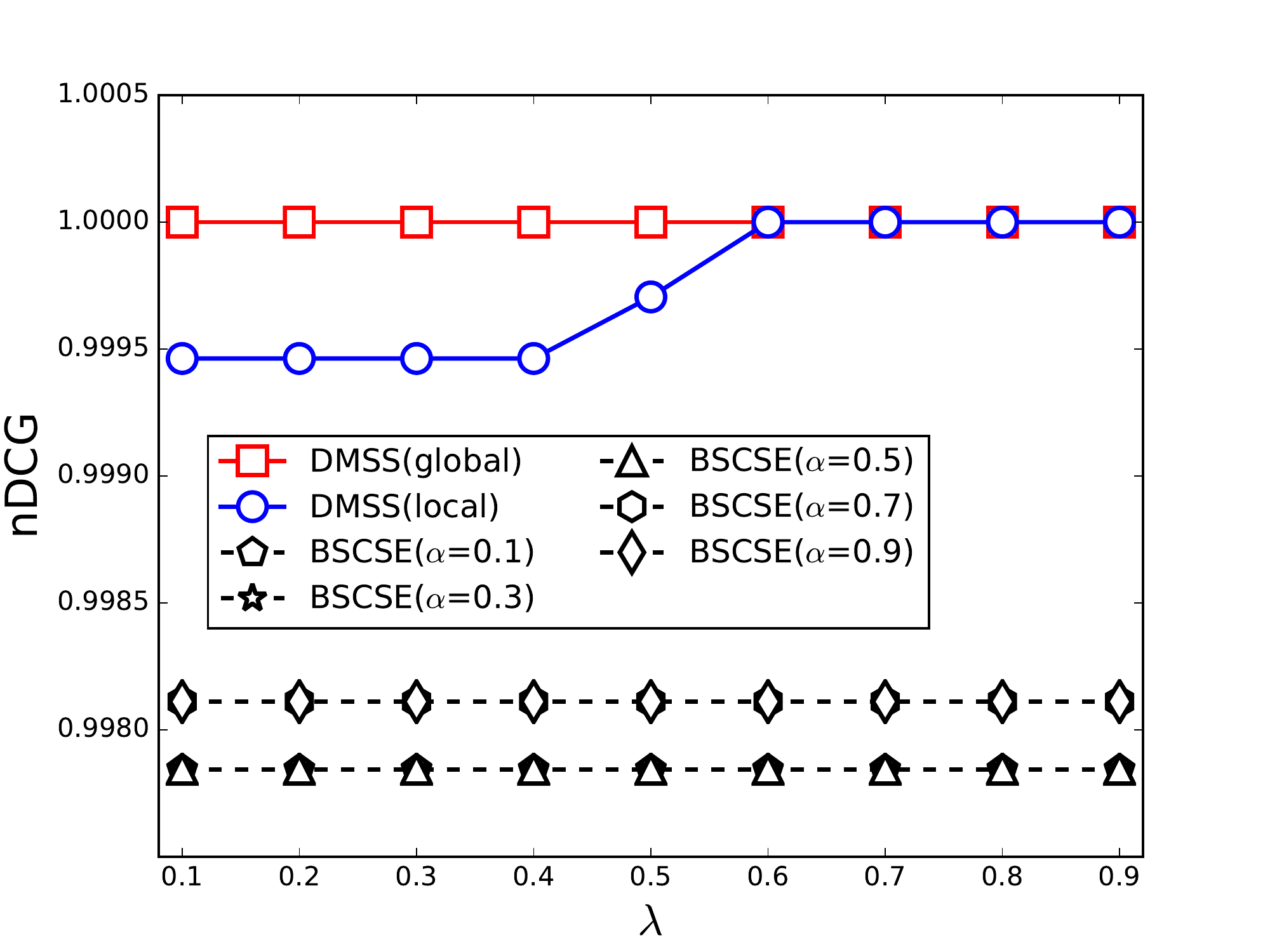}
  \caption{The comparison of BSCSE with different $\alpha$ against RMSS under the different decaying parameters $\lambda$. CIKM is the source object.}
  \label{bscse_cikm}
\end{figure}

\begin{figure}[htb]
  \centering
  \includegraphics[width=0.6\textwidth]{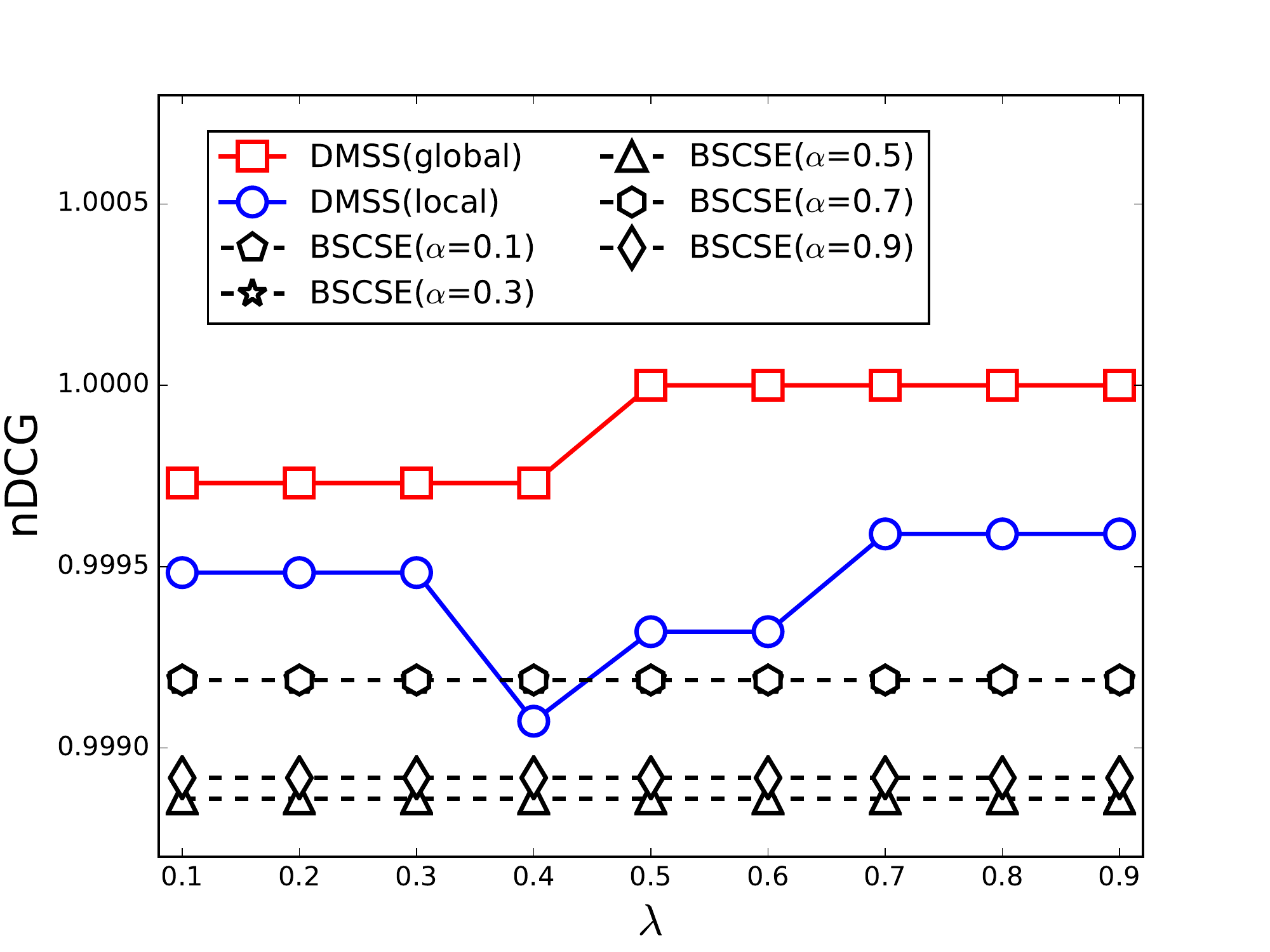}
  \caption{The comparison of BSCSE with different $\alpha$ against RMSS under the different decaying parameters $\lambda$. SIGMOD is the source object.}
  \label{bscse_sigmod}
\end{figure}

\begin{figure}[htb]
  \centering
  \includegraphics[width=0.6\textwidth]{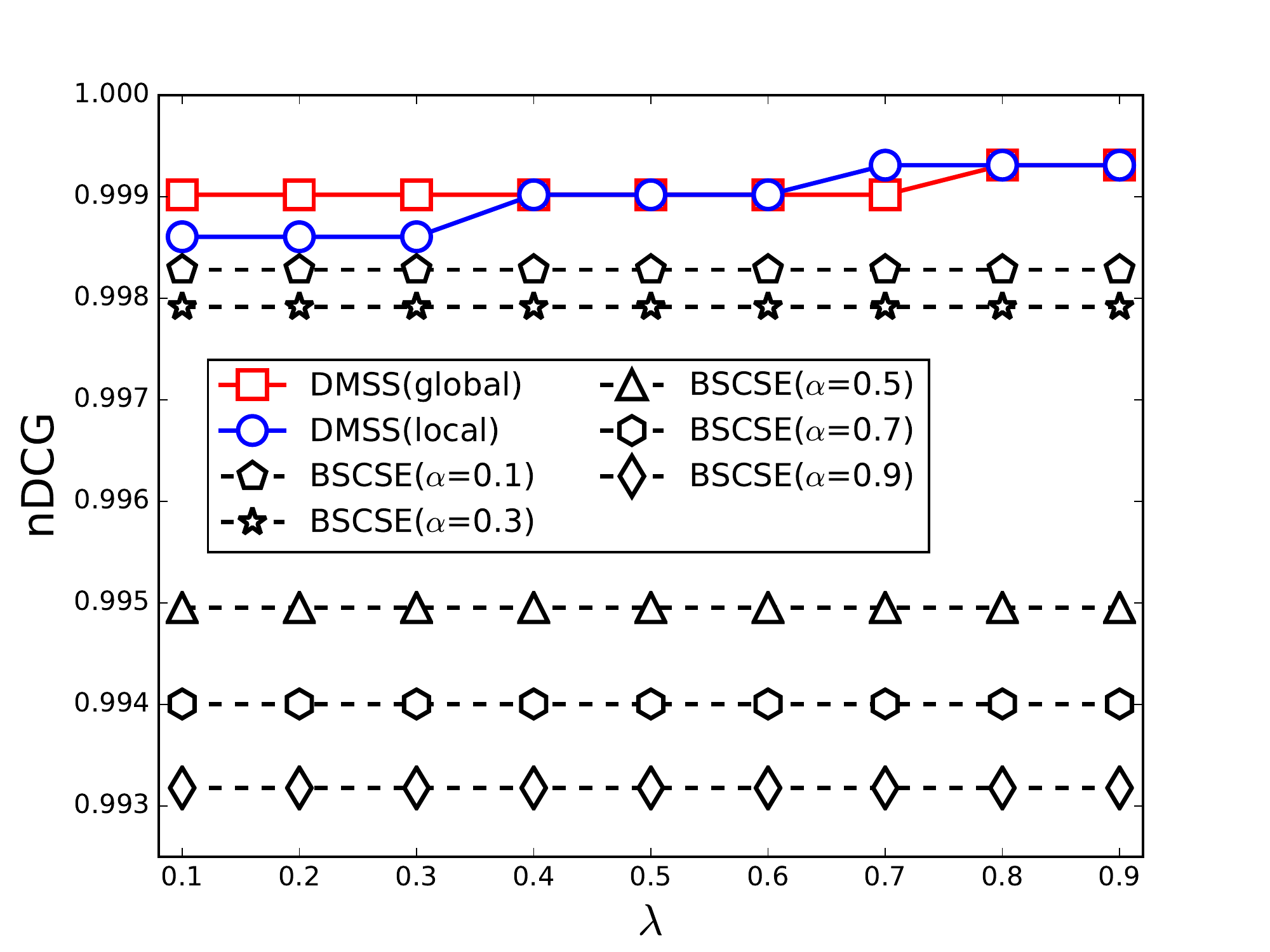}
  \caption{The comparison of BSCSE with different $\alpha$ against RMSS under the different decaying parameters $\lambda$. TKDE is the source object.}
  \label{bscse_tkde}
\end{figure}

\begin{figure}[htb]
  \centering
  \includegraphics[width=0.6\textwidth]{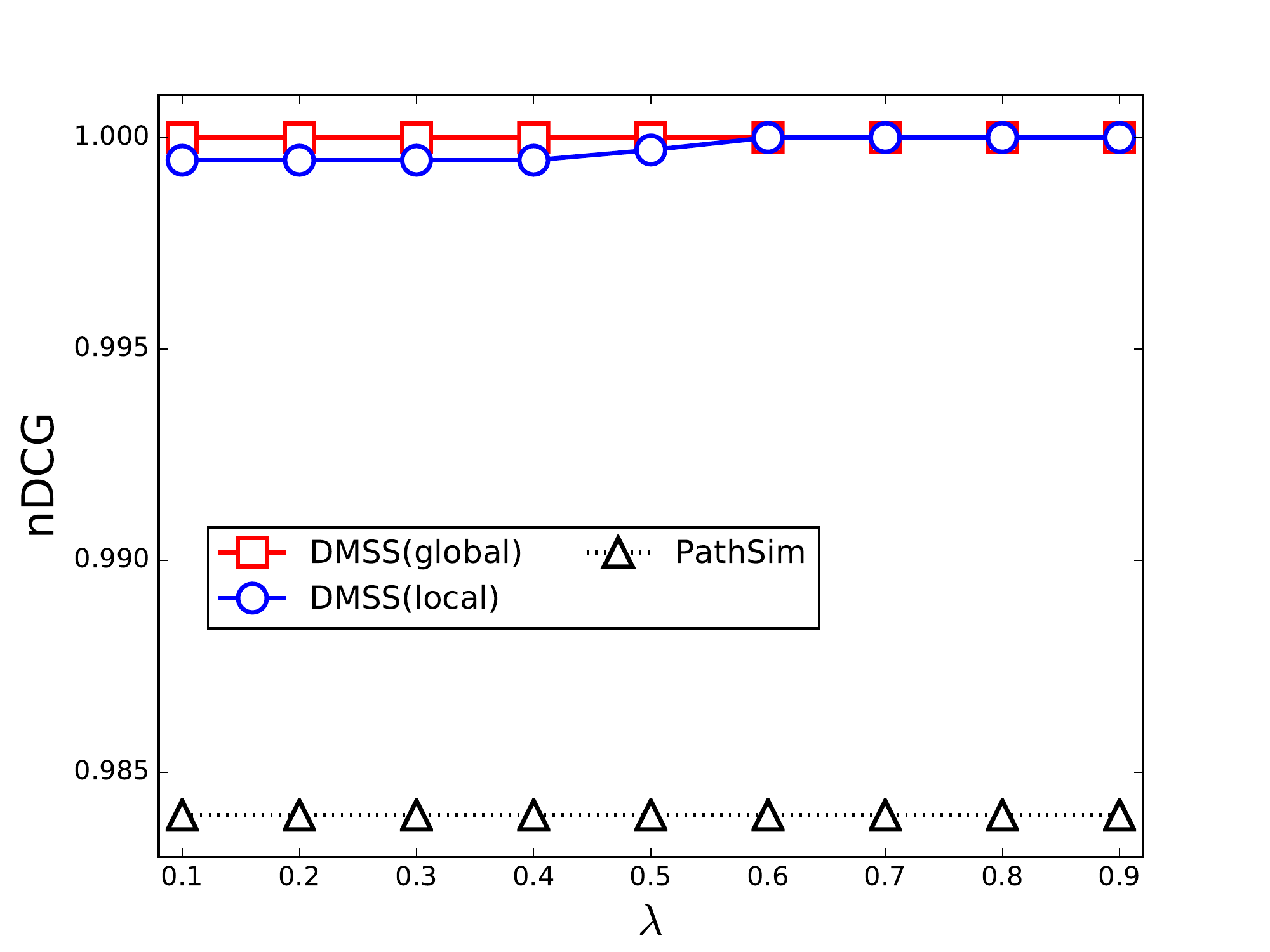}
  \caption{The comparison of PathSim with different $\alpha$ against RMSS under the different decaying parameters $\lambda$. CIKM is the source object.}
  \label{pathsim_cikm}
\end{figure}

\begin{figure}[htb]
  \centering
  \includegraphics[width=0.6\textwidth]{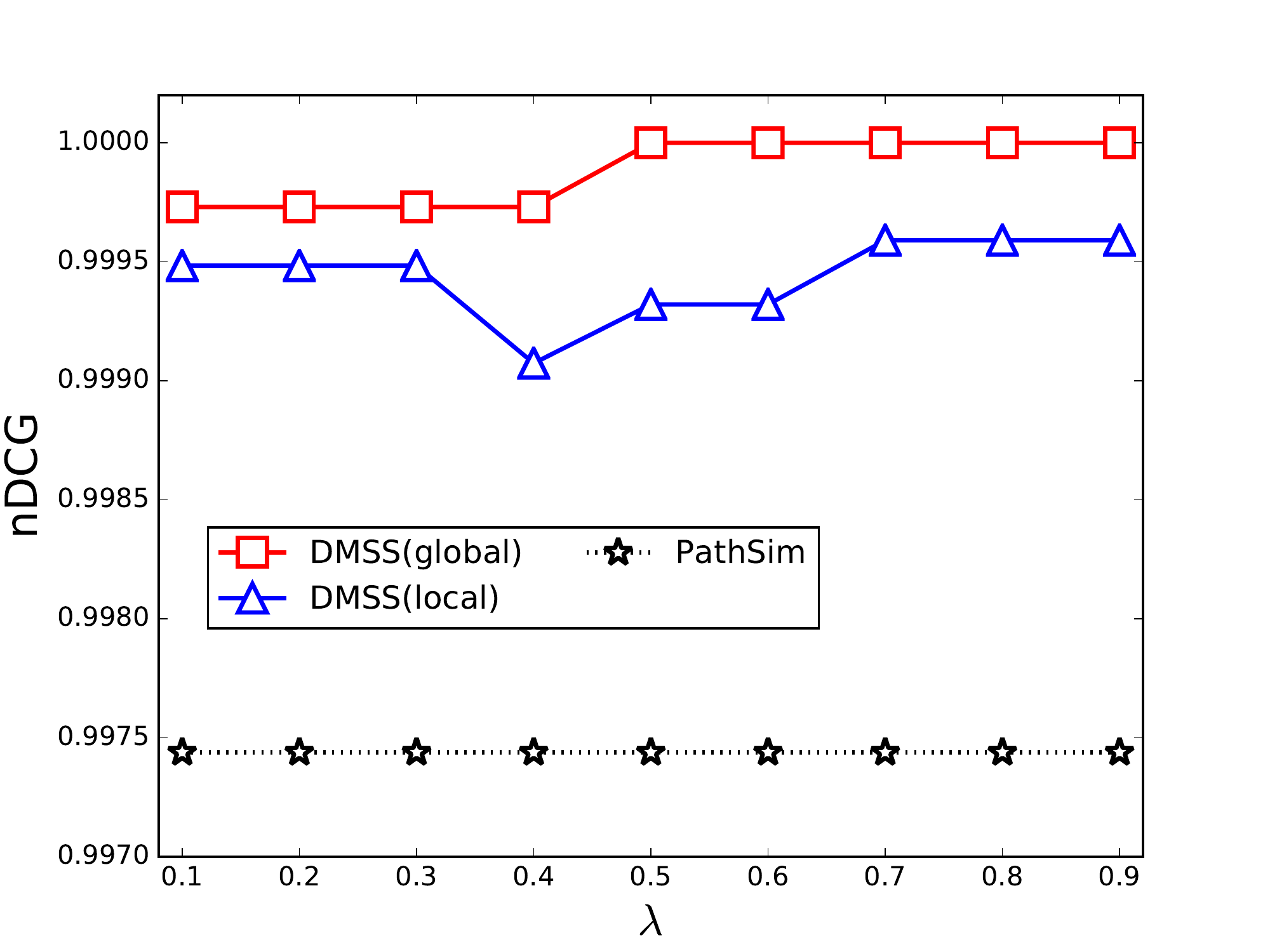}
  \caption{The comparison of PathSim with different $\alpha$ against RMSS under the different decaying parameters $\lambda$. SIGMOD is the source object.}
  \label{pathsim_sigmod}
\end{figure}

\begin{figure}[htb]
  \centering
  \includegraphics[width=0.6\textwidth]{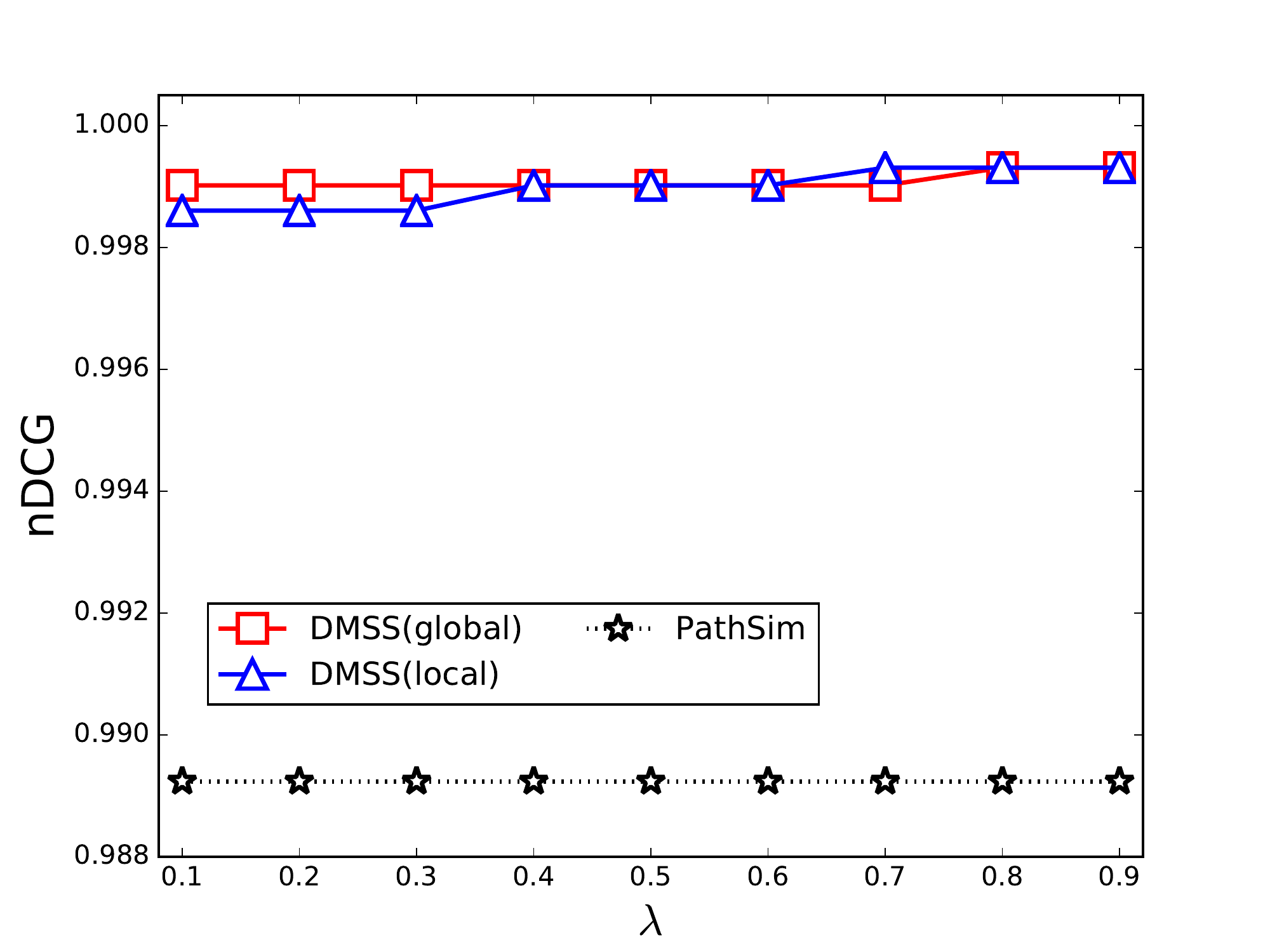}
  \caption{The comparison of PathSim with different $\alpha$ against RMSS under the different decaying parameters $\lambda$. TKDE is the source object.}
  \label{pathsim_tkde}
\end{figure}

For BPCRW, Fig. \ref{bpcrw_cikm}, Fig. \ref{bpcrw_sigmod} and Fig. \ref{bpcrw_tkde} illustrate the comparisons of BPCRW with different $\alpha$ against RMSS under different $\lambda$.
According to these figures, we know that 1) RMSS with global weighting strategy performs slightly better than BPCRW with different $\alpha$;
2) RMSS with local weighting strategy slightly outperforms BPCRW with different $\alpha$ when selecting CIKM and TKDE as the source objects.
However, the $nDCG$ yielded by RMSS with local weighting strategy is a little smaller than that yielded by BPCRW when the source object is SIGMOD and $\lambda=0.4,0.5,0.6$.

For BSCSE, Fig. \ref{bscse_cikm}, Fig. \ref{bscse_sigmod} and Fig. \ref{bscse_tkde} illustrate the comparisons of BSCSE with different $\alpha$ against RMSS under different $\lambda$.
According to these figures, we know that 1) RMSS with global weighting strategy slightly performs better than BSCSE with different $\alpha$;
2) RMSS with local weighting strategy slightly outperforms BSCSE with different $\alpha$ when selecting CIKM and TKDE as the source objects.
However, the $nDCG$ yielded by RMSS with local weighting strategy is a little smaller than that yielded by $BSCSE$ when the source object is SIGMOD and $\lambda=0.4$.

For PathSim, Fig. \ref{pathsim_cikm}, Fig. \ref{pathsim_sigmod} and Fig. \ref{pathsim_tkde} illustrate the comparisons of PathSim with different $\alpha$ against RMSS under different $\lambda$.
According to these figures, we know that RMSS with local and global weighting strategies slightly performs better than PathSim.

In summary, RMSS on the whole performs better than the baselines in terms of the ranking task.

\subsection{Time Efficiency}\label{subsec:timeefficiency}

Here, we evaluate the time efficiency of RMSS, BPCRW, PathSim and BSCSE on BioIN, DBLPc and DBLPr.
Table \ref{time} shows the running time of computing the similarities among all objects using these metrics.
According to this table, we know that BPCRW performs much better than BSCSE and PathSim in terms of time efficiency.
In addition, we discover that RMSS spends most of its time on the offline part, and spends a little time (even lower than BPCRW) on the online part.
In practice, we can use distributed computing environment equipped with GPU to periodically compute the offline part, and then use the online part to provide real-time service for users.

Furthermore, we discover that the running time of BSCSE on DBLPc and DBLPr is much larger than that on BioIN.
This is because the nature of DBLPc and DBLPr is different from that of BioIN.
On DBLPc or DBLPr, we only consider the meta-structure $(V,P,(A,T),P,A)$.
Each venue accepts a lot of papers. This causes too many author-term pairs. The algorithm for computing BSCSE spends too much time on traversing these pairs.
On BioIN, we consider two meta-structures $(G,CC,(Si,Sub),CC,G)$ and $(G,(GO,Ti),G)$.
For $(G,(GO,Ti),G)$, each gene links to a small number of pairs whose entries respectively belong to $GO$ and $Ti$.
For $(G,CC,(Si,Sub),CC,G)$, each gene links to a small number of chemical compounds and these chemical compounds link to a small number of pairs whose entries respectively belong to $Si$ and $Sub$.
Therefore, The algorithm for computing BSCSE spends a little time on traversing these pairs.

\begin{table}\footnotesize
  \centering
  \caption{Average Running Time (Sec) of computing similarity matrices using RMSS and the baselines}
  \begin{tabular}{|c|c|c|c|c|}
    \hline
     \multicolumn{2}{|c|}{\diagbox{Metric}{Dataset}} & DBLPc & DBLPr & BioIN \\
    \hline
      \multicolumn{2}{|c|}{BPCRW} & $10.407$ & $9.709$ & $5.320$ \\
    \hline
      \multicolumn{2}{|c|}{PathSim} & 229.642 & 162.693 & 630.083 \\
    \hline
      \multicolumn{2}{|c|}{BSCSE} & 10305.906 & 9896.015 & 9.913 \\
    \hline
      \multirow{2}{*}{RMSS (local)} & Offline & 15067.527 & 13480.6.9 & 652.957 \\
    \cline{2-5}
                                      & Online & $\mathbf{0.0077}$ & $\mathbf{0.0079}$ & $\mathbf{3.0484}$ \\
    \hline
      \multirow{2}{*}{RMSS (global)} & Offline & 14547.407 & 12517.916 & 652.733 \\
    \cline{2-5}
                                       & Online & $\mathbf{0.0237}$ & $\mathbf{0.0083}$ & $\mathbf{3.1769}$ \\
    \hline
  \end{tabular}
  \label{time}
\end{table}

\section{Conclusion}\label{sec:conclusion}

In this paper, we propose RMSS, a recurrent meta-structure based similarity metric in HINs.
The recurrent meta-structure can be constructed automatically.
To extract semantics encapsulated in the RecurMS, we first decompose it into several recurrent meta-paths and recurrent meta-trees, and
then combine the commuting matrices of the recurrent meta-paths and recurrent meta-trees according to different weights.
It is noteworthy that we propose two kinds of weighting strategies to determine the weights of different schematic structures.
As a result, RMSS is defined by the combination of these commuting matrices.
Experimental evaluations show that (1) PathSim, BPCRW and BSCSE are sensitive to meta-paths or meta-structures;
(2) RMSS with local and global weighting strategies outperforms the baselines in terms of clustering and ranking.
In conclusion, the proposed RMSS is insensitive to different schematic structures, and outperforms the state-of-the-art metrics in terms of clustering and ranking tasks.
That is to say, using RMSS can robustly and exactly evaluate the similarities between objects.

\section{Acknowledgement}\label{sec:ackknowledgement}

The work is supported in part by the National Science Foundation of China (Grant Numbers: 61472299 and 61672417).
Any opinions, findings and conclusions expressed here are those of the authors and do not necessarily reflect the views of the funding agencies.

\bibliographystyle{ACM-Reference-Format}
\bibliography{sample-bibliography}

\end{document}